\documentclass{article}
\usepackage{fullpage}
\usepackage{amssymb}
\usepackage{amsmath}
\usepackage{amsthm}
\usepackage{hyperref}
\usepackage{cleveref}
\usepackage{xcolor}
\usepackage{tcolorbox}
\usepackage{authblk}
\usepackage{physics}
\usepackage{bm}
\usepackage{tikz}
\usetikzlibrary{patterns}
\usepackage{tikz-3dplot}
\usepackage{multirow}
\renewcommand{\d}{\ensuremath{\textrm{d}}}
\newcommand{\A}{\ensuremath{\mathcal{A}}}
\newcommand{\B}{\ensuremath{\mathcal{B}}}
\newcommand{\E}{\ensuremath{\mathbb{E}}}
\renewcommand{\O}{\ensuremath{\mathcal{O}}}
\renewcommand{\P}{\ensuremath{\mathbb{P}}}
\newcommand{\N}{\ensuremath{\mathbb{N}}}
\newcommand{\R}{\ensuremath{\mathbb{R}}}

\DeclareMathOperator{\polylog}{polylog}
\DeclareMathOperator{\argmin}{argmin}
\DeclareMathOperator{\diag}{diag}
\DeclareMathOperator{\Unif}{Unif}
\DeclareMathOperator{\ADV}{ADV}
\DeclareMathOperator{\sgn}{sgn}
\DeclareMathOperator{\Parity}{Parity}
\newtheorem{theorem}{Theorem}[section]
\newtheorem{lemma}[theorem]{Lemma}
\newtheorem{definition}[theorem]{Definition}
\crefname{section}{Section}{Sections}
\crefname{subsection}{Subsection}{Subsections}
\crefname{subsubsection}{Subsubsection}{Subsubsections}
\crefname{appendix}{Appendix}{Appendices}
\crefname{figure}{Figure}{Figures}
\crefname{table}{Table}{Tables}
\crefname{theorem}{Theorem}{Theorems}
\crefname{lemma}{Lemma}{Lemmas}
\crefname{definition}{Definition}{Definitions}
\title{Quantum algorithms for multivariate Monte Carlo estimation}
\author[1]{Arjan Cornelissen}
\author[2]{Sofiene Jerbi}
\affil[1]{\normalsize QuSoft -- University of Amsterdam}
\affil[2]{Institute for Theoretical Physics, University of Innsbruck}
\date{July 2021}
\begin{document}
\maketitle

\begin{abstract}
	We consider the problem of estimating the expected outcomes of Monte Carlo processes whose outputs are described by multidimensional random variables. We tightly characterize the quantum query complexity of this problem for various choices of oracle access to the Monte Carlo process and various normalization conditions on its output variables and the estimates to be made. More specifically, we design quantum algorithms that provide a quadratic speed-up in query complexity with respect to the precision of the resulting estimates and show that these complexities are asymptotically optimal in each of our scenarios. Interestingly, we find that these speed-ups with respect to the precision come at the cost of a (up to exponential) slowdown with respect to the dimension of the Monte Carlo random variables, for most of the parameter settings we consider. We also describe several applications of our algorithms, notably in the field of machine learning, and discuss the implications of our results.
\end{abstract}

\section{Introduction}

Monte Carlo methods are used extensively in various fields of science and engineering, such as statistical physics \cite{binder12}, finance \cite{glasserman13}, or machine learning \cite{andrieu03}. At the core of these methods is a Monte Carlo process, e.g., a randomized algorithm $\mathcal{A}$, whose expected outcome is to be estimated via repeated random executions.

An interesting question in complexity theory is whether, given the ability to simulate a Monte Carlo process $\mathcal{A}$ on a quantum computer, quantum algorithms can reduce the sample complexity of such Monte Carlo estimations, i.e., the number of executions of $\mathcal{A}$ needed to obtain a close estimate. From this perspective, the univariate case where $\mathcal{A}$ produces a single real-valued outcome $v(\mathcal{A}) \in [-B,B] \subset \mathbb{R}$ is well studied. An $\varepsilon$-close estimate (in additive error) of $\mathbb{E}_{\mathcal{A}}[v(\mathcal{A})]$ can indeed be estimated with high probability using $\widetilde{\mathcal{O}}\left( \frac{B^2}{\varepsilon^2} \right)$ executions of $\mathcal{A}$ classically, while $\widetilde{\mathcal{O}}\left( \frac{B}{\varepsilon}\right)$ simulations of $\mathcal{A}$ are sufficient quantumly \cite{montanaro15}.\footnote{The $\widetilde{\mathcal{O}}$ notation hides polylogarithmic factors in $B$ and $1/\varepsilon$.} Moreover, these sample complexities are known to be asymptotically optimal, up to polylogarithmic factors, since estimating mean values reduces to this problem \cite{canetti95,nayak99}.

The multivariate case $v(\mathcal{A}) \in \mathbb{R}^d$ however, appearing notably in machine learning applications (see Section \ref{sec:applications}), remains largely unaddressed by quantum algorithms. Classically, the multivariate problem does not introduce any special considerations: the $d$ additive approximations of $\mathbb{E}_{\mathcal{A}}[v_1(\mathcal{A})], \ldots, \mathbb{E}_{\mathcal{A}}[v_d(\mathcal{A})]$ can all be computed simultaneously (i.e., using the same executions of $\mathcal{A}$) with only a logarithmic overhead $\log(d)$ in sample complexity (due to the Hoeffding bound, see \cref{appdx:classical-complexity-MVMC}). In the quantum scenario however, this \emph{simultaneous} evaluation of several expectation values is more complicated. The quantum algorithms for the univariate case rely on quantum amplitude estimation \cite{brassard02}, which involves as a critical step an encoding of the expectation values in the relative phase of a quantum register. Given that the phase is bounded and periodic, this imposes a normalization constraint on the encoding of multiple expectation values and therefore causes a linear overhead in $d$ without special restrictions on the random variables $(v_1(\mathcal{A}), \ldots, v_d(\mathcal{A}))$ (i.e., when they are only bounded in $\ell_\infty$-norm). It is an open question whether, and under which conditions, this overhead can be circumvented, while maintaining a quadratic speed-up in the precision $\varepsilon$ of the Monte Carlo estimates.

In this work, we investigate this apparent limitation of quantum algorithms to speed up multivariate Monte Carlo (MVMC) estimation without an associated slowdown. We choose to formulate the MVMC problem in the general framework of Markov reward processes (MRP) \cite{Cor18,sutton98}, i.e., Monte Carlo processes whose dynamics are described by Markovian transition probabilities and the random variables of interest are described by real-valued reward vectors. This formalism is convenient to separate the query complexity to the dynamics of the Monte Carlo process from that of its random variables, and characterize their upper and lower bounds. In our analysis, we only consider access models to the MRP reward vectors that are \emph{natural} to quantum algorithms, that is, where these real-valued vectors are encoded in normalized amplitudes or phases of input quantum states. The reason for this choice is that encoding in phase/amplitude is a common step in most quantum algorithms that assume, e.g., binary oracle access. By limiting our attention to these natural access models, we hence study the power (and limitations) of this type of encodings in quantum algorithms.

\paragraph{Main contributions}
We investigate the quantum query complexity of the multivariate Monte Carlo estimation problem in various quantum-accessible settings. In the general formulation of this problem in terms of Markov reward processes with vectorized rewards, we design quantum algorithms that speed up the estimation of the resulting expected rewards (or value functions). We also prove lower bounds showing that our query complexities are asymptotically optimal in all parameters, up to polylogarithmic factors. These query complexities are summarized in \cref{fig:table-query}. On a higher level, and to the best of our knowledge, this is the first multivariate problem where we have essentially tight bounds and hence complete understanding of its quantum query complexity, which suggests that the essentially optimal quantum algorithmic techniques required to tackle such multidimensional problems are collated in this work. We also describe several applications appearing in the field of quantum machine learning (quantum approaches to reinforcement learning, regression and generative modeling) that are compatible with our results.

\paragraph{Quantum algorithms} The algorithmic techniques we employ are extensions of those developed in recent years, combined in a novel way. First, we use the observation in \cite{jordan05}, later developed by \cite{GAW18}, that it is possible to perform concurrent runs of phase estimation, as long as one can compute the inner product between the vector of phases and the vectors $\mathbf{x}$ in a hypercubic lattice with uniform spacing, centered around the origin. Next, we combine the oracle conversion techniques from \cite{GAW18}, the value function estimation construction from \cite{Cor18}, and the idea in \cite{vanApeldoorn21} to compute these inner products, and improve over the latter result by realizing that the error probability can still be bounded if we can do so on all but a constant fraction of the vectors $\mathbf{x}$, similar to the analysis in \cite{GAW18}, Appendix A. Finally, we analyze our construction with some well-known techniques from statistics that extend our results to the case where the grid is rotated by an arbitrary orthogonal matrix.

\begin{table}[h]
	\centering
	\begin{tabular}{r|r|c|c|c}
		Input & Oracle type & Exact-depth & Cumulative-depth & Path-independent \\\hline
		\multirow{7}{*}{Reward vectors} & Phase & $\displaystyle \widetilde{\Theta}\left(\frac{R_{\max}}{\varepsilon}d^{1+\frac1p}\right)$ & $\displaystyle \widetilde{\Theta}\left(T^*\frac{R_{\max}}{\varepsilon}d^{1+\frac1p}\right)$ & $\displaystyle \widetilde{\Theta}\left(T^*\frac{R_{\max}}{\varepsilon}d^{1+\frac1p}\right)$ \\
		& Probability & $\displaystyle \widetilde{\Theta}\left(\frac{R_{\max}}{\varepsilon}d^{1-\frac{1}{2q}+\frac1p}\right)$ & $\displaystyle \widetilde{\Theta}\left(T^*\frac{R_{\max}}{\varepsilon}d^{1-\frac{1}{2q}+\frac1p}\right)$ & $\displaystyle \widetilde{\Theta}\left(T^*\frac{R_{\max}}{\varepsilon}d^{1-\frac{1}{2q}+\frac1p}\right)$ \\
		& Distribution & $\displaystyle \widetilde{\Theta}\left(\frac{R_{\max}}{\varepsilon}d^{1-\frac1q+\frac1p}\right)$ & $\displaystyle \widetilde{\Theta}\left(T^*\frac{R_{\max}}{\varepsilon}d^{1-\frac1q+\frac1p}\right)$ & $\displaystyle \widetilde{\Theta}\left(T^*\frac{R_{\max}}{\varepsilon}d^{1-\frac1q+\frac1p}\right)$ \\
		& Lattice & $\displaystyle \widetilde{\Theta}\left(\frac{R_{\max}}{\varepsilon}d^{1-\frac1q+\frac1p}\right)$ & $\displaystyle \widetilde{\Theta}\left(T^*\frac{R_{\max}}{\varepsilon}d^{1-\frac1q+\frac1p}\right)$ & $\displaystyle \widetilde{\Theta}\left(T^*\frac{R_{\max}}{\varepsilon}d^{1-\frac1q+\frac1p}\right)$ \\\hline
		Walk dynamics & Probability & $\displaystyle \widetilde{\Theta}\left(T\frac{R_{\max}}{\varepsilon}d^{\xi(p,q)}\right)$ & $\displaystyle \widetilde{\Theta}\left((T^*)^2\frac{R_{\max}}{\varepsilon}d^{\xi(p,q)}\right)$ & $\displaystyle \widetilde{\Theta}\left((T^*)^2\frac{R_{\max}}{\varepsilon}d^{\xi(p,q)}\right)$
	\end{tabular}\\
	where $T^* = \min\{T,1/(1-\gamma)\}$ and $\xi(p,q) = \frac1p + \max\{0,\frac12-\frac1q\}$.
	\caption{Summary of the query complexities in \cref{thm:mvmc-alg,thm:mvmc-lower-bounds}. The $d$-dimensional reward vectors of the Monte Carlo process are assumed bounded by $R_\text{max}$ in $\ell_q$-norm, and the estimates of their corresponding value functions are to be made $\varepsilon$-precise in $\ell_p$-norm. We consider exact-depth, cumulative-depth and path-independent definitions of these value functions with (effective) depths $T^{(*)}$ and a discount factor $\gamma \in [0,1]$ (see \cref{sec:MMCstatement}). The different types of oracle access to the reward vectors are defined in \cref{sec:MMCaccessmodels}. The $\widetilde{\Theta}$ notation holds in the limit $T^{(*)}, R_{\max}, d \to \infty$ and $\varepsilon \downarrow 0$, and hides polylogarithmic factors in $T^{(*)}$, $R_{\max}$, $d$ and $1/\varepsilon$. The lower bounds for the walk dynamics only hold in the case where $\varepsilon = \O(R_{\max}d^{-1+1/p+\xi(p,q)})$.}
	\label{fig:table-query}
\end{table}

\paragraph{Lower bounds} The lower bounds we give also arise from combining well-known results in novel ways. In the non-rotated case, we show that multivariate Monte Carlo estimation is at least as hard as recovering a bit string up to errors in only a constant fraction of the bits. In the lower bound on the query complexity of the reward oracle, we need to recover this bit string using a fractional phase oracle, whereas for the transition probability oracle, each of the bits is the output of a composition of the majority and parity function. The hardness of this problem is shown by combining the information theoretic lower bound from \cite{FGGS99} with both general adversary bounds for functions and relations, as described in \cite{Belovs15}, which, to the best of our knowledge, is a novelty that we expect to be of independent interest. Finally, to generalize to the rotated case in any $\ell_p$-norm, we modify the analysis with some norm conversions and advanced techniques from analysis and probability theory.

\paragraph{Related work}
Quantum algorithms for Monte Carlo estimation were first explored by \cite{montanaro15}. This work however only considers the univariate problem where the Monte Carlo process generates a 1-dimensional random variable. In addition to studying the case where this random variable is bounded in absolute value, the authors also present quantum algorithms compatible with random variables with bounded variance. We instead restrict our attention to bounded multidimensional random variables in $\ell_q$-norm, $q \geq 1$.

Concurrently to our work, \cite{vanApeldoorn21} has investigated the related problem of estimating probability vectors $\bm{p} = (p_1, \ldots, p_d)$, $\norm{\bm{p}}_1 = 1$, given access to their associated probability oracles (see \cref{subsec:oracles}). This problem can be viewed as a special case of MVMC for a depth-one process with $d$ accessible states, and unit $d$-dimensional reward vectors $(0,\ldots, 0, 1, 0, \ldots, 0)$ in each of these states. Therefore, the upper and lower bound results of \cite{vanApeldoorn21} can be derived from our main \cref{thm:mvmc-alg,thm:mvmc-lower-bounds}. Our results also solve one of the author's open questions: given a matrix $A \in [-1,1]^{d\times d}$ and access to a probability vector $\bm{p}$ via a probability oracle, the author asks whether a vector $\bm{q}$ such that $\norm{A\bm{p}-\bm{q}}_\infty \leq \varepsilon$ can be obtained with query complexity $\mathcal{O}\left(\frac{\sqrt{d}}{\varepsilon}\right)$ to this probability oracle, which is provably optimal. Using the algorithms in \cref{thm:mvmc-alg} for a depth-one process with $d$ accessible states, and reward vectors corresponding to a column of $A$ in each of these states (hence $q=\infty$ and, by definition of the problem, $p=\infty$), we obtain this query complexity of the transition probability oracle, up to logarithmic factors.

\paragraph{Organization}
We first introduce in \cref{sec:prelim} some preliminary notions used in this manuscript, namely the types of oracles we use and known techniques to convert between these. We then provide a formal statement of the MVMC problem in \cref{sec:MMCstatement}, and \cref{sec:MMCaccessmodels} covers the different ways in which we assume to have oracular access to the specific instances of this problem. After that, in \cref{sec:MMCalgorithms} we provide algorithms that solve it for each of the assumed access models and analyses of their query complexities, and in \cref{sec:MMClowerbounds} we prove corresponding lower bounds on these query complexities. In \cref{sec:applications}, we describe applications that can be formulated as MVMC estimation and for which we discuss the implications of our results. Finally, in \cref{sec:MMCdiscussion}, we discuss our results more broadly and hint at some interesting lines of future research.

\section{Preliminaries}
\label{sec:prelim}

In this section, we introduce the preliminary notions that will be used throughout this paper. We start with some notational remarks in \cref{subsec:notation}. After that we discuss the types of oracular access that we consider, in \cref{subsec:oracles}, and finally we recall some results about oracle conversions from \cite{GAW18} in \cref{subsec:analogcomputation}.

\subsection{Notation}
\label{subsec:notation}

Throughout this text, we use $\N = \{1,2,\dots\}$, and for all $n \in \N$, we let $[n] = \{1, \dots, n\}$ and $[n]_0 = \{0,1,\dots,n\}$. An orthogonal matrix $O$ is a matrix with real entries that satisfies $OO^T = O^TO = I$, i.e., it is a unitary matrix with real entries. Whenever we refer to a Hadamard matrix $H \in \R^{d \times d}$, with $d \in \N$, then we refer to an orthogonal matrix with entries $\pm 1/\sqrt{d}$, which in particular implies that $H^TH = HH^T = I$.

If $G$ is a finite set, then any random variable $S \sim \Unif(G)$ takes values in $G$ uniformly at random. Whenever we use boldface letters, such as $\mathbf{x}$, we denote vectors with real-valued entries. Referring to their components is usually done without boldface, i.e., $x_j$ denotes the $j$th entry of $\mathbf{x}$.

When we use big-$\O$-notation, we always supply the corresponding limit in which it holds, i.e., when a function $f : \N \to \R$ satisfies $f(n) = \O(n)$ in the limit $n \to \infty$, then there exists a $N \in \N$ and $M \geq 0$ such that for all $n \geq N$, $|f(n)| \leq Mn$. Similarly, if $f : [0,1] \to \R$ and $f(\varepsilon) = \O(1/\varepsilon)$ in the limit where $\varepsilon \downarrow 0$, then there exists a $\delta \in [0,1]$ and $M \geq 0$ such that $|f(\varepsilon)| \leq M/\varepsilon$, for all $0 < \varepsilon \leq \delta$.

We can also use the big-$\O$-notation with multiple variables, each with their corresponding limits. If $f : \N \to [0,1] \to \R$ satisfies $f(n,\varepsilon) = \O(n/\varepsilon)$, then we mean that there exists an $M \geq 0$, $n_0 \geq \N$ and $\delta \in [0,1]$, such that whenever either $n \geq n_0$ or $0 < \varepsilon \leq \delta$, it holds that $|f(n,\varepsilon)| \leq Mn/\varepsilon$. The \textit{or} is important here, if \textit{any} of the variables displayed in the big-$\O$-notation is close to its limit, then this implies that the left-hand side is bounded by the right-hand side, up to a universal constant.

We can hide polylogarithmic factors in the big-$\O$-notation. If we write $f(n,\delta) = \O(g(n,\delta)\polylog(n/\delta))$, we mean that there exists an integer $m \in \N$ such that $f(n,\delta) = \O(g(n,\delta)\log^m(n/\delta))$. Note in particular that this notation is transitive, i.e., if $f(x) = \O(\polylog(g(x)))$ and $g(x) = \O(\polylog(h(x)))$, then $f(x) = \O(\polylog(h(x)))$ as well. If the expression inside the polylog equals the rest of the expression in the big-$\O$-notation, then we can abbreviate this with a tilde, i.e., $f(x) = \O(g(x)\polylog(g(x)))$ can be abbreviated to $f(x) = \widetilde{\O}(g(x))$.

If $f(x) = \O(g(x))$, then we can equally write $g(x) = \Omega(f(x))$, and similarly with the tilde. If both $f(x) = \O(g(x))$ and $f(x) = \Omega(g(x))$, then we write $f(x) = \Theta(g(x))$, and similarly we write $f(x) = \widetilde{\Theta}(g(x))$ if $f$ and $g$ are equal to one another up to polylogarithmic factors.

When we use quantum states, a ``ket'', $\ket{\cdot}$, always denotes a unit vector. A quantum register, or register for short, is a state space with an associated canonical basis. When we say that a quantum algorithm is acting on several registers, we mean that it is acting on the state space formed by the tensor product of the individual registers. In this case, we omit the tensor product symbol when talking about the individual states, i.e., we write $\ket{x}\ket{y}$ instead of $\ket{x} \otimes \ket{y}$, when it is understood that $\ket{x}$ is a state in the first register, and $\ket{y}$ a state in the second register.

\subsection{Types of oracle access}
\label{subsec:oracles}

In theoretical computer science, it is common to give algorithms the possibility of accessing their input through a routine that satisfies a given specification. Such routines are collectively referred to as oracles, and are meant to be implemented by the user of these algorithms, encoding the specific instance on which the algorithm is to be run.

The specifications of these subroutines, however, can vary. For some applications, one way of implementing the oracle circuit might be much more natural than another. To accommodate a wide range of applications and ensure flexible usage of our results, we investigate several types of oracles that our algorithm can access.

We give our algorithms access to functions $f : X \to \R^d$, where $X$ is a finite set, through four different types of oracles. They are outlined in \cref{def:oracles}.

\begin{definition}[Oracle types]
	\label{def:oracles}
	Let $d \in \N$ and let $X$ be a finite set, whose elements $x \in X$ can be encoded in mutually orthogonal states $\ket{x}$ in an input register.
	
	\begin{enumerate}
		\item \textbf{Phase oracles}
		
		First, for all $j \in [d]$ let $a_j < b_j$, and $f : X \to \times_{j=1}^d [a_j,b_j] \subseteq \R^d$. The \textit{phase oracle} evaluating $f$ is the operation $O_f$ that acts on the input register and a register with basis states $\ket{j}$ for all $j \in [d]$, and is defined as
		\[O_f : \ket{x}\ket{j} \mapsto e^{i\frac{f(x)_j - \frac{b_j+a_j}{2}}{b_j - a_j}}\ket{x}\ket{j}, \qquad (x \in X, j \in [d]).\]
		\item  \textbf{Probability oracles}
		
		Similarly, for all $j \in [d]$ let $b_j > 0$, and $f : X \to \times_{j=1}^d [0,b_j] \subseteq \R^d$. The \textit{probability oracle} evaluating $f$ is the operation $U_f$ that acts on the input register, a register that contains the basis states $\ket{j}$ for all $j \in [d]$, and one extra qubit, and is defined as
		\[U_f : \ket{x}\ket{j}\ket{0} \mapsto \ket{x}\ket{j}\left(\sqrt{\frac{f(x)_j}{b_j}}\ket{1} + \sqrt{1-\frac{f(x)_j}{b_j}}\ket{0}\right), \qquad (x \in X, j \in [d]).\]
		\item  \textbf{Distribution oracles}
		
		Next, let $B > 0$ and $f : X \to \R_{\geq 0}^d$ such that for all $x \in X$, $\norm{f(x)}_1 \leq B$. Then, the \textit{distribution oracle} evaluating $f$ is the operation $D_f$ that acts on the input register and a register that contains basis states $\ket{j}$ for all $j \in [d]_0$, and is defined as
		\[D_f : \ket{x}\ket{0} \mapsto \ket{x}\left(\sum_{j=1}^d \sqrt{\frac{f(x)}{B}}\ket{j} + \sqrt{1 - \frac{\norm{f(x)}_1}{B}}\ket{0}\right), \qquad (x \in X).\]
		\item  \textbf{Lattice oracles}
		
		Finally, let $f : X \to Y \subseteq \R^d$. Furthermore, let $G \subseteq \R^d$ be a finite set, whose elements $\mathbf{x} \in G$ can also be encoded in mutually orthogonal states $\ket{\mathbf{x}}$ in a \textit{vector register}. We let
		\[m = \min\{\mathbf{x}^T\mathbf{y} : \mathbf{x} \in G, \mathbf{y} \in Y\}, \qquad \text{and} \qquad M = \max\{\mathbf{x}^T\mathbf{y} : \mathbf{x} \in G, \mathbf{y} \in Y\}.\]
		A \textit{lattice oracle} evaluating $f$ on $G$ is an operator $L_{G,f}$ that acts on the vector and input register, as
		\[L_{G,f} : \ket{\mathbf{x}}\ket{x} \mapsto e^{i\frac{\mathbf{x}^Tf(x) - \frac{M+m}{2}}{M-m}}\ket{\mathbf{x}}\ket{x}, \qquad (\mathbf{x} \in G, x \in X).\]
	\end{enumerate}
\end{definition}

Phase and probability oracles were already explicitly defined before, in \cite{GAW18}, Definitions 8 and 10. The distribution and lattice oracles have also been used before. For instance in quantum random walks, it is common to model the walking dynamics on the graph via distribution oracles, first introduced by \cite{Ambainis03} and later developed by \cite{Szegedy04} and subsequent works. Lattice oracles are for instance considered in \cite{vanApeldoorn20}, page 24.

We remark that the set $G$ in the definition of lattice oracles, need not have any lattice structure. However, we will often consider the set $G$ to be (part of) a lattice, as for instance depicted in \cref{fig:hypercubiclattice}, motivating our choice to refer to these objects as lattice oracles in this setting. In other settings, the term \textit{inner product oracle} might be more apt.

\subsection{Oracle conversion techniques}
\label{subsec:analogcomputation}

In this section, we recall some techniques from \cite{GAW18}, regarding phase and probability oracles. On a high level, these results allow us to interconvert between these two oracle types with multiplicative overhead that is only polylogarithmic in the precision. This in turn enables us to do non-trivial arithmetic computations with these objects, that is, any combination of additions and multiplications, without ever storing intermediate values in binary. As such, we refer to these techniques as \textit{analog computation}, as opposed to \textit{digital computation} where the binary digits are explicitly manipulated.

\begin{lemma}[Conversion from probability oracle to phase oracle]
	\label{lem:probabilitytophase}
	Let $d \in \N$ and for all $j \in [d]$, $b_j > 0$. Let $f : X \to \times_{j=1}^d [0,b_j] \subseteq \R^d$, and let $U_f$ be a probability oracle evaluating $f$. We can construct a phase oracle evaluating $f$, $O_f$, up to operator norm error $\delta > 0$, with $\O(\polylog(1/\delta))$ calls to $U_f$, as $\delta \downarrow 0$.
\end{lemma}

\begin{proof}
	This follows directly from \cite{GAW18}, Theorem 14.
\end{proof}

\begin{lemma}[Conversion from phase oracle to probability oracle]
	\label{lem:phasetoprobability}
	Let $d \in \N$ and for all $j \in [d]$, $a_j < b_j$. Let $f : X \to \times_{j=1}^d [a_j,b_j] \subseteq \R^d$, and let $O_f$ be a phase oracle evaluating $f$. Let $g : X \to \times_{j=1}^d [0,2(b_j-a_j)] \subseteq \R^d$, $g(x)_j = f(x)_j - a_j + \frac{b_j - a_j}{2}$. We can construct a probability oracle evaluating $g$, $U_g$, up to operator norm $\delta > 0$, with $\O(\polylog(1/\delta))$ calls to $O_f$, as $\delta \downarrow 0$.
\end{lemma}

\begin{proof}
	Note that $g$ has codomain $\times_{j=1}^d [0,2(b_j-a_j)]$, but its range will be bounded away from the end points, and contained in the smaller set $\times_{j=1}^d [(b_j-a_j)/2, 3(b_j-a_j)/2]$. This allows us to use \cite{GAW18}, Lemma 16.
\end{proof}

Note that \cref{lem:phasetoprobability} subtly changes the codomain of the function to be computed. If $f$ is a function taking values in the interval $[1,3]$, then the probability oracle constructed in \cref{lem:phasetoprobability} computes the same function but with codomain $[0,4]$, ensuring that all the actual function values are far away from the boundaries of the codomain. This is necessary, as already remarked at the bottom of page 21 in \cite{GAW18}, since probability oracles that compute function values close to the boundaries of the codomain are actually more powerful than phase oracles. This explains why we obtain different query complexity results when we have access to phase oracles on the one hand and probability oracles on the other.

\section{Multivariate Monte Carlo estimation}
\label{sec:MMCstatement}

In this section, we formally introduce the problem that we aim to solve in this paper. As already mentioned in the introduction, we phrase the problem in the language of Markov reward processes, which mimics the nomenclature used in reinforcement learning. This allows us to clearly differentiate between the two different defining properties of Monte Carlo estimation problems, namely the values of the random variables themselves, and the dynamics underlying their probability distribution. For more elaborate introductions into Markov reward processes, one can for instance consult \cite{Cor18}, or \cite{sutton98}.

\subsection{Markovian walks}

Let $S$ be a finite set, referred to as the \textit{state space}, whose elements $s \in S$ we refer to as \textit{states}. Since $S$ is finite, one can always think of these states as nodes in a graph. Next, we define a \textit{probability transition matrix} $P : S \times S \to [0,1]$, which satisfies the constraint that every row $P(s,\cdot)$ is a probability distribution. Intuitively, the entry $P(s_1,s_2)$ in the probability transition matrix describes how likely it is to traverse from state $s_1$ to $s_2$ during one time step. Together with an \textit{initial state} $s_0 \in S$ and a positive integer $T$, we can now define a Markovian walk of length $T$ on the state space starting at $s_0$. To that end, we define a sequence of random variables $s_1, \dots, s_T$ taking values in $S$, such that
\[\P(s_t = s' | s_{t-1} = s) = P(s,s').\]
Furthermore, the probability that this Markovian walk traverses any given path $\tau = (s_0, \dots, s_T) \in \{s_0\} \times S^T$ is given by
\[\P(\tau) = \prod_{t=1}^T P(s_{t-1},s_t).\]
We denote the probability distribution over all such paths of length $T$ starting at initial state $s_0$ by $P(T;s_0)$.

\subsection{Reward functions}
\label{subsec:settings}

Next, we turn to a concept called the \textit{reward function}. We consider three different cases, each successive one imposing more structure on this object than the last.

\subsubsection{Exact-depth case}

Let $d,T \in \N$. To every path $\tau \in S^{T+1}$, we associate a $d$-dimensional reward vector with real entries, i.e., we define a \textit{depth-$T$ reward function} $R : S^{T+1} \to \R^d$. In order to provide rigorous convergence results, we fix some $R_{\max} > 0$ and $q \in [1,\infty]$, and we require that all reward vectors are bounded in $\ell_q$-norm by $R_{\max}$.

For any given initial state $s_0 \in S$, we can now ask how much reward the Markovian walk of length $T$ will obtain on average. This is captured by the \textit{depth-$T$ value function}, which we define as
\begin{equation}
\label{eq:Vexactdepth}
V : S \to \R^d, \qquad V(s_0) = \underset{\tau \sim P(T;s_0)}{\E}\left[R(\tau)\right].
\end{equation}
Loosely speaking, the exact-depth multivariate Monte Carlo estimation problem amounts to finding a sufficiently close approximation to this value function $V$. For concreteness, we fix $\varepsilon > 0$, $p \in [1,\infty]$ and $O$ a $d \times d$ orthogonal matrix (i.e., $O$ has real entries and $OO^T = O^TO = I$), and we demand that the resulting approximation to $OV(s_0)$ is $\varepsilon$-close in $\ell_p$-norm, i.e., if the outcome of our algorithm is $\mathbf{v} \in \R^d$, then we demand that
\[\norm{\mathbf{v} - OV(s_0)}_p \leq \varepsilon.\]
In short, we say that an algorithm solves the multivariate Monte Carlo estimation problem rotated by $O$ with rewards bounded by $R_{\max}$ in $\ell_q$-norm $\varepsilon$-precisely w.r.t.\ the $\ell_p$-norm, if it produces a vector $\mathbf{v} \in \R^d$ satisfying the above constraint with probability at least $2/3$.

\subsubsection{Cumulative-depth case}

Let $d \in \N$ and $T \in \N \cup \{\infty\}$. For all integers $t$ such that $0 \leq t \leq T$, we now define a depth-$t$ reward function $R^{(t)} : S^{t+1} \to \R^d$, with function values bounded by $R_{\max}$ in $\ell_q$-norm for some $R_{\max} > 0$ and $q \in [1,\infty]$. Let $\gamma \in [0,1]$ be a \textit{discount factor}. We assert that if $T = \infty$ then $\gamma < 1$, and define the following depth-$T$ reward function $R : S^{T+1} \to \R^d$, in the sense of the previous section, for every path $\tau = (s_0, \dots, s_T) \in S^{T+1}$ as
\[R(\tau) = \sum_{t=0}^T \gamma^tR^{(t)}(s_0, \dots, s_t).\]
The cumulative-depth Monte Carlo estimation problem is the problem of solving the exact-depth Monte Carlo estimation problem, with this particular structure on the reward function $R$.

Observe that the value function, with this extra structure in $R$, can now be written as
\[V(s_0) = \underset{\tau \sim P(T;s_0)}{\E} \left[R(\tau)\right] = \underset{\tau \sim P(T;s_0)}{\E} \left[\sum_{t=0}^T \gamma^t R^{(t)}(s_0, \dots, s_t)\right] = \sum_{t=0}^T \gamma^t \underset{\tau \sim P(t;s_0)}{\E} \left[R^{(t)}(s_0, \dots, s_t)\right],\]
where we used the shorthand notation $\tau = (s_0, \dots, s_T)$. Hence, if one has an algorithm that solves the exact-depth Monte Carlo estimation problem, then one can solve the cumulative-depth case by solving the exact-depth case for each of the individual $R^{(t)}$'s with $t$ running from $0$ to $T$, and then summing all the outcomes. Doing this naively, however, means that we would have to perform all these individual runs with precision $\varepsilon/T$, introducing an extra multiplicative overhead of $T$ to the resulting query complexity. Therefore, we employ a slightly more elaborate method to convert the cumulative-depth case into the exact-depth case. The details are in \cref{lem:valuefunctionoraclecumulativedepth}.

This setting might look artificial, but it is useful in the context of the policy gradient theorem in reinforcement learning, since this essentially reduces evaluating the gradient of the value function of a one-dimensional path-independent version of the problem, to solving the multivariate cumulative-depth case. We talk about this in more detail in \cref{subsec:rl-application}.

\subsubsection{Path-independent case}

In many cases it is natural to associate reward vectors to individual states, and have the reward obtained along a path be some weighted sum of the rewards obtained at the states the walk traverses. This is a special case of the cumulative-depth case, and we formalize it here.

Let $d \in \N$ and $T \in \N \cup \{\infty\}$. We associate a reward vector to every state, i.e., we have a \textit{state-reward function} $R_S : S \to \R^d$. Similarly as before, we require that all reward vectors are bounded by $R_{\max}$ in $\ell_q$-norm, for some $R_{\max} > 0$ and $q \in [1,\infty]$. Furthermore, we define a discount factor $\gamma = [0,1]$, such that if $T = \infty$ then $\gamma < 1$. For all integer $t$ satisfying $0 \leq t \leq T$, we define the depth-$t$ reward function, in the sense of the previous section, as
\[R^{(t)} : S^{t+1} \to \R^d, \qquad R^{(t)}(s_0, \dots, s_t) = \gamma^tR_S(s_t).\]
The path-independent version of the Monte Carlo problem is the same as the cumulative-depth version, with this extra constraint on the structure of the reward functions $R^{(t)}$.

Note that the value function, with this extra structure on the reward function $R$, can be rewritten into the more familiar form
\[V(s_0) = \underset{\tau \sim P(T;s_0)}{\E}\left[R(\tau)\right] = \underset{\tau \sim P(T;s_0)}{\E}\left[\sum_{t=0}^T \gamma^t R^{(t)}(s_0, \dots, s_t)\right] = \underset{\tau \sim P(T;s_0)}{\E}\left[\sum_{t=0}^T \gamma^tR_S(s_t)\right],\]
where the shorthand notation $\tau = (s_0, \dots, s_T) \in S^{T+1}$ is used. This formulation is closer to the definition in \cite{sutton98}.

For future reference, in the cumulative-depth and path-independent cases, we define a quantity called the \textit{effective depth}, as
\begin{equation}
	\label{eq:Tstar}
	T^* = \min\left\{T, \frac{1}{1-\gamma}\right\}.
\end{equation}

\subsection{Access models}
\label{sec:MMCaccessmodels}

In order to construct quantum algorithms that solve the multivariate Monte Carlo estimation problem, we must describe how such algorithms have access to the quantities that define the specific instance of the problem. In this section, we present several such input models.

In all cases, we assume that we have full classical knowledge of the state space $S$, the reward bounds $q$ and $R_{\max}$, the depth $T$, the discount factor $\gamma$, the rotation matrix $O$ and the precision parameters $p$ and $\varepsilon$. We also assume that all states $s \in S$ are encoded into basis states $\ket{s}$ of a \textit{state register}. Finally, we assume to have access to the probability transition matrix by means of an oracle that acts on two such state registers, as
\begin{equation}
\label{eq:DP}
D_P : \ket{s}\ket{0} \mapsto \ket{s}\sum_{s' \in S} \sqrt{P(s,s')}\ket{s'}, \qquad (s \in S).
\end{equation}
Here, $\ket{0}$ can be any arbitrary fiducial state. For future reference, we also mention that we can take several state registers, and combine them into a single \textit{path register}. Context makes it clear how many state registers go into a single path register.

It remains to describe how we have access to the reward function. For this, we consider several different options, each with its own version for the exact-depth, cumulative-depth and path-independent case.

\subsubsection{Phase oracles}

If we say that we have access to the reward function by means of a phase oracle, we mean that we can make queries to the oracles $O_R$, $O_{R^{(t)}}$, or $O_{R_S}$, in the exact-depth, cumulative-depth or path-independent cases respectively. Besides a register that contains the states $\ket{j}$ for all $j \in [d]$, $O_R$ and $O_{R^{(t)}}$ act on a path register, whereas $O_{R_S}$ only acts on a state register. Their action is defined as
\begin{align}
\label{eq:OR}
O_R &: \ket{\tau}\ket{j} \mapsto e^{i\frac{R(\tau)_j}{2R_{\max}}}\ket{\tau}\ket{j}, & (\tau \in S^{T+1}, j \in [d]), \\
\label{eq:ORt}
O_{R^{(t)}} &: \ket{\tau}\ket{j} \mapsto e^{i\frac{R^{(t)}(\tau)_j}{2R_{\max}}}\ket{\tau}\ket{j}, & (t \in [T]_0, \tau \in S^{t+1}, j \in [d]), \\
\label{eq:ORS}
O_{R_S} &: \ket{s}\ket{j} \mapsto e^{i\frac{R_S(s)_j}{2R_{\max}}}\ket{s}\ket{j}, & (s \in S, j \in [d]).
\end{align}

\subsubsection{Probability oracles}

If we say that we have access to the reward function by means of a probability oracle, we assume that the rewards are entry-wise non-negative, and that we can make queries to the oracles $U_R$, $U_{R^{(t)}}$, or $U_{R_S}$, in the exact-depth, cumulative-depth and path-independent cases respectively. These oracles act on the same registers as their phase oracle counterparts, and one additional qubit, and their action is defined as
\begin{align}
\label{eq:UR}
U_R &: \ket{\tau}\ket{j}\ket{0} \mapsto \ket{\tau}\ket{j}\left(\sqrt{\frac{R(\tau)_j}{R_{\max}}}\ket{1} + \sqrt{1-\frac{R(\tau)_j}{R_{\max}}}\ket{0}\right), & (\tau \in S^{T+1}, j \in [d]), \\
\label{eq:URt}
U_{R^{(t)}} &: \ket{\tau}\ket{j}\ket{0} \mapsto \ket{\tau}\ket{j} \left(\sqrt{\frac{R^{(t)}(\tau)_j}{R_{\max}}}\ket{1} + \sqrt{1-\frac{R^{(t)}(\tau)_j}{R_{\max}}}\ket{0}\right), & (t \in [T]_0, \tau \in S^{t+1}, j \in [d]), \\
\label{eq:URS}
U_{R_S} &: \ket{s}\ket{j}\ket{0} \mapsto \ket{s}\ket{j}\left(\sqrt{\frac{R(s)_j}{R_{\max}}}\ket{1} + \sqrt{1-\frac{R(s)_j}{R_{\max}}}\ket{0}\right), & (s \in S, j \in [d]).
\end{align}

\subsubsection{Distribution oracles}

If we say that we have access to the reward function by means of a distribution oracle, we assume that the reward vectors are entry-wise non-negative, and that we can make queries to the oracles $D_R$, $D_{R^{(t)}}$, or $D_{R_S}$, in the exact-depth, cumulative-depth and path-independent cases respectively. All of these oracles act on a register that contains the states $\ket{j}$ for all $j \in [d]_0$, and $D_R$ and $D_{R^{(t)}}$ also act on a path register, whereas $D_{R_S}$ uses a state register instead. The action is defined as
\begin{align}
\label{eq:DR}
D_R &: \ket{\tau}\ket{0} \mapsto \ket{\tau}\left(\sum_{j=1}^d \sqrt{\frac{R(\tau)_j}{d^{1-\frac1q}R_{\max}}}\ket{j} + \sqrt{1 - \frac{\norm{R(\tau)}_1}{d^{1-\frac1q}R_{\max}}}\ket{0}\right), & (\tau \in S^{T+1}, j \in [d]), \\
\label{eq:DRt}
D_{R^{(t)}} &: \ket{\tau}\ket{0} \mapsto \ket{\tau}\left(\sum_{j=1}^d \sqrt{\frac{R^{(t)}(\tau)_j}{d^{1-\frac1q}R_{\max}}}\ket{j} + \sqrt{1 - \frac{\norm{R^{(t)}(\tau)}_1}{d^{1-\frac1q}R_{\max}}}\ket{0}\right), & (t \in [T]_0, \tau \in S^{t+1}, j \in [d]), \\
\label{eq:DRS}
D_{R_S} &: \ket{s}\ket{0} \mapsto \ket{s}\left(\sum_{j=1}^d \sqrt{\frac{R_S(s)_j}{d^{1-\frac1q}R_{\max}}}\ket{j} + \sqrt{1 - \frac{\norm{R_S(s)}_1}{d^{1-\frac1q}R_{\max}}}\ket{0}\right),\!\!\!\!\!\!\! & (s \in S, j \in [d]).
\end{align}

\subsubsection{Lattice oracles}

If we say that we have access to the reward function by means of a lattice oracle on a finite set $G \subseteq \R^d$, we assume that we can make queries to the oracles $L_{G,R}$, $L_{G,R^{(t)}}$, or $L_{G,R_S}$, in the exact-depth, cumulative-depth and path-independent cases respectively. These oracles all act on a \textit{vector register}, encoding the basis states $\ket{\mathbf{x}}$ for all $\mathbf{x} \in G$, and $L_{G,R}$ and $L_{G,R^{(t)}}$ act on a path register, whereas $L_{G,R_S}$ acts on a state register. Their action is defined as
\begin{align}
	\label{eq:LR}
	L_{G,R} &: \ket{\mathbf{x}}\ket{\tau} \mapsto e^{i\frac{\mathbf{x}^TR(\tau)}{2r_G(q)R_{\max}}}\ket{\mathbf{x}}\ket{\tau}, & (\mathbf{x} \in G, \tau \in S^{T+1}), \\
	\label{eq:LRt}
	L_{G,R^{(t)}} &: \ket{\mathbf{x}}\ket{\tau} \mapsto e^{i\frac{\mathbf{x}^TR^{(t)}(\tau)}{2r_G(q)R_{\max}}}\ket{\mathbf{x}}\ket{\tau}, & (\mathbf{x} \in G, t \in [T]_0, \tau \in S^{t+1}), \\
	\label{eq:LRS}
	L_{G,R_S} &: \ket{\mathbf{x}}\ket{s} \mapsto e^{i\frac{\mathbf{x}^TR_S(s)}{2r_G(q)R_{\max}}}\ket{\mathbf{x}}\ket{s}, & (\mathbf{x} \in G, s \in S),
\end{align}
where $r_G : [1,\infty] \to \R$ is defined as
\begin{equation}
	\label{eq:rG}
	r_G(q) = \max\left\{\mathbf{x}^T\mathbf{y} : \mathbf{x} \in G, \mathbf{y} \in \R^d, \norm{\mathbf{y}}_q \leq 1\right\} = \max_{\mathbf{x} \in G} \norm{\mathbf{x}}_{\frac{1}{1-\frac1q}}.
\end{equation}
Intuitively, if $G$ is point-symmetric around the origin, one can think of $r_G$ as the radius of the set $G$.

\section{Quantum algorithms}
\label{sec:MMCalgorithms}

In this section, we present algorithms that solve the multivariate Monte Carlo estimation problem for all settings defined in \cref{subsec:settings}, and all oracles types defined in \cref{sec:MMCaccessmodels}. We construct these algorithms in three steps. First, in \cref{subsec:rewardoracleconversions} we show how the access models introduced in \cref{sec:MMCaccessmodels} can be converted into a lattice oracle. Then, in \cref{subsec:valueoracleconstruction}, we show how this lattice oracle can be used compute inner products with the value function. Finally, in \cref{subsec:valuefunctionestimation}, we show how we can use these building blocks to actually retrieve a classical description of the value function. We end with a summary of our results in \cref{subsec:algorithmresults}.

Throughout our constructions, we will use three different notions of \textit{radius} of a finite set $G \subseteq \R^d$. We already saw the definition of $r_G$ in \cref{eq:rG}, which we repeat here for convenience, along with the notions of \textit{approximate radius} and \textit{effective radius}. To that end, let $\delta \geq 0$, and $r_{G,\delta}, \overline{r}_{G,\delta} : [1,\infty] \to \R$, defined as
\begin{align*}
	r_G(q) &= \max\left\{\mathbf{x}^T\mathbf{y} : \mathbf{x} \in G, \mathbf{y} \in \R^d, \norm{\mathbf{y}}_q \leq 1\right\} = \norm{\mathbf{x}}_{\frac{1}{1-\frac1q}}, & \text{(radius)}, \\
	r_{G,\delta}(q) &= \min\left\{t \geq 0 : \underset{\mathbf{x} \sim \Unif(G)}{\P}\left[\norm{\mathbf{x}}_{\frac{1}{1-\frac1q}} \geq t\right] \leq \delta\right\}, & \text{(approximate radius)}, \\
	\overline{r}_{G,\delta}(q) &= \min\left\{t \geq 0 : \forall \mathbf{y} \in \R^d, \norm{\mathbf{y}}_q \leq 1 \Rightarrow \underset{\mathbf{x} \sim \Unif(G)}{\P}\left[\left|\mathbf{x}^T\mathbf{y}\right| \geq t\right] \leq \delta \right\}, & \text{(effective radius)}.
\end{align*}
Given any such finite set $G \subseteq \R^d$, we define its trimmed versions as
\begin{equation}
	\label{eq:trimmedsets}
	G^{(q)}_{\delta} = \left\{\mathbf{x} \in G : \norm{\mathbf{x}}_{\frac{1}{1-\frac1q}} \leq r_{G,\delta}(q)\right\}, \qquad \text{and} \qquad G^{(q)}_{\delta,\mathbf{y}} = \left\{\mathbf{x} \in G : \left|\mathbf{x}^T\mathbf{y}\right| \leq \overline{r}_{G,\delta}(q)\right\}.
\end{equation}
These sets are graphically displayed in \cref{fig:trimmedsets}. It is immediately clear that the trimmed sets $G_{\delta}^{(q)}$ and $G_{\delta,\mathbf{y}}^{(q)}$ are both subsets of $G$, and in general the trimmed sets are incomparable, as is for instance the case in the figure. We can also say something about the fraction of points in $G$ that are contained in these trimmed sets, which is the objective of the following lemma.

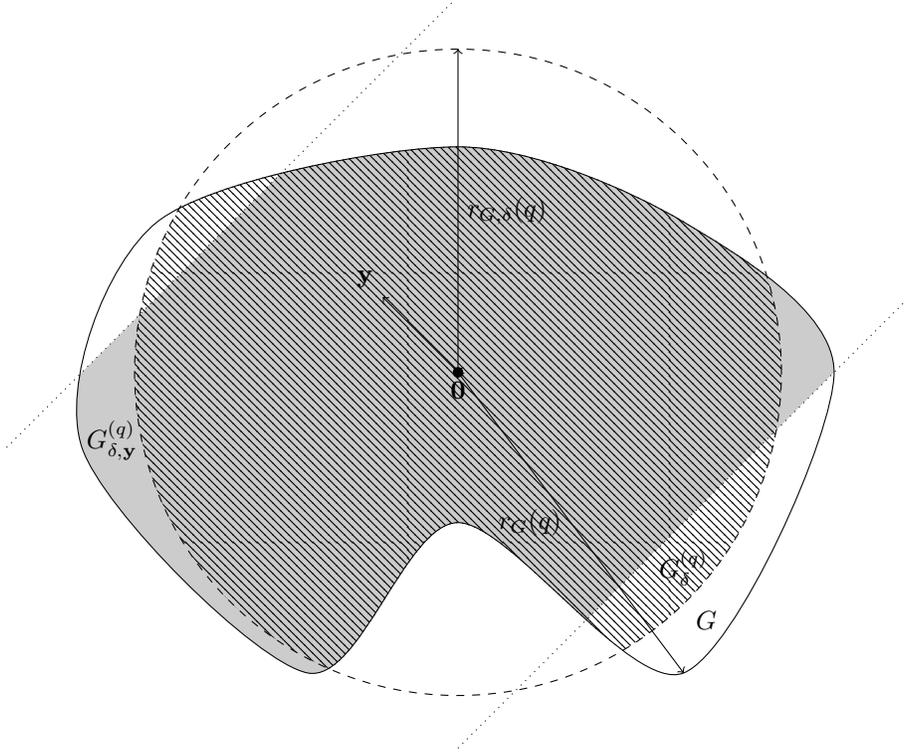
\begin{figure}[h!]
	\centering
	\begin{tikzpicture}
	\begin{scope}
	\clip (-6,-1) to (0,5) to (6,1) to (0,-5) to (-2,-5);
	\fill[gray!40] plot[smooth cycle] coordinates {(5,0) (3,2) (0,3) (-4,2) (-5,-1) (-2,-4) (0,-2) (3,-4)};
	\end{scope}
	\begin{scope}
	\clip (0,0) circle[radius=4.3];
	\draw[pattern=north west lines, dotted] plot[smooth cycle] coordinates {(5,0) (3,2) (0,3) (-4,2) (-5,-1) (-2,-4) (0,-2) (3,-4)};
	\end{scope}
	\draw plot[smooth cycle] coordinates {(5,0) (3,2) (0,3) (-4,2) (-5,-1) (-2,-4) (0,-2) (3,-4)};
	\fill (0,0) circle[radius=.2em] node[below] {$\mathbf{0}$};
	\draw[dashed] (0,0) circle[radius=4.3];
	\draw[dotted] (-6,-1) to (0,5);
	\draw[dotted] (0,-5) to (6,1);
	\draw[->] (0,0) to (-1,1) node[above left] {$\mathbf{y}$};
	\draw[<->] (0,0) to node[left] {$r_G(q)$} (3,-4);
	\draw[<->] (0,0) to node[right] {$r_{G,\delta}(q)$} (0,4.3);
	\node at (3.3,-3.3) {$G$};
	\node at (3,-2.6) {$G_{\delta}^{(q)}$};
	\node at (-4.6,-.9) {$G_{\delta,\mathbf{y}}^{(q)}$};
	\end{tikzpicture}
	\caption{A set $G$, and its trimmed sets $G_{\delta}^{(q)}$ and $G_{\delta,\mathbf{y}}^{(q)}$. The set $G$ is the region enclosed by the solid line. $G_{\delta}^{(q)}$ is represented by the dotted region, and $G_{\delta,\mathbf{y}}^{(q)}$ is the gray region. From the figure, it is apparent that even though both trimmed sets are subsets of $G$, they are incomparable with one another.}
	\label{fig:trimmedsets}
\end{figure}

\begin{lemma}[Trimmed sets]
	\label{lem:trimmedgrids}
	Let $G \subseteq \R^d$ be a finite set, $q \in [1,\infty]$, $\delta \geq 0$ and $\mathbf{y} \in \R^d$ with $\norm{\mathbf{y}}_q \leq 1$. Then,
	\[\left|G_{\delta}^{(q)}\right| \geq (1-\delta)|G|, \qquad \text{and} \qquad \left|G_{\delta,\mathbf{y}}^{(q)}\right| \geq (1-\delta)|G|,\]
	and we also have that $r_{G^{(q)}_{\delta}}(q) = r_{G,\delta}(q)$.
\end{lemma}

\begin{proof}
	All proofs are one-liners, e.g.,
	\[\left|G_{\delta}^{(q)}\right| \geq |G| \cdot \underset{\mathbf{x} \sim \Unif(G)}{\P}\left[\mathbf{x} \in G_{\delta}^{(q)}\right] = |G| \cdot \underset{\mathbf{x} \sim \Unif(G)}{\P}\left[\norm{\mathbf{x}}_{\frac{1}{1-\frac1q}} \leq r_{G,\delta}(q)\right] \geq |G|(1-\delta),\]
	and similarly for $G_{\delta,\mathbf{y}}^{(q)}$. The final statement follows directly from the definitions, completing the proof.
\end{proof}

It turns out that the \textit{approximate radius} is the relevant quantity when we are converting the reward function oracles to lattice oracles, which is the objective of \cref{subsec:rewardoracleconversions}, whereas the \textit{effective radius} is the relevant quantity for calculating the value function from such a lattice oracle, which is the objective of \cref{subsec:valueoracleconstruction}. As we will see in \cref{subsec:valuefunctionestimation}, the fact that these two radii differ for a regular square lattice, is the deep reason behind the difference in query complexities to the reward oracles and the probability transition oracle, when $q > 1$.

\subsection{Reward oracle conversions}
\label{subsec:rewardoracleconversions}

In this subsection, we focus on the interconvertibility between oracles providing access to the reward function. All the proofs in this section are provided for the exact-depth setting, but they carry over to the cumulative-depth and path-independent setting word for word. The graph displayed in \cref{fig:oracleconversions} shows the conversions that we present in this section, with the corresponding overheads in query complexity. Other conversions from those shown in the figure do exist (i.e., from probability oracle to phase oracle and vice versa follows directly from \cref{lem:probabilitytophase,lem:phasetoprobability}), but these are not necessary in our construction of efficient quantum algorithms for multivariate Monte Carlo estimation.

\begin{figure}[h!]
	\centering
	\begin{tikzpicture}[xscale=4,yscale=3,oracle/.style={draw,rounded corners=.2em}]
	\node[oracle] (OR) at (-1,0) {$O_R$};
	\node[oracle] (UR) at (0,1) {$U_R$};
	\node[oracle] (DR) at (1,0) {$D_R$};
	\node[oracle] (LRG) at (0,0) {$L_{G,R}$};
	\draw[->] (OR) to node[below] {$\frac{r_G(\infty)}{r_G(q)}$} node[above] {\cref{lem:phasetolattice}} (LRG);
	\draw[->] (UR) to node[right] {$\sqrt{\frac{r_G(\infty)}{r_G(q)}}$} node[left] {\rotatebox{90}{\cref{lem:probabilitytolattice}}} (LRG);
	\draw[->] (DR) to node[below] {$\sqrt{\frac{d^{1-\frac1q}r_G(1)}{r_G(q)}}$} node[above] {\cref{lem:distributiontolattice}} (LRG);
	\end{tikzpicture}
	\caption{The oracle conversions that we present in \cref{subsec:rewardoracleconversions}. An arrow pointing from oracle type A to B means that an oracle of type B can be constructed given access to the reward function via oracle A, with a number of queries to A specified by the complexity labeling the edge. The complexities shown are in the limit of the factor shown to $\infty$ and the precision downwards to $0$, and they are all hiding polylogarithmic factors in the reciprocal of the precision with which the conversion is to be performed. The complexity shown at the conversion from $O_R$ to $L_{R,G}$ in addition also hides a polylogarithmic factor in the complexity displayed.}
	\label{fig:oracleconversions}
\end{figure}
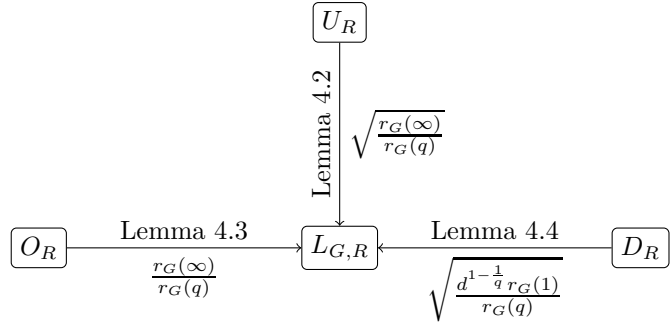

We emphasize that all results presented here hold for arbitrary sets $G \subseteq \R^d$, so in particular they also hold for trimmed sets $G_{\delta}^{(q)}$, where in the complexities all $G$'s are replaced by $G_{\delta}^{(q)}$. In particular, in \cref{subsec:valuefunctionestimation} we use these results on the intersection of two trimmed grids.

We start with showing how a probability oracle can be converted into a lattice oracle.

\begin{lemma}[Conversion from probability oracle to lattice oracle]
	\label{lem:probabilitytolattice}
	Let $d,T \in \N$, $\delta > 0$, $q \in [1,\infty]$ and $R : S^{T+1} \to \R^d$ be a reward function whose reward vectors are entry-wise non-negative and bounded by $R_{\max}$ in $\ell_q$-norm. Suppose that we have access to this reward function by means of a probability oracle $U_R$, as defined in \cref{eq:UR}. Let $G \subseteq \R^d$ be a finite set. Then, we can implement the operation $L_{G,R}$, as defined in \cref{eq:LR}, up to norm error $\delta$ with a number of calls to $U_R$ that scales as
	\[\O\left(\sqrt{\frac{r_G(\infty)}{r_G(q)}}\polylog\left(\frac{1}{\delta}\right)\right) \qquad \left(\frac{r_G(\infty)}{r_G(q)} \to \infty, \;\; \delta \downarrow 0\right).\]
	The same statement holds in the cumulative-depth and path-independent cases, if we replace all $R$'s by $R^{(t)}$'s and $R_S$'s, respectively.
\end{lemma}

\begin{proof}
	Let $\mathbf{x} \in G$ and $\tau \in S^{T+1}$. We let $\mathbf{x}^{(+)}, \mathbf{x}^{(-)} \in \R^d$ be the vectors that contain all the positive resp.\ negative entries of $\mathbf{x}$, and $0$ otherwise. For both $\mathbf{x}^{(+)}$ and $\mathbf{x}^{(-)}$ separately, we follow the idea from \cite{vanApeldoorn20}, page 24. We use a vector register, a path register, a register containing the states $\ket{j}$ for $j \in [d]_0$, and a single qubit. We start with the state
	\[\ket{\mathbf{x}}\ket{\tau}\ket{0}\ket{0},\]
	which we can turn, without making any queries, into
	\[\ket{\mathbf{x}}\ket{\tau}\left(\sum_{j=1}^d \sqrt{\frac{x_j^{(+)}}{r_G(\infty)}}\ket{j}\ket{0} + \sqrt{1-\frac{\norm{\mathbf{x}^{(+)}}_1}{r_G(\infty)}}\ket{0}\ket{0}\right).\]
	Next, we apply the reward oracle $U_R$ to the last three registers, controlled on the second to last register not being in state $\ket{0}$, to obtain the state
	\[\ket{\mathbf{x}}\ket{\tau}\left(\sum_{j=1}^d \sqrt{\frac{x_j^{(+)}R(\tau)_j}{r_G(\infty)R_{\max}}}\ket{j}\ket{1} + \sqrt{\frac{x_j^{(+)}}{r_G(\infty)} \cdot \left(1 - \frac{R(\tau)_j}{R_{\max}}\right)}\ket{j}\ket{0} + \sqrt{1-\frac{\norm{\mathbf{x}^{(+)}}_1}{r_G(\infty)}}\ket{0}\ket{0}\right).\]
	The overlap between the resulting state and the subspace spanned by the states of the form $\ket{\mathbf{x}}\ket{\tau}\ket{\cdot}\ket{1}$ equals
	\[\sqrt{\sum_{j=1}^d \frac{x_j^{(+)}R(\tau)_j}{r_G(\infty)R_{\max}}} = \sqrt{\frac{\left(\mathbf{x}^{(+)}\right)^TR(\tau)}{r_G(\infty)R_{\max}}} \leq \sqrt{\frac{\norm{\mathbf{x}^{(+)}}_{\frac{1}{1-\frac1q}} \cdot \norm{R(\tau)}_q}{r_G(\infty)R_{\max}}}  \leq \sqrt{\frac{r_G(q)}{r_G(\infty)}},\]
	where we used H\"older's inequality in the first inequality, and the definition of $r_G$ from \cref{eq:rG} in the second one.
	
	We now use the entire above operation as a black box, and we amplify the part of the resulting state that is of the form $\ket{\mathbf{x}}\ket{\tau}\ket{\cdot}\ket{1}$, with a factor of $\sqrt{r_G(\infty)/(2r_G(q))}$. This can be realized by interpreting our black box as a block-encoding between the one-dimensional subspace spanned by the initial state, $\ket{\mathbf{x}}\ket{\tau}\ket{0}\ket{0}$, and the subspace of all states of the form $\ket{\mathbf{x}}\ket{\tau}\ket{\cdot}\ket{1}$. It has only one singular value, namely the overlap, which can be amplified by applying the  function $x \mapsto x\sqrt{r_G(\infty)/(2r_G(q))}$ to it. According to \cite{GSLW18}, Corollary 66, a polynomial approximating this function can be constructed up to precision $\delta' > 0$ with degree $\O(\sqrt{r_G(\infty)/r_G(q)}\log(1/\delta'))$, which equals the number of calls to the black box that we need to make.
	
	We now have a procedure that for all $\mathbf{x} \in G$ constructs a state whose overlap with states of the form $\ket{\mathbf{x}}\ket{\tau}\ket{\cdot}\ket{1}$ is equal to
	\[\sqrt{\frac{r_G(\infty)}{2r_G(q)}} \cdot \sqrt{\frac{\left(\mathbf{x}^{(+)}\right)^TR(\tau)}{r_G(\infty)R_{\max}}} = \sqrt{\frac{\left(\mathbf{x}^{(+)}\right)^TR(\tau)}{2r_G(q)R_{\max}}}.\]
	This operation, if we could implement it perfectly, can be turned into a phase oracle with operator norm precision $\delta/8$ with a total of $\O(\polylog(1/\delta))$ invocations, following \cref{lem:probabilitytophase}. Hence, if we choose $\delta' = \Theta(\delta/\polylog(1/\delta))$, we can implement such a phase oracle from $U_R$ with a total error of only $\delta/4$ in operator norm, and with a number of calls to $U_R$ that scales as $\O(\sqrt{r_G(\infty)/r_G(q)}\polylog(1/\delta))$. The resulting operation we end up implementing is
	\[\ket{\mathbf{x}}\ket{\tau} \mapsto e^{i\frac{\left(\mathbf{x}^{(+)}\right)^TR(\tau)}{4r_G(q)R_{\max}}} \ket{\mathbf{x}}\ket{\tau}.\]
	We can run the same operation with $-\mathbf{x}^{(-)}$ in reverse, again with operator norm error in $\delta/4$, which implements the mapping
	\[\ket{\mathbf{x}}\ket{\tau} \mapsto e^{i\frac{\left(\mathbf{x}^{(-)}\right)^T R(\tau)}{4r_G(q)R_{\max}}}\ket{\mathbf{x}}\ket{\tau}.\]
	Running these two operations consecutively, and repeating everything twice, implements the desired mapping $L_{G,R}$, up to norm error $\delta$.
	
	The conversions in the cumulative-depth and path-independent cases follow analogous arguments.
\end{proof}

Next, we show how the phase oracle can be used to construct a lattice oracle. We use the conversion from the probability oracle as a subroutine.

\begin{lemma}[Conversion from phase oracle to lattice oracle]
	\label{lem:phasetolattice}
	Let $d,T \in \N$, $\delta > 0$, $q \in [1,\infty]$, and $R : S^{T+1} \to \R^d$ be a reward function whose reward vectors are bounded by $R_{\max}$ in $\ell_q$-norm. Suppose that we have access to this reward function by means of a phase oracle $O_R$, as defined in \cref{eq:OR}. Let $G \subseteq \R^d$ be a finite set. Then, we can implement the lattice oracle $L_{G,R}$, as defined in \cref{eq:LR}, up to operator norm $\delta$ with a number of calls to $O_R$ that scales as
	\[\O\left(\frac{r_G(\infty)}{r_G(q)}\polylog\left(\frac{r_G(\infty)}{r_G(q)\delta}\right)\right), \qquad \left(\frac{r_G(\infty)}{r_G(q)} \to \infty, \;\; \delta \downarrow 0\right).\]
	The same statement holds in the cumulative-depth and path-independent cases by replacing the $R$'s by $R^{(t)}$'s and $R_S$'s, respectively.
\end{lemma}

\begin{proof}
	Let $\mathbf{x} \in G$ and $\tau \in S^{T+1}$. Recall that we can implement a fractional phase oracle $O_R^{1/2}$ up to precision $\delta' > 0$, with just $\O(\polylog(1/\delta'))$ calls to $O_R$, according to \cite{GSLW18}, Corollary 72. By adding a power of the $Z$-gate to the control qubit if we call it in a controlled manner, we can henceforth build the operation described by
	\[\ket{\tau}\ket{j} \mapsto e^{i\left(\frac12 + \frac{R(\tau)_j}{4R_{\max}}\right)}\ket{\tau}\ket{j},\]
	with precision $\delta' > 0$.
	
	Now, the phase shift incurred by this operation is always contained in the interval $[1/4,3/4]$, meaning that we can turn this operation into a probability oracle via \cref{lem:phasetoprobability}. Specifically, with $\O(\polylog(1/\delta''))$ calls to the previous operation, we can construct the mapping
	\[\ket{\tau}\ket{j}\ket{0} \mapsto \ket{\tau}\ket{j}\left(\sqrt{\frac12+\frac{R(\tau)_j}{4R_{\max}}}\ket{1}+ \sqrt{\frac12 - \frac{R(\tau)_j}{4R_{\max}}}\ket{0}\right),\]
	with precision $\delta'' > 0$.
	
	This is a probability oracle to the slightly modified reward function $\overline{R} : S^{T+1} \to \R^d$, defined as
	\[\overline{R}(\tau) = R(\tau) + 2R_{\max}\mathbf{1},\]
	where $\overline{R}_{\max} = 4R_{\max}$. Note that this reward function $\overline{R}$ is bounded in $\ell_{\infty}$-norm by $\overline{R}_{\max}$, but not in any $\ell_q$-norm with $q < \infty$, since we have shfited it away from the origin in all directions. According to \cref{lem:probabilitytolattice}, we can turn this probability oracle into a lattice oracle $L_{G,\overline{R}}$ with overhead logarithmic in the precision, i.e., we can implement the operation
	\[\ket{\mathbf{x}}\ket{\tau} \mapsto e^{i\frac{\mathbf{x}^T\overline{R}(\tau)}{2r_G(\infty)\overline{R}_{\max}}}\ket{\mathbf{x}}\ket{\tau} = e^{i\frac{\mathbf{x}^TR(\tau)}{8r_G(\infty)R_{\max}} + i\frac{\mathbf{x}^T\mathbf{1}}{4r_G(\infty)}}\ket{\mathbf{x}}\ket{\tau},\]
	up to precision $\delta/8 \cdot r_G(q) / r_G(\infty)$, using a number of calls to $U_{\overline{R}}$ that scales as $\O(\polylog(r_G(\infty)/(r_G(q)\delta)))$.
	
	We can remove the global phase again by applying a power of the $Z$-gate to the control qubit, and then we can run this operation a total of $4r_G(\infty)/r_G(q)$ times, to implement the mapping
	\[\ket{\mathbf{x}}\ket{\tau} \mapsto e^{i\frac{\mathbf{x}^TR(\tau)}{2r_G(q)R_{\max}}}\ket{\mathbf{x}}\ket{\tau},\]
	up to precision $\delta/2$. It suffices to choose $\delta' = \delta'' = \Theta((\delta r_G(q)/r_G(\infty)) / \polylog(r_G(\infty)/(r_G(q)\delta)))$ in order to ensure that we only lose another $\delta/2$ in all the times we need to construct $U_{\overline{R}}$, amounting to a total norm error of at most $\delta$.
	
	The total number of calls to $O_R$ scales as
	\[\O\left(\frac{r_G(\infty)}{r_G(q)}\polylog\left(\frac{r_G(\infty)}{r_G(q)\delta}\right) \cdot \polylog\left(\frac{1}{\delta''}\right) \cdot \polylog\left(\frac{1}{\delta'}\right)\right) = \O\left(\frac{r_G(\infty)}{r_G(q)}\polylog\left(\frac{r_G(\infty)}{r_G(q)\delta}\right)\right),\]
	completing the proof for the exact-depth case.
	
	The argument in the cumulative-depth and path-independent cases is completely analogous.
\end{proof}

Finally, we also show how one can convert a distribution oracle to a lattice oracle.

\begin{lemma}[Conversion from distribution oracle to lattice oracle]
	\label{lem:distributiontolattice}
	Let $d,T \in \N$, $\delta > 0$, $q \in [1,\infty]$, and $R : S^{T+1} \to \R^d$ be a reward function whose reward vectors are entry-wise non-negative and bounded by $R_{\max}$ in $\ell_q$-norm. Suppose that we have access to this reward function by means of a distribution oracle $D_R$, as defined in \cref{eq:DR}. Let $G \subseteq \R^d$ be a finite set. Then we can implement the lattice oracle $L_{G,R}$, as defined in \cref{eq:LR}, up to operator norm error $\delta$ with a number of calls to $D_R$ that scales as
	\[\O\left(\sqrt{\frac{d^{1-\frac1q}r_G(1)}{r_G(q)}}\polylog\left(\frac{1}{\delta}\right)\right), \qquad \left(d, \frac{r_G(1)}{r_G(q)} \to \infty, \;\; \delta \downarrow 0\right).\]
	The same statement holds in the cumulative-depth and path-independent cases, if we replace all $R$'s by $R^{(t)}$'s and $R_S$'s, respectively.
\end{lemma}

\begin{proof}
	Let $\mathbf{x} \in G$ and $\tau \in S^{T+1}$. Again, we let $\mathbf{x}^{(+)} \in \R^d$ be the vector that contains the positive entries of $\mathbf{x}$ and is $0$ otherwise, and similarly let $\mathbf{x}^{(-)} \in \R^d$ be the vector that only contains the negative entries of $\mathbf{x}$, and is $0$ in all the other coordinates. The approach taken here now follows the technique displayed in \cite{vanApeldoorn20}, at the top of page 24. We act on a vector register, a path register, a register that contains the states $\ket{j}$ for all $j \in [d]_0$, and a single extra qubit. We start in the state
	\[\ket{\mathbf{x}}\ket{\tau}\ket{0}\ket{0}.\]
	First, we call the distribution oracle on the path and the coordinate register, to obtain
	\[\ket{\mathbf{x}}\ket{\tau}\left(\sum_{j=1}^d \sqrt{\frac{R(\tau)_j}{d^{1-\frac1q}R_{\max}}}\ket{j} + \sqrt{1-\frac{\norm{R(\tau)}_1}{d^{1-\frac1q}R_{\max}}}\ket{0}\right)\ket{0}.\]
	Next, without making any queries, we can turn this state into
	\[\ket{\mathbf{x}}\ket{\tau}\left(\sum_{j=1}^d \sqrt{\frac{R(\tau)_j}{d^{1-\frac1q}R_{\max}}}\ket{j} \left(\sqrt{\frac{x_j^{(+)}}{r_G(1)}}\ket{1} + \sqrt{1-\frac{x_j^{(+)}}{r_G(1)}}\ket{0}\right) + \sqrt{1-\frac{\norm{R(\tau)}_1}{d^{1-\frac1q}R_{\max}}}\ket{0}\ket{0}\right).\]
	The total overlap with the subspace spanned by the states of the form $\ket{\mathbf{x}}\ket{\tau}\ket{\cdot}\ket{1}$ now equals
	\[\sqrt{\sum_{j=1}^d \frac{x_j^{(+)}R(\tau)_j}{d^{1-\frac1q}R_{\max}r_G(1)}} = \sqrt{\frac{\left(\mathbf{x}^{(+)}\right)^TR(\tau)}{d^{1-\frac1q}R_{\max}r_G(1)}} \leq \sqrt{\frac{\norm{\mathbf{x}^{(+)}}_{\frac{1}{1-\frac1q}}\norm{R(\tau)}_q}{d^{1-\frac1q}R_{\max}r_G(1)}} \leq \sqrt{\frac{r_G(q)}{d^{1-\frac1q}r_G(1)}}.\]
	Hence, using a similar argument as in \cref{lem:probabilitytolattice}, Corollary 66 from \cite{GSLW18} implies that we can multiply this overlap with a factor of $\sqrt{d^{1-1/q}r_G(1)/(2r_G(q))}$, and we can implement the resulting operation up to norm error $\delta' > 0$ using $\O(\sqrt{d^{1-1/q}r_G(1)/r_G(q)}\polylog(1/\delta'))$ calls to the previous operations. The resulting operation produces a state that has overlap with the subspace spanned by states of the form $\ket{\mathbf{x}}\ket{\tau}\ket{\cdot}\ket{1}$ equal to
	\[\sqrt{\frac{d^{1-\frac1q}r_G(1)}{2r_G(q)}} \cdot \sqrt{\frac{\left(\mathbf{x}^{(+)}\right)^TR(\tau)}{d^{1-\frac1q}R_{\max}r_G(1)}} = \sqrt{\frac{\left(\mathbf{x}^{(+)}\right)^TR(\tau)}{2r_G(q)R_{\max}}}.\]
	Thus, the resulting operation can be turned into a phase oracle with operator norm error $\delta/8$ with $\O(\polylog(1/\delta))$ calls, using \cref{lem:probabilitytophase}. The resulting operation is
	\[\ket{\mathbf{x}}\ket{\tau} \mapsto e^{i\frac{\left(\mathbf{x}^{(+)}\right)^TR(\tau)}{4r_G(q)R_{\max}}}\ket{\mathbf{x}}\ket{\tau},\]
	which we can implement with total error $\delta/4$ if we choose $\delta' = \Theta(\delta/\polylog(1/\delta))$. The total number of calls to $D_R$ then scales as
	\[\O\left(\sqrt{\frac{d^{1-\frac1q}r_G(1)}{r_G(q)}}\polylog\left(\frac{1}{\delta}\right)\right).\]
	We can perform the same construction in reverse with the negative entries of $\mathbf{x}$, i.e., with similar cost and error we can implement
	\[\ket{\mathbf{x}}\ket{\tau} \mapsto e^{i\frac{\left(\mathbf{x}^{(-)}\right)^TR(\tau)}{4r_G(q)R_{\max}}}\ket{\mathbf{x}}\ket{\tau}.\]
	Running both consecutively, and repeating the whole construction twice, we implement $L_{G,R}$ with operator norm error $\delta$. This completes the proof in the exact-depth case.
	
	The proofs for the cumulative-depth and path-independent cases follow analogously.
\end{proof}

\subsection{Computation of the value function}
\label{subsec:valueoracleconstruction}

In this section, we use the lattice oracle for the reward function, for which we gave constructions in the previous section, to construct an object that acts almost as a lattice oracle for the value function. More precisely, let $\delta > 0$, and let the operation $\overline{L}_{G,V,\delta}$ be defined as
\begin{equation}
	\label{eq:LGV}
	\overline{L}_{G,V,\delta} : \frac{1}{\sqrt{|G|}} \sum_{\mathbf{x} \in G} e^{i\varphi(\mathbf{x},s_0)} \ket{\mathbf{x}} \ket{s_0} \mapsto \frac{1}{\sqrt{|G|}} \sum_{\mathbf{x} \in G} e^{i\varphi(\mathbf{x},s_0) + i\frac{\mathbf{x}^TV(s_0)}{2\overline{r}_{G,\delta}(q)}} \ket{\mathbf{x}}\ket{s_0},
\end{equation}
where $\varphi(\mathbf{x},s_0) \in \R$ for all $\mathbf{x} \in G$ and $s_0 \in S$, the state space of the Markov reward process on which the value function is defined. The construction we provide implements this operation $\overline{L}_{G,V,\delta}$, up to some small operator norm error.

The difference between this operation $\overline{L}_{G,V,\delta}$ and an approximate implementation of a regular lattice oracle evaluating the value function, in terms of the definition of lattice oracles given in \cref{def:oracles}, is that here we only require that it acts approximately well on a superposition over all states $\ket{\mathbf{x}}$ with equal weight, where $\mathbf{x} \in G$. In particular, we do not demand that it acts approximately well on any given $\ket{\mathbf{x}}$ individually, with $\mathbf{x} \in G$ -- loosely speaking this means that it is okay if we screw up big time on some of the $\mathbf{x}$'s, as long as this happens only for a small fraction of the points in $G$ and is compensated by performing exceptionally well on the others.

This distinction between $\overline{L}_{G,V,\delta}$ and a regular lattice oracle evaluating the value function, has the high-level implication that it is possible to implement it using a lattice oracle that evaluates the reward function on all but a few points of $G$. In particular, we construct the operation $\overline{L}_{G,V,\delta}$ given access to a lattice oracle evaluating the reward function on a slightly smaller set $\overline{G}$, with $|\overline{G}| \geq (1-\delta)|G|$ for some specific $\delta > 0$.

The constructions of $\overline{L}_{G,V,\delta}$ presented in this subsection depend on the setting, that is, exact-depth, cumulative-depth, or path-independent. We first present the construction in the exact-depth case, in \cref{lem:valuefunctionoracleexactdepth}, and subsequently use this as a subroutine in the cumulative-depth and path-independent case, in \cref{lem:valuefunctionoraclecumulativedepth,lem:valuefunctionoraclepathindependent}.

\begin{lemma}[Computation of the value function (exact-depth)]
	\label{lem:valuefunctionoracleexactdepth}
	Let $d,T \in \N$, $R_{\max} > 0$, $q \in [1,\infty]$, $\delta > 0$, $S$ a state space, $P : S \times S \to [0,1]$ a probability transition matrix and $R : S^{T+1} \to \R^d$ be a depth-$T$ reward function, bounded by $R_{\max}$ in $\ell_q$-norm, and let $G \subseteq \R^d$ be a finite set. Then there exists a function $\delta' = \Theta(\delta^2/\polylog(1/\delta))$, in the limit where $\delta \downarrow 0$, such that the following statement holds. Suppose that $\overline{G} \subseteq G$ such that $|\overline{G}| \geq (1-\delta')|G|$, and that we have access to the reward function by means of a lattice oracle $L_{\overline{G},R}$, as defined in \cref{eq:LR}. Then we can implement the operation $\overline{L}_{G,V,\delta'}$, defined in \cref{eq:LGV}, up to operator norm $\delta$ using a number of calls to $D_P$ and $L_{\overline{G},R}$ that scales as
	\[\O\left(T\polylog\left(\frac{1}{\delta}\right)\right), \qquad \text{and} \qquad \O\left( \frac{r_{\overline{G}}(q)}{\overline{r}_{G,\delta'}(q)}\polylog\left(\frac{1}{\delta}\right)\right), \qquad \left(\frac{r_{\overline{G}}(q)}{\overline{r}_{G,\delta'}(q)} \to \infty, \;\; \delta \downarrow 0\right),\]
	respectively.
\end{lemma}

\begin{proof}
	In this proof, we are going to implement many operations that we will use as black boxes in subsequent steps of the construction. We do an independent analysis on the implementation error in each of these building blocks, and combine the entire error analysis at the end of the proof.
	
	We are going to act on a vector register, $T+1$ state registers, and a single extra qubit, which start in the state
	\[\ket{\psi} = \frac{1}{\sqrt{|G|}} \sum_{\mathbf{x} \in G} e^{i\varphi(\mathbf{x},s_0)} \ket{\mathbf{x}}\ket{s_0}\ket{0}^{\otimes T}\ket{0}.\]
	With $T$ consecutive calls to $D_P$, each acting on two adjacent state registers starting from the left and moving to the right one state register at the time, we prepare the probability distribution over the paths of length $T$ starting at $s_0$ in the $T+1$ state registers,
	\begin{equation}
		\label{eq:pathtree}
		\frac{1}{\sqrt{|G|}} \sum_{\mathbf{x} \in G} e^{i\varphi(\mathbf{x},s_0)} \ket{\mathbf{x}}\left(\sum_{\tau \in \{s_0\} \times S^T} \sqrt{\P(\tau)}\ket{\tau}\right)\ket{0}.
	\end{equation}
	We remark that this operation is block-diagonal with respect to the subspaces spanned by states of the form $\ket{\mathbf{x}}\ket{s_0}\ket{\cdot}^{\otimes T}\ket{\cdot}$. Next, for all $\tau \in S^{T+1}$ and $\mathbf{x} \in \overline{G}$, recall from the definition of lattice oracles, \cref{eq:LR}, that $L_{\overline{G},R}$ acts as
	\[\ket{\mathbf{x}}\ket{\tau} \mapsto e^{i\frac{\mathbf{x}^TR(\tau)}{2r_{\overline{G}}(q)R_{\max}}}\ket{\mathbf{x}}\ket{\tau}.\]
	Using $\O(r_{\overline{G}}(q)/\overline{r}_{G,\delta'}(q) + \polylog(1/\delta''))$ calls to $L_{\overline{G},R}$, i.e., many regular calls and one fractional phase oracle call as can be constructed via \cite{GSLW18}, Corollary 72, we can convert this into an operation that for all $\mathbf{x} \in \overline{G}$ and $\tau \in S^{T+1}$ acts as
	\begin{equation}
		\label{eq:scaledLR}
		\ket{\mathbf{x}}\ket{\tau} \mapsto e^{i\frac{\mathbf{x}^TR(\tau)}{2\overline{r}_{G,\delta'}(q)R_{\max}}}\ket{\mathbf{x}}\ket{\tau},
	\end{equation}
	up to precision $\delta'' > 0$. It is immediate that this operation is block-diagonal with respect to the subspaces spanned by states of the form $\ket{\mathbf{x}}\ket{s_0}\ket{\cdot}^{\otimes T}\ket{\cdot}$ too. On top of that, for all $\tau \in S^{T+1}$, we have that for all $\mathbf{x} \in G_{\delta',R(\tau)}^{(q)}$,
	\[\left|\frac{\mathbf{x}^TR(\tau)}{2\overline{r}_{G,\delta'}(q)R_{\max}}\right| \leq \frac{\overline{r}_{G,\delta'}(q)\norm{R(\tau)}_q}{2\overline{r}_{G,\delta'}(q)R_{\max}} \leq \frac12,\]
	by the definition of the effective radius. Hence, using the construction from \cref{lem:phasetoprobability}, with $\O(\polylog(1/\delta'''))$ calls, we can turn the operation from \cref{eq:scaledLR} into one that acts for all $\mathbf{x} \in \overline{G} \cap G_{\delta',R(\tau)}^{(q)}$ and $\tau \in S^{T+1}$ as
	\begin{equation}
		\label{eq:Utau}
		\ket{\mathbf{x}}\ket{\tau}\ket{0} \mapsto \ket{\mathbf{x}}\ket{\tau}\left(\sqrt{\frac12 + \frac{\mathbf{x}^TR(\tau)}{4\overline{r}_{G,\delta'}(q)R_{\max}}}\ket{1} + \sqrt{\frac12 - \frac{\mathbf{x}^TR(\tau)}{4\overline{r}_{G,\delta'}(q)R_{\max}}}\ket{0}\right),
	\end{equation}
	up to precision $\delta''' > 0$. Moreover, since the construction in \cref{lem:phasetoprobability} only acts on the latter registers, the resulting operation remains block-diagonal w.r.t.\ the subspaces spanned by states of the form $\ket{\mathbf{x}}\ket{s_0}\ket{\cdot}^{\otimes T}\ket{\cdot}$. When we apply this operation to the state from \cref{eq:pathtree}, we approximately obtain the state
	\[\ket{\phi} = \frac{1}{\sqrt{|G|}} \sum_{\mathbf{x} \in G} e^{i\varphi(\mathbf{x},s_0)} \ket{\mathbf{x}}\left(\sum_{\tau \in \{s_0\} \times S^T} \sqrt{\P(\tau)}\ket{\tau}\left(\sqrt{\frac12 + \frac{\mathbf{x}^T R(\tau)}{4\overline{r}_{G,\delta'}(q)R_{\max}}}\ket{1} + \sqrt{\frac12 - \frac{\mathbf{x}^T R(\tau)}{4\overline{r}_{G,\delta'}(q)R_{\max}}}\ket{0}\right)\right).\]
	Let $U$ be the operation that constructs the above state $\ket{\phi}$ from the initial state $\ket{\psi}$ perfectly, and let $\widetilde{U}$ be the complete unitary operation we described above, where the operation in \cref{eq:Utau} is implemented perfectly. Then, the norm squared error we make can be bounded by
	\begin{align*}
		\norm{\widetilde{U}\ket{\psi} - \ket{\phi}}^2 &= \sum_{\mathbf{x} \in G} \sum_{\tau \in \{s_0\} \times S^T} \norm{(\bra{\mathbf{x}}\bra{\tau} \otimes I)(\widetilde{U}\ket{\psi} - \ket{\phi})}^2 \leq \sum_{\tau \in \{s_0\} \times S^T} \sum_{\mathbf{x} \in G \setminus (\overline{G} \cap G_{\delta',R(\tau)}^{(q)})} \frac{\P(\tau)}{|G|} \cdot 4 \\
		&\leq 4\sum_{\tau \in \{s_0\} \times S^T} \P(\tau) \cdot \frac{\left|G \setminus \left(\overline{G} \cap G_{\delta',R(\tau)}^{(q)}\right)\right|}{|G|} \leq 4\sum_{\tau \in \{s_0\} \times S^T} \P(\tau) \cdot \frac{\left|G \setminus \overline{G}\right| + \left|G \setminus G_{\delta',R(\tau)}^{(q)}\right|}{|G|} \\
		&\leq 4\sum_{\tau \in \{s_0\} \times S^T} \P(\tau) \cdot 2\delta' = 8\delta'.
	\end{align*}
	We now want to convert $U$ into parallel phase oracles, evaluating the overlap of every individual branch $\ket{\mathbf{x}}\ket{s_0}$ with the states of the form $\ket{\mathbf{x}}\ket{s_0}\ket{\cdot}^{\otimes T}\ket{1}$ while retaining the relative phases. For each such branch, this overlap equals, after renormalization of the branch,
	\[\sqrt{\sum_{\tau \in \{s_0\} \times S^T} \P(\tau) \cdot \left(\frac12 + \frac{\mathbf{x}^TR(\tau)}{4\overline{r}_{G,\delta'}(q)R_{\max}}\right)} = \sqrt{\frac12 + \frac{\mathbf{x}^T\E\left[R(\tau)\right]}{4\overline{r}_{G,\delta'}(q)R_{\max}}} = \sqrt{\frac12 + \frac{\mathbf{x}^T V(s_0)}{4\overline{r}_{G,\delta'}(q)R_{\max}}}.\]
	The idea is to run the construction presented in \cref{lem:probabilitytophase}, but with the probability oracle replaced by $U$. With a number of calls that scales as $\O(\polylog(1/\delta))$, we can implement the operation
	\begin{equation}
		\label{eq:LGVdelta}
		\frac{1}{\sqrt{|G|}} \sum_{\mathbf{x} \in G} e^{i\varphi(\mathbf{x},s_0)} \ket{\mathbf{x}}\ket{s_0} \mapsto \frac{1}{\sqrt{|G|}} \sum_{\mathbf{x} \in G} e^{i\varphi(\mathbf{x},s_0) + i\left(\frac12 + \frac{\mathbf{x}^TV(s_0)}{4\overline{r}_{G,\delta'}(q)R_{\max}}\right)}\ket{\mathbf{x}}\ket{s_0},
	\end{equation}
	up to precision $\delta/4$. We crucially use here that $U$ acts block-diagonally on the subspaces of states of the form $\ket{\mathbf{x}}\ket{s_0}\ket{\cdot}^{\otimes T}\ket{\cdot}$. The constant global phase $e^{i/2}$ can be removed with a single qubit gate on the control qubit if we call this operation in a controlled manner. Finally, we run the entire construction above twice to multiply the function value by $2$ so that we recover the required multiplicative factor displayed in \cref{eq:LGV}. Hence, the resulting norm error that we obtain in this step is $\delta/2$.
	
	Next, we turn to the error analysis. In the final step, using a perfect implementation of $U$ implies that we pick up an operator norm error of $\delta/2$. We make $\O(\polylog(1/\delta))$ calls to $U$, so we need to argue that we lose only $\O(\delta/\polylog(1/\delta))$ in norm error for every call to $U$.
	
	Crucially, we have argued before that all operations comprising $\widetilde{U}$ act block-diagonally on subspaces spanned by states of the form $\ket{\mathbf{x}}\ket{s_0}\ket{\cdot}^{\otimes T}\ket{\cdot}$. Similarly, all other operations in \cref{lem:probabilitytophase} only act on the latter registers, leaving those containing $\ket{\mathbf{x}}$ and $\ket{s_0}$ alone as well. Thus, throughout the construction of the operation displayed in \cref{eq:LGVdelta}, the weight on each of the branches with $\ket{\mathbf{x}}$ remains uniform, meaning that on every call to $\widetilde{U}$, we only pick up an error of $\sqrt{8\delta'}$. Hence, we can indeed choose $\delta' = \Theta(\delta^2/\polylog(1/\delta))$, in order to make sure that in this step we cumulatively obtain at most $\delta/4$ in norm error.
	
	Now, it remains to ensure we accumulate at most $\delta/4$ in the remaining operations. To that end, we observe that we can indeed choose $\delta'' = \Theta(\delta/\polylog(1/(\delta\delta''')))$ and $\delta''' = \Theta(\delta/\polylog(1/\delta))$, such that the accumulated error in the construction of both \cref{eq:scaledLR,eq:Utau} is at most $\delta/8$ each.
	
	Finally, we check the query complexity claims. Observe that we call $U$ a total of $\O(\polylog(1/\delta))$ times. $U$ itself is implemented using $T$ calls to $D_P$, proving the claimed query complexity to $D_P$, and a number of calls to the operation in \cref{eq:scaledLR} that scales as $\O(\polylog(1/\delta'''))$. This operation in turn is implemented with a number of calls to $L_{\overline{G},R}$ that satisfies $\O(r_{\overline{G}}(q)/\overline{r}_{G,\delta'}(q) + \polylog(1/\delta''))$. Thus, the total number of calls to $L_{\overline{G},R}$ required is
	\[\O\left(\polylog\left(\frac{1}{\delta}\right) \cdot \polylog\left(\frac{1}{\delta'''}\right) \cdot \left(\frac{r_{\overline{G}}(q)}{\overline{r}_{G,\delta'}(q)} + \polylog\left(\frac{1}{\delta''}\right)\right)\right) = \O\left(\frac{r_{\overline{G}}(q)}{\overline{r}_{G,\delta'}(q)} \cdot \polylog\left(\frac{1}{\delta}\right)\right),\]
	where we used that $\delta''$ and $\delta'''$ can be chosen as $\delta$ up to polylogarithmic factors.
\end{proof}

There is a slightly annoying subtlety in the statement of the previous lemma, which is that an explicit formula for $\delta'$ is not given, rather we give an existence result of $\delta' = \Theta(\delta^2 / \polylog(1/\delta))$. We remark here that in principle it is possible to relate these two parameters more concretely and give a direct formula to compute $\delta'$ from $\delta$, but this would require a more careful analysis of the oracle conversion result, \cref{lem:probabilitytophase}, beyond big-$\O$-notation. As far as we are aware, such a result is not available in the current literature. In the end, we will use such existence arguments in the proof of \cref{thm:mvmc-alg} anyway, so there is no end-to-end qualitative improvement to be gained by figuring out the direct relation between $\delta'$ and $\delta$.

Next, we show how we can perform a similar construction in the cumulative-reward case. If $\gamma < 1$, then the high-level idea is to truncate the summation in cumulative-depth value function, and treat the remaining terms as individual exact-depth value functions. Up to polylogarithmic factors, it turns out to be sufficient to do this truncation at $T^*$, defined in \cref{eq:Tstar}.

\begin{lemma}[Computation of the value function (cumulative-depth)]
	\label{lem:valuefunctionoraclecumulativedepth}
	Let $d,T \in \N$, $R_{\max} > 0$, $q \in [1,\infty]$, $0 < \delta < 2$, and $\gamma \in [0,1]$ such that if $T = \infty$, then $\gamma < 1$. Let $S$ a state space, $P : S \times S \to [0,1]$ a probability transition matrix, for all integer $t$ such that $0 \leq t \leq T$, let $R^{(t)} : S^{t+1} \to \R^d$ be a depth-$t$ reward function, bounded by $R_{\max}$ in $\ell_q$-norm, and let $G \subseteq \R^d$ be a finite set. Then there exists a function $\delta' = \Theta(\delta^2/\polylog(1/\delta))$ such that the following holds. Suppose that $\overline{G} \subseteq G$ with $|\overline{G}| \geq (1-\delta')|G|$, and that we have access to these reward functions by means of lattice oracles $L_{\overline{G},R^{(t)}}$, as defined in \cref{eq:LRt}. Then, we can implement the operation $\overline{L}_{G,V,\delta'}$, as defined in \cref{eq:LGV}, up to operator norm $\delta$ with a number of calls to $D_P$ and $L_{\overline{G},R^{(t)}}$ that scale as
	\[\O\left((T^*)^2\polylog\left(\frac{T^*}{\delta}\right)\right), \qquad \text{and} \qquad \O\left(T^*\frac{r_{\overline{G}}(q)}{\overline{r}_{G,\delta'}(q)}\polylog\left(\frac{T^*}{\delta}\right)\right), \qquad \left(T^*, \frac{r_{\overline{G}}(q)}{\overline{r}_{G,\delta'}(q)} \to \infty, \;\; \delta \downarrow 0\right),\]
	respectively.
\end{lemma}

\begin{proof}
	Let
	\begin{equation}
		\label{eq:Tdelta}
		T_{\delta} = \min\left\{T, \left\lceil T^* \cdot \ln\frac{2}{\delta(1-\gamma)}\right\rceil \right\} = \O\left(T^*\polylog\left(\frac{T^*}{\delta}\right)\right).
	\end{equation}
	Let $s_0 \in S$ and observe that
	\[V(s_0) = \underset{\tau \sim P(T;s_0)}{\E}\left[\sum_{t=0}^T \gamma^t R^{(t)}(s_0, \dots, s_t)\right] = \sum_{t=0}^T \gamma^t \underset{\tau \sim P(t;s_0)}{\E}\left[R^{(t)}(\tau)\right].\]
	Suppose that $T_{\delta} > T$. Since all of the expectations in the right-most expression are bounded by $R_{\max}$ in $\ell_q$-norm, we find that
	\[\norm{V(s_0) - \sum_{t=0}^{T_{\delta}} \gamma^t \underset{\tau \sim P(t;s_0)}{\E}\left[R^{(t)}(\tau)\right]}_q \leq \sum_{t=T_{\delta}+1}^{\infty} \gamma^t \norm{\underset{\tau \sim P(t;s_0)}{\E}\left[R^{(t)}(\tau)\right]}_q \leq \sum_{t=T_{\delta}+1}^{\infty} \gamma^tR_{\max} = \gamma^{T_{\delta}+1}\frac{R_{\max}}{1-\gamma}.\]
	Since $T_{\delta} > T$ by assumption, we also find that $T_{\delta}+1 > \ln(2/(\delta(1-\gamma)))/(1-\gamma)$, and hence
	\begin{align*}
		\gamma^{T_{\delta}+1} &\leq \gamma^{\frac{\ln\frac{2}{\delta(1-\gamma)}}{1-\gamma}} = \exp\left(\frac{\ln(\gamma)}{1-\gamma} \cdot \ln\frac{2}{\delta(1-\gamma)}\right) = \exp\left(\frac{\ln(1 - (1-\gamma))}{1-\gamma} \cdot \ln\frac{2}{\delta(1-\gamma)}\right) \\
		&\leq \exp\left(-\frac{1-\gamma}{1-\gamma} \cdot \ln\frac{2}{\delta(1-\gamma)}\right) = \frac{\delta(1-\gamma)}{2},
	\end{align*}
	which implies that
	\[\norm{V(s_0) - \sum_{t=0}^{T_{\delta}} \gamma^t \underset{\tau \sim P(t;s_0)}{\E}\left[R^{(t)}(\tau)\right]}_q \leq \frac{\delta R_{\max}}{2},\]
	which is also trivially true if $T_{\delta} = T$, since then the left-hand side is $0$ by definition. Hence, we define the operation
	\begin{equation}
		\label{eq:LVtilde}
		\widetilde{L}_{G,V,\delta'} : \frac{1}{\sqrt{|G|}} \sum_{\mathbf{x} \in G} e^{i\varphi(\mathbf{x},s_0)} \ket{\mathbf{x}}\ket{s_0} \mapsto \frac{1}{\sqrt{|G|}} \sum_{\mathbf{x} \in G} e^{i\varphi(\mathbf{x},s_0) + i\frac{\mathbf{x}^T\sum_{t=0}^{T_{\delta}} \gamma^t\underset{\tau \sim P(t;s_0)}{\E} \left[R^{(t)}(\tau)\right]}{2\overline{r}_{G,\delta'}(q)R_{\max}}}\ket{\mathbf{x}}\ket{s_0}.
	\end{equation}
	and we argue that it is close to $\overline{L}_{G,V,\delta'}$. To that end, let
	\[\mathbf{y} = V(s_0) - \sum_{t=0}^{T_{\delta}} \gamma^t \underset{\tau \sim P(t;s_0)}{\E}\left[R^{(t)}(\tau)\right],\]
	and observe that for all $\mathbf{x} \in G_{\delta',\mathbf{y}}^{(q)}$,
	\begin{align*}
		\left|\frac{\mathbf{x}^T \left(V(s_0) - \sum_{t=0}^{T_{\delta}} \gamma^t \underset{\tau \sim P(t;s_0)}{\E} \left[R^{(t)}(\tau)\right]\right)}{2\overline{r}_{G,\delta'}(q)R_{\max}}\right| &\leq \frac{\overline{r}_{G,\delta'}(q) \cdot  \norm{V(s_0) - \sum_{t=0}^{T_{\delta}} \gamma^t \underset{\tau \sim P(t;s_0)}{\E}\left[R^{(t)}(\tau)\right]}_q}{2\overline{r}_{G,\delta'}(q)R_{\max}} \leq \frac{\delta}{4},
	\end{align*}
	which implies that
	\begin{align*}
		&\norm{\left(\overline{L}_{G,V,\delta'} - \widetilde{L}_{G,V,\delta'}\right) \frac{1}{\sqrt{|G|}} \sum_{\mathbf{x} \in G} e^{i\varphi(\mathbf{x},s_0)} \ket{\mathbf{x}}\ket{s_0}}^2 = \frac{1}{|G|}\sum_{\mathbf{x} \in G}  \left|e^{i\frac{\mathbf{x}^TV(s_0)}{2\overline{r}_{G,\delta'}(q)R_{\max}}} - e^{i\frac{\mathbf{x}^T \sum_{t=0}^{T_{\delta}} \gamma^t \underset{\tau \sim P(t;s_0)}{\E}\left[R^{(t)}(\tau)\right]}{2\overline{r}_{G,\delta'}(q)R_{\max}}}\right| \\
		&\quad \leq 4 \cdot \frac{\left|G \setminus G_{\delta',\mathbf{y}}^{(q)}\right|}{|G|} + \frac{1}{|G|} \sum_{\mathbf{x} \in G_{\delta',\mathbf{y}}^{(q)}} \left|\frac{\mathbf{x}^T\left(V(s_0) - \sum_{t=0}^{T_{\delta}} \gamma^t \underset{\tau \sim P(t;s_0)}{\E} \left[R^{(t)}(\tau)\right]\right)}{2\overline{r}_{G,\delta'}(q)R_{\max}}\right|^2 \leq 4\delta' + \frac{\left|G_{\delta',\mathbf{y}}^{(q)}\right|}{|G|} \cdot \frac{\delta^2}{16} \\
		&\quad \leq 4\delta' + \frac{\delta^2}{16}.
	\end{align*}
	Hence, if $\delta' \leq \delta^2/16$, we find that the norm error difference between states produced by $\overline{L}_{G,V,\delta'}$ and $\widetilde{L}_{G,V,\delta'}$ is at most $\delta/2$. Thus, it suffices to show that we can implement $\widetilde{L}_{G,V,\delta'}$ up to norm error $\delta/2$.
	
	If we use the machinery for the exact-depth case, as elaborated upon in \cref{lem:valuefunctionoracleexactdepth}, with the reward function $R^{(t)}$, then for any $t \in [T_{\delta}]_0$ we can implement the following operation
	\begin{equation}
		\label{eq:valuefunctionterm}
		\frac{1}{\sqrt{|G|}} \sum_{\mathbf{x} \in G} e^{i\varphi(\mathbf{x},s_0)} \ket{\mathbf{x}}\ket{s_0} \mapsto \frac{1}{\sqrt{|G|}} \sum_{\mathbf{x} \in G} e^{i\varphi(\mathbf{x},s_0) + i\frac{\mathbf{x}^T\underset{\tau \sim P(t;s_0)}{\E} \left[R^{(t)}(\tau)\right]}{2\overline{r}_{G,\eta'}(q)R_{\max}}}\ket{\mathbf{x}}\ket{s_0}
	\end{equation}
	up to precision $\eta > 0$ with a number of calls to $L_{\overline{G},R^{(t)}}$ and $D_P$ that scales as
	\[\O\left(\frac{r_{\overline{G}}(q)}{\overline{r}_{G,\eta'}(q)} \polylog\left(\frac{1}{\eta}\right)\right), \qquad \text{and} \qquad \O\left(t\polylog\left(\frac{1}{\eta}\right)\right),\]
	respectively, where $\eta' = \Theta(\eta^2\polylog(1/\eta))$, and $|\overline{G}| \geq (1-\eta')|G|$. With $\O(\polylog(T_{\delta}/\delta))$ calls to a perfect execution of the previous operation, it can be turned into a fractional phase oracle with precision $\delta/(4(T_{\delta}+1))$, implementing the operation
	\[\frac{1}{\sqrt{|G|}} \sum_{\mathbf{x} \in G} e^{i\varphi(\mathbf{x},s_0)} \ket{\mathbf{x}}\ket{s_0} \mapsto \frac{1}{\sqrt{|G|}} \sum_{\mathbf{x} \in G} e^{i\varphi(\mathbf{x},s_0) + i\frac{\mathbf{x}^T\gamma^t\underset{\tau \sim P(t;s_0)}{\E}\left[R^{(t)}(\tau)\right]}{2\overline{r}_{G,\eta'}(q)R_{\max}}}\ket{\mathbf{x}}\ket{s_0}.\]
	Hence, applying all these operations consecutively, with $t$ running from $0$ to $T_{\delta}$, implements $\widetilde{L}_{G,V,\delta'}$ up to operator norm error $\delta/4$. This implies that it suffices to choose $\eta = \Theta((\delta/T_{\delta})\polylog(T_{\delta}/\delta))$ to ensure that the resulting error per call to \cref{eq:valuefunctionterm} is at most $\delta/(4(T_{\delta}+1))$, and hence we can choose $\delta' = \eta' = \Theta((\delta/T_{\delta})^2/\polylog(T_{\delta}/\delta))$. With suitable constants, it follows automatically that $\delta' \leq \delta^2/16$, which is what we required in the previous step.
	
	The total number of calls to $L_{\overline{G},R^{(t)}}$ and $D_P$ now becomes
	\[\O\left(T_{\delta}\frac{r_{\overline{G}}(q)}{\overline{r}_{G,\delta'}(q)} \cdot \polylog\left(\frac{T_{\delta}}{\delta}\right)\right), \qquad \text{and} \qquad \O\left(\sum_{t=0}^{T_{\delta}} t\polylog\left(\frac{T_{\delta}}{\delta}\right)\right) = \O\left(T_{\delta}^2\polylog\left(\frac{T_{\delta}}{\delta}\right)\right),\]
	respectively. Finally, the observation that $T_{\delta} = \O(T^*\polylog(T^*/\delta))$ completes the proof.
\end{proof}

Finally, we show how the operation $\overline{L}_{G,V,\delta}$ can be implemented in the path-independent case.

\begin{lemma}[Computation of the value function (path-independent)]
	\label{lem:valuefunctionoraclepathindependent}
	Let $d,T \in \N$, $R_{\max} > 0$, $q \in [1,\infty]$, $\delta > 0$, and $\gamma \in [0,1]$, such that if $T = \infty$, then $\gamma < 1$. Let $S$ a state space, $P : S \times S \to [0,1]$ a probability transition matrix, $R_S : S \to \R^d$ a state-reward function, bounded by $R_{\max}$ in $\ell_q$-norm, and let $G \subseteq \R^d$ be a finite set. Then there exists a $\delta' = \Theta((\delta/T^*)^2/\polylog(T^*/\delta))$ such that the following holds. Suppose that $\overline{G} \subseteq G$ such that $|\overline{G}| \geq (1-\delta')|G|$, and that we have access to the state-reward function by means of lattice oracles $L_{\overline{G},R_S}$, as defined in \cref{eq:LRS}. Then, we can implement the operation $\overline{L}_{G,V,\delta'}$, defined in \cref{eq:LGV}, up to norm error $\delta$ with a number of calls to $D_P$ and $L_{\overline{G},R_S}$ that scale as
	\[\O\left((T^*)^2\polylog\left(\frac{T^*}{\delta}\right)\right), \qquad \text{and} \qquad \O\left(T^*\frac{r_{\overline{G}}(q)}{\overline{r}_{G,\delta'}(q)}\polylog\left(\frac{T^*}{\delta}\right)\right),\]
	respectively, in the limit where $T^*, r_{\overline{G}}(q)/\overline{r}_{G,\delta'}(q) \to \infty$ and $\delta \downarrow 0$.
\end{lemma}

\begin{proof}
	Let $T_{\delta}$ be as in \cref{eq:Tdelta}, and let $t$ be an integer such that $0 \leq t \leq T_{\delta}$. Let $\mathbf{x} \in G$ and $\ket{\tau} = (s_0, \dots, s_t) \in S^{t+1}$. We can represent the path register by $t+1$ state registers. If we call $L_{\overline{G},R_S}$ on the first and last register, then we implement the following mapping:
	\[\ket{\mathbf{x}}\ket{\tau} = \ket{\mathbf{x}}\ket{s_0} \cdots \ket{s_t} \mapsto e^{i\frac{\mathbf{x}^TR_S(s_t)}{2r_{G,\delta'}(q)R_{\max}}}\ket{\mathbf{x}}\ket{s_0} \cdots \ket{s_t} = e^{i\frac{\mathbf{x}^TR^{(t)}(\tau)}{2r_{G,\delta'}(q)R_{\max}}}\ket{\mathbf{x}}\ket{\tau},\]
	which exactly equals the mapping implemented by $L_{\overline{G},R^{(t)}}$. This means we can directly use the machinery from the cumulative-depth case, completing the proof.
\end{proof}

\subsection{Value function estimation}
\label{subsec:valuefunctionestimation}

In this subsection, we show how one can obtain a classical estimate of the value function, if we have access to this function via the operation defined in \cref{eq:LGV}. The method we use was first introduced in \cite{vanApeldoorn20}, page~24, which in turn uses some ideas from the gradient estimation algorithm introduced in \cite{jordan05}, and later developed in \cite{GAW18}.

Concretely, let $d \in \N$, $\varepsilon > 0$, $p \in [1,\infty]$ and $O$ an orthogonal $d \times d$ matrix. If we want to obtain a $\varepsilon$-precise estimate in $\ell_p$-norm of the quantity $OV(s_0)$, then we use the set $G_O$, defined as
\begin{equation}
	\label{eq:GO}
	G_O = \left\{\frac{O^T\mathbf{x}}{2^n} : \mathbf{x} \in \left\{-2^{n-1}+\frac12, -2^{n-1}+\frac32, \dots, 2^{n-1}-\frac12\right\}^d\right\} \quad \text{with} \quad n = \left\lceil\log\left(\frac{24d^{\frac12+\frac1p}}{\varepsilon}\right)\right\rceil.
\end{equation}
This set can be visualized as a hypercubic lattice in $d$ dimensions with side length $1$, rotated by the orthogonal matrix $O^T$. A graphical depiction of such a set when $d = 2$ is supplied in \cref{fig:hypercubiclattice}.

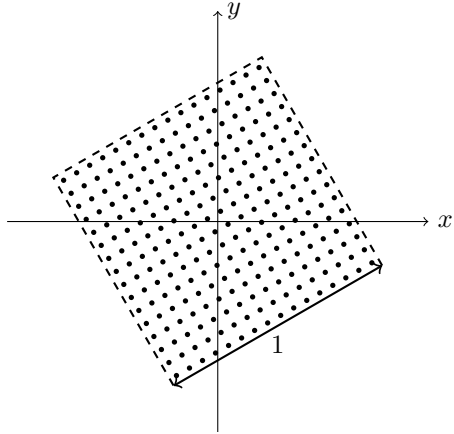
\begin{figure}[h!]
	\centering
	\begin{tikzpicture}[scale=.4]
		\draw[->] (-7,0) -- (7,0) node[right] {$x$};
		\draw[->] (0,-7) -- (0,7) node[right] {$y$};
		\foreach \x in {-3.75,-3.25,...,3.75} {
			\foreach \y in {-3.75,-3.25,...,3.75} {
				\fill ({cos(30)*\x+sin(30)*\y},{sin(30)*\x-cos(30)*\y}) circle[radius=.25em];
			}
		}
		\draw[<->,thick] ({cos(30)*-4+sin(30)*4},{sin(30)*-4-cos(30)*4}) to node[below] {$1$} ({cos(30)*4+sin(30)*4},{sin(30)*4-cos(30)*4});
		\draw[dashed,thick] ({cos(30)*4+sin(30)*4},{sin(30)*4-cos(30)*4}) to ({cos(30)*4+sin(30)*-4},{sin(30)*4-cos(30)*-4}) to ({cos(30)*-4+sin(30)*-4},{sin(30)*-4-cos(30)*-4}) to ({cos(30)*-4+sin(30)*4},{sin(30)*-4-cos(30)*4});
	\end{tikzpicture}
	\caption{Graphical depiction of the set $G_O$, defined in \cref{eq:GO}, in the case where $d = 2$. The side length of the grid is always $1$, and the number of points in every dimension is $2^n$, so in this case $n = 4$. The grid is rotated by the matrix $O^T$, in this case $O$ is a clockwise rotation matrix of $30$ degrees.}
	\label{fig:hypercubiclattice}
\end{figure}

As we already hinted at in the previous sections, it will prove crucial to analyze the different radii of this set $G_O$. Intuitively, the approximate radius $r_{G_O,\delta}(q)$ can be thought of as the radius of the $\ell_q$-ball that encloses exactly a $(1-\delta)$-fraction of $G_O$, and the effective radius $\overline{r}_{G_O,\delta}(q)$ can be thought of as half of the minimal distance between two parallel hyperplanes, both equally far from the origin, such that they chop off exactly a $\delta$-fraction of the grid. We calculate these quantities in the following lemma.

\begin{lemma}[Radii of the grid]
	\label{lem:radii}
	Let $d \in \N$, $\varepsilon > 0$, $p \in [1,\infty]$, $O$ an orthogonal $d \times d$ matrix, and $G_O$ as in \cref{eq:GO}. Then, for all $\delta > 0$,
	\begin{align*}
		r_{G_O,\delta}(q) &\leq d^{1-\frac1q}\sqrt{\frac{\ln\frac{2d}{\delta}}{2}} = \O\left(d^{1-\frac1q}\polylog\left(\frac{1}{\delta}\right)\right),\\
		\overline{r}_{G_O,\delta}(q) &\leq d^{\max\{0,\frac12-\frac1q\}}\sqrt{\frac{\ln\frac{2}{\delta}}{2}} = \O\left(d^{\max\{0,\frac12-\frac1q\}}\polylog\left(\frac{1}{\delta}\right)\right),
	\end{align*}
	where the big-$\O$-notation holds in the limit where $d \to \infty$ and $\delta \downarrow 0$.
\end{lemma}

\begin{proof}
	For the first statement, observe that it suffices to prove that
	\[\underset{\mathbf{x} \sim \Unif(G_O)}{\P}\left[\norm{\mathbf{x}}_{\frac{1}{1-\frac1q}} \geq d^{1-\frac1q}\sqrt{\frac{\ln\frac{2d}{\delta}}{2}}\right] \leq \delta.\]
	By H\"older's inequality, we have for all $\mathbf{x} \in \R^d$,
	\[\norm{\mathbf{x}}_{\frac{1}{1-\frac1q}} \leq d^{1-\frac1q}\norm{\mathbf{x}}_{\infty},\]
	and hence,
	\[\underset{\mathbf{x} \sim \Unif(G_O)}{\P}\left[\norm{\mathbf{x}}_{\frac{1}{1-\frac1q}} \geq d^{1-\frac1q}\sqrt{\frac{\ln\frac{2d}{\delta}}{2}}\right] \leq \underset{\mathbf{x} \sim \Unif(G_O)}{\P}\left[\norm{\mathbf{x}}_{\infty} \geq \sqrt{\frac{\ln\frac{2d}{\delta}}{2}}\right],\]
	which implies that it suffices to prove the case where $q = 1$. We now observe that, for all $t \geq 0$,
	\begin{align*}
		\underset{\mathbf{x} \sim \Unif(G_O)}{\P}\left[\norm{\mathbf{x}}_{\infty} \geq t\right] &= \underset{\mathbf{x} \sim \Unif(G_I)}{\P}\left[\norm{O^T\mathbf{x}}_{\infty} \geq t\right] = \underset{\mathbf{x} \sim \Unif(G_I)}{\P}\left[\max_{j \in [d]} \left|\sum_{k=1}^d O_{kj}x_k\right| \geq t\right] \\
		&\leq \sum_{j=1}^d \P\left[\left|\sum_{k=1}^d O_{kj}x_k\right| \geq t\right] \leq \sum_{j=1}^d 2\exp\left(-\frac{2t^2}{\sum_{k=1}^d O_{kj}^2}\right) = 2d\exp\left(-2t^2\right),
	\end{align*}
	where we used the union bound, and Hoeffding's inequality. Plugging in $t = \sqrt{\ln(2d/\delta)/2}$ yields
	\[\underset{\mathbf{x} \sim \Unif(G_O)}{\P}\left[\norm{\mathbf{x}}_{\infty} \geq \sqrt{\frac{\ln\frac{2d}{\delta}}{2}}\right] \leq 2d\exp\left(-2 \cdot \frac{\ln\frac{2d}{\delta}}{2}\right) = \delta,\]
	completing the proof of the first statement. 
	
	For the second statement, observe that it is sufficient to prove that for all $\mathbf{y} \in \R^d$, satisfying $\norm{\mathbf{y}}_q \leq 1$,
	\[\underset{\mathbf{x} \sim \Unif(G_O)}{\P}\left[\left|\mathbf{x}^T\mathbf{y}\right| \geq d^{\max\{0,\frac12-\frac1q\}}\sqrt{\frac{\ln\frac{2}{\delta}}{2}}\right] \leq \delta.\]
	To that end, let $\mathbf{y} \in \R^d$ such that $\norm{\mathbf{y}}_q \leq 1$. We again employ Hoeffding's inequality, which this time allows us to obtain that for all $t \geq 0$,
	\[\underset{\mathbf{x} \sim \Unif(G_O)}{\P}\left[\left|\mathbf{x}^T\mathbf{y}\right| \geq t\right] = \underset{\mathbf{x} \sim \Unif(G_I)}{\P}\left[\left|\mathbf{x}^TO\mathbf{y}\right| \geq t\right] \leq 2\exp\left(-\frac{2t^2}{\sum_{j=1}^d (O\mathbf{y})_j^2}\right) = 2\exp\left(-\frac{2t^2}{\norm{O\mathbf{y}}_2^2}\right).\]
	We can bound this further by H\"older's inequality, as
	\[\norm{O\mathbf{y}}_2 = \norm{\mathbf{y}}_2 \leq \norm{\mathbf{y}}_q \cdot \begin{cases}
		d^{\frac12-\frac1q}, & \text{if } q \geq 2, \\
		1, & \text{otherwise}
	\end{cases} \leq d^{\max\{0,\frac12-\frac1q\}}.\]
	By choosing $t = d^{\max\{0,1/2-1/q\}}\sqrt{\ln(2/\delta)/2}$, we find that
	\begin{align*}
		\underset{\mathbf{x} \sim \Unif(G_O)}{\P}\left[\left|\mathbf{x}^T\mathbf{y}\right| \geq d^{\max\{0,\frac12-\frac1q\}}\sqrt{\frac{\ln\frac{2}{\delta}}{2}}\right] &\leq 2\exp\left(-\frac{2d^{\max\{0,1-\frac2q\}}}{d^{\max\{0,1-\frac2q\}}} \cdot \frac{\ln\frac{2}{\delta}}{2}\right) = \delta,
	\end{align*}
	completing the proof.
\end{proof}

Note that there is a profound difference between the approximate and effective radius of $G_O$ -- if $q \geq 2$, the difference is in the order of $\sqrt{d}$. In \cref{lem:valuefunctionoracleexactdepth}, we were able to exploit these differences in the query complexity to $D_P$, but not in the query complexity to the reward oracles. Since in \cref{sec:MMClowerbounds} we present matching lower bounds for all values of $q$, we find that this barrier is fundamental, there indeed exists no other trick to reduce the query complexity to the reward oracles in \cref{lem:valuefunctionoracleexactdepth} in full generality.

Now that we have analyzed the radii of the grid $G_O$, we turn to the algorithm that approximates the value function. The idea for this algorithm stems from \cite{jordan05}, and is the fundamental idea behind \cite{GAW18} and \cite{Cor18}, Chapter 4.

\begin{lemma}[Value function estimation]
	\label{lem:valuefunctionestimationlattice}
	Let $d \in \N$, $0 < \varepsilon < R_{\max}$, $p,q \in [1,\infty]$, $\delta > 0$, $O$ a $d \times d$ orthogonal matrix, and $G_O$ as in \cref{eq:GO}. Let
	\[M = \left\lceil\frac{32\pi \overline{r}_{G_O,\delta}(q)R_{\max}d^{\frac1p}}{\varepsilon}\right\rceil,\]
	and let $V(s_0) \in \R^d$ be accessible through $\overline{L}_{G_O,V,\delta}$. Then, we can compute a vector $\mathbf{v} \in \R^d$ such that
	\[\norm{\mathbf{v} - OV(s_0)}_p \leq \varepsilon,\]
	with probability at least $5/6$, with a number of calls to $\overline{L}_{G_O,V,\delta}$ that scales as
	\[\widetilde{\O}\left(M\right) = \widetilde{\O}\left(\frac{\overline{r}_{G_O,\delta}(q)R_{\max}d^{\frac1p}}{\varepsilon}\right) = \widetilde{\O}\left(\frac{R_{\max}}{\varepsilon}d^{\frac1p+\max\{0,\frac12-\frac1q\}}\right), \qquad \left(R_{\max},d \to \infty, \;\; \varepsilon \downarrow 0\right),\]
	where the tilde hides polylogarithmic factors in $R_{\max}$, $d$, $1/\varepsilon$.
\end{lemma}

\begin{proof}
	We will act on a vector register encoding the vectors from $G_O$, as defined in \cref{eq:GO}, and a state register, starting in state
	\[\frac{1}{\sqrt{2^{nd}}}\sum_{\mathbf{x} \in G_O} \ket{\mathbf{x}}\ket{s_0}.\]
	Next, we call the lattice oracle $\overline{L}_{G_O,V,\delta}$ a total of $M$ times, which prepares the state
	\begin{align*}
	\frac{1}{\sqrt{2^{nd}}} \sum_{\mathbf{x} \in G_O} e^{iM\frac{\mathbf{x}^T V(s_0)}{2\overline{r}_{G_O,\delta}(q)R_{\max}}}\ket{\mathbf{x}} \ket{s_0} &= \frac{1}{\sqrt{2^{nd}}} \sum_{\mathbf{y} \in \{-2^{n-1} + \frac12, \dots, 2^{n-1} - \frac12\}^d} e^{\frac{2\pi i}{2^n} \cdot \mathbf{y}^T\frac{MOV(s_0)}{4\pi \overline{r}_{G_O,\delta}(q)R_{\max}}} \ket{\mathbf{y}} \ket{s_0} \\
	&= \bigotimes_{j=1}^d \left[\frac{1}{\sqrt{2^n}} \sum_{y_j=-2^{n-1}+\frac12}^{2^{n-1}-\frac12} e^{\frac{2\pi i}{2^n} \cdot y_j \cdot \frac{M(OV(s_0))_j}{4\pi \overline{r}_{G_O,\delta}(q)R_{\max}}}\ket{y_j}\right]\ket{s_0}.
	\end{align*}
	Next, we run the inverse quantum Fourier transform on $n$ qubits on each of the $d$ individual parts of the above tensor product separately, after which we measure each of the $n$-qubit states. We interpret the outcomes as signed $n$-bit integers $b_j$, with $j \in [d]$, which we bundle together in a vector $\mathbf{b} \in \R^d$. Since for all $j \in [d]$,
	\[\left|\frac{M(OV(s_0))_j}{4\pi \overline{r}_{G_O^{(q)},\delta}(q)R_{\max}}\right| \leq \frac{32\pi \overline{r}_{G_O^{(q)},\delta}(q)R_{\max}d^{\frac1p}}{4\pi \overline{r}_{G_O^{(q)},\delta}(q)R_{\max}\varepsilon}\left|(UV(s_0))_j\right| = \frac{8d^{\frac12+\frac1p}}{\varepsilon} = \frac132^{\log\left(\frac{24d^{\frac12+\frac1p}}{\varepsilon}\right)} \leq \frac13 \cdot 2^n,\]
	the analysis of the inverse quantum Fourier transform, \cite{NC00}, Equation (5.34), now implies that for each $j \in [d]$,
	\[\P\left[\left|b_j - \frac{M(OV(s_0))_j}{4\pi \overline{r}_{G_O^{(q)},\delta}(q)R_{\max}}\right| \leq 4\right] \geq \frac56.\]
	Next, we let
	\[\mathbf{v} = \frac{4\pi \overline{r}_{G_O^{(q)},\delta}(q)R_{\max}}{M}\mathbf{b},\]
	which implies for all $j \in [d]$ that
	\[\left|v_j - (OV(s_0))_j\right| = \frac{4\pi \overline{r}_{G_O^{(q)},\delta}(q)R_{\max}}{M}\left|b_j - \frac{M(OV(s_0))_j}{4\pi \overline{r}_{G_O^{(q)},\delta}(q)R_{\max}}\right| \leq \frac{\varepsilon}{4d^{\frac1p}}\left|b_j - \frac{M(OV(s_0))_j}{4\pi \overline{r}_{G_O^{(q)},\delta}(q)R_{\max}}\right|,\]
	and hence for all $j \in [d]$,
	\[\P\left[\left|v_j - (OV(s_0))_j\right| \leq \frac{\varepsilon}{d^{\frac1p}}\right] \geq \P\left[\left|b_j - \frac{M(OV(s_0))_j}{4\pi \overline{r}_{G_O^{(q)},\delta}(q)R_{\max}}\right| \leq 4\right] \geq \frac56.\]
	Finally, we let
	\[N = \left\lceil18\ln(6d)\right\rceil,\]
	and we run the above procedure $N$ times, generating vectors $\mathbf{v}_1, \dots, \mathbf{v}_N$. We take the coordinate-wise median of all these vectors, which we call $\mathbf{v}$. For any $j \in [d]$, let $X_k$ be the random variable denoting whether the $j$th entry of $\mathbf{v}_k$ approximates $(OV(s_0))_j$ with precision $\varepsilon/d^{1/p}$. By the Hoeffding bound, we find that
	\begin{align*}
	\P\left[\sum_{j=1}^N X_j \leq \frac{N}{2}\right] &\leq \P\left[\left|\sum_{j=1}^N X_j - N\E\left[X_j\right]\right| \geq \left|\frac{N}{2} - N\E\left[X_j\right]\right|\right] \leq \exp\left(-2N\left|\frac12 - \E\left[X_j\right]\right|^2\right) \\
	&\leq \exp\left(-\frac{N}{18}\right) \leq \exp\left(-\ln(6d)\right) = \frac{1}{6d},
	\end{align*}
	and hence with probability at least $1-1/(6d)$, more than half of the approximations to the $j$th entry of $OV(s_0)$ are $\varepsilon/d^{1/p}$-close. This means that the median is also $\varepsilon/d^{1/p}$-close, and hence by the union bound we find that with probability at least $5/6$ all entries of $\mathbf{v}$ are this close to $OV(s_0)$, so
	\[\norm{\mathbf{v} - OV(s_0)}_p \leq d^{\frac1p}\norm{\mathbf{v} - OV(s_0)}_\infty \leq \varepsilon.\]
	Moreover, this step only adds a multiplicative factor of $\O(N)$, which is logarithmic in $d$, to the number of calls to $\overline{L}_{G_O,V,\delta}$ we perform to generate one vector $\mathbf{v}_j$. Thus, the total query complexity is
	\[\widetilde{\O}\left(M\right) = \widetilde{\O}\left(\frac{\overline{r}_{G_O,\delta}(q)R_{\max}d^{\frac1p}}{\varepsilon}\right) = \widetilde{\O}\left(\frac{R_{\max}}{\varepsilon} d^{\frac1p + \max\{0, \frac12-\frac1q\}}\right).\]
	In the final equality, we used the asymptotic complexity of $\overline{r}_{G_O,\delta}$, derived in \cref{lem:radii}. This completes the proof.
\end{proof}

\subsection{Results}
\label{subsec:algorithmresults}

In this subsection, we stitch all the constructions from \cref{subsec:rewardoracleconversions,subsec:valueoracleconstruction,subsec:valuefunctionestimation} together into one theorem statement that lists all the query complexities for all possible combinations of oracle type, setting, and approximation parameters. We present matching lower bounds in \cref{sec:MMClowerbounds}, up to polylogarithmic factors, for all these complexities, implying that the techniques we presented in this section are essentially optimal.

\begin{theorem}[Multivariate Monte Carlo estimation algorithms]
    \label{thm:mvmc-alg}
Let $d \in \N$, $0 < \varepsilon < R_{\max}$, $\gamma \in [0,1]$, $p,q \in [1,\infty]$, and $T \in \N \cup \{\infty\}$ such that if $T = \infty$, then $\gamma < 1$. Let $S$ be a state space, $P : S \times S \to [0,1]$ a probability transition matrix and $O$ a $d \times d$ orthogonal matrix. The query complexities of the quantum algorithms solving the multivariate Monte Carlo estimation problem rotated by $O$ up to error $\varepsilon$ in $\ell_p$-norm in the exact-depth, cumulative-depth and path-independent cases are listed in the table below:
	\begin{center}
		\begin{tabular}{r|c|c}
			Case & Calls to reward oracle & Calls to probability transition oracle \\\hline
			Exact-depth & $\displaystyle\widetilde{\O}\left(\frac{R_{\max}}{\varepsilon}d^{\frac1p} \cdot \xi(d,q)\right)$ & $\displaystyle\widetilde{\O}\left(T\frac{R_{\max}}{\varepsilon}d^{\frac1p+\max\{0,\frac12-\frac1q\}}\right)$ \\
			Cumulative-depth & $\displaystyle\widetilde{\O}\left(T^*\frac{R_{\max}}{\varepsilon}d^{\frac1p} \cdot \xi(d,q)\right)$ & $\displaystyle\widetilde{\O}\left((T^*)^2\frac{R_{\max}}{\varepsilon}d^{\frac1p+\max\{0,\frac12-\frac1q\}}\right)$ \\
			Path-independent & $\displaystyle\widetilde{\O}\left(T^*\frac{R_{\max}}{\varepsilon}d^{\frac1p} \cdot \xi(d,q)\right)$ & $\displaystyle\widetilde{\O}\left((T^*)^2\frac{R_{\max}}{\varepsilon}d^{\frac1p+\max\{0,\frac12-\frac1q\}}\right)$
		\end{tabular}
	\end{center}
	where $\xi(d,q)$ depends on the specific access model to the reward oracle that we have:
	\begin{center}
		\begin{tabular}{r|c}
			Access model & $\xi(d,q)$ \\\hline
			Phase oracle & $d$ \\
			Probability oracle & $d^{1-\frac{1}{2q}}$ \\
			Distribution oracle & $d^{1-\frac1q}$ \\
			Lattice oracle on $G_O$ & $d^{1-\frac1q}$
		\end{tabular}
	\end{center}
	The big-$\O$-notation is in the limit where $T^{(*)},R_{\max},d \to \infty$ and $\varepsilon \downarrow 0$, the tilde hides polylogarithmic factors in $T^{(*)}$, $R_{\max}$, $d$ and $1/\varepsilon$.
\end{theorem}

\begin{proof}
	All the claimed complexities follow from combining \cref{lem:valuefunctionestimationlattice} with the result for the specific setting from \cref{subsec:valueoracleconstruction}, and possibly the result from \cref{subsec:rewardoracleconversions} corresponding to the given oracle type. We illustrate the proof method in the exact-depth setting, when we have access to the reward function via a phase oracle, and omit the others since they follow the exact same reasoning.
	
	Let $\delta,\delta',\eta,\eta' > 0$, and let
	\[\overline{G}_{\delta'} = (G_O)_{\delta'/2}^{(q)} \cap (G_O)_{\delta'/2}^{(\infty)}.\]
	It follows immediately that
	\[\frac{\left|G_O \setminus \overline{G}_{\delta'}\right|}{|G_O|} \leq \frac{\left|G_O \setminus (G_O)_{\delta'/2}^{(q)}\right| + \left|G_O \setminus (G_O)_{\delta'/2}^{(\infty)}\right|}{|G_O|} \leq \frac{\delta'}{2} + \frac{\delta'}{2} = \delta',\]
	and hence $|\overline{G}_{\delta'}| \geq (1-\delta')|G_O|$. We employ \cref{lem:phasetolattice} to implement the lattice oracle $L_{\overline{G}_{\delta'},R}$ up to norm error $\delta$ with a number of calls to $O_R$ that scales as
	\[\O\left(\frac{r_{\overline{G}_{\delta'}}(\infty)}{r_{\overline{G}_{\delta'}}(q)}\polylog\left(\frac{r_{\overline{G}_{\delta'}}(\infty)}{r_{\overline{G}_{\delta'}}(q)\delta}\right)\right).\]
	By \cref{lem:valuefunctionoracleexactdepth}, we can implement $\overline{L}_{G_O,V,\eta'}$ up to norm error $\eta$, where $\eta' = \Theta(\eta^2/\polylog(1/\eta))$, using a number of calls to $D_P$ and $L_{\overline{G}_{\eta'},R}$ that scales as
	\[\O\left(T\polylog\left(\frac{1}{\eta}\right)\right), \qquad \text{and} \qquad \O\left(\frac{r_{\overline{G}_{\eta'}}(q)}{\overline{r}_{G_O,\eta'}(q)}\polylog\left(\frac{1}{\eta}\right)\right),\]
	respectively. Finally, by \cref{lem:valuefunctionestimationlattice}, we can calculate the value function from $\overline{L}_{G,V,\eta'}$, using a number of calls that scales as
	\[\widetilde{\O}\left(\frac{\overline{r}_{G_O,\eta'}(q)R_{\max}d^{\frac1p}}{\varepsilon}\right).\]
	Hence, we can choose $\delta' = \eta'$, $\eta = \widetilde{\Theta}(\varepsilon/(\overline{r}_{G_O,\eta'}(q)R_{\max}d^{1/p}))$, and $\delta = \widetilde{\Theta}(\eta\overline{r}_{G_O,\eta'}(q)/r_{\overline{G}_{\delta'}}(q))$, such that the total norm error that we make is upper bounded by $1/6$, which implies that the total success probability of the algorithm decreases at most from $5/6$ to $2/3$, as is for instance proven in \cite{Cor18}, Appendix B. Moreover, the resulting query complexity to $O_R$ becomes
	\begin{align*}
		\widetilde{\O}\left(\frac{\overline{r}_{G_O,\eta'}(q)R_{\max}d^{\frac1p}}{\varepsilon} \cdot \frac{r_{\overline{G}_{\delta'}}(q)}{\overline{r}_{G_O,\eta'}(q)} \cdot \frac{r_{\overline{G}_{\delta'}}(\infty)}{r_{\overline{G}_{\delta'}}(q)}\right) &= \widetilde{\O}\left(\frac{R_{\max}}{\varepsilon}d^{\frac1p} \cdot r_{\overline{G}_{\delta'}}(\infty)\right) = \widetilde{\O}\left(\frac{R_{\max}}{\varepsilon}d^{\frac1p} \cdot r_{G_O,\frac{\delta'}{2}}(\infty)\right) \\
		&= \widetilde{\O}\left(\frac{R_{\max}}{\varepsilon}d^{\frac1p+1}\right),
	\end{align*}
	where we used that since $\overline{G}_{\delta'} \subseteq (G_O)_{\delta'/2}^{(\infty)}$, it follows from \cref{lem:trimmedgrids} that
	\[r_{\overline{G}_{\delta'}}(\infty) \leq r_{(G_O)_{\delta'/2}^{(\infty)}}(\infty) = r_{G_O,\frac{\delta'}{2}}(\infty).\]
	Finally, we can also analyze the query complexity to $D_P$, which becomes
	\[\widetilde{\O}\left(\frac{\overline{r}_{G_O,\eta'}(q)R_{\max}d^{\frac1p}}{\varepsilon} \cdot T\right) = \widetilde{\O}\left(T\frac{R_{\max}}{\varepsilon}d^{\frac1p + \max\{0,\frac12-\frac1q\}}\right).\]
	The proofs for all other cases follow similarly, but with slightly modified constants, and possibly different choices for $\overline{G}_{\delta'}$. This completes the proof.
\end{proof}

\section{Lower bounds}
\label{sec:MMClowerbounds}

In this section, we provide lower bounds on the query complexities of algorithms that solve the multivariate Monte Carlo estimation problem. First, we show that solving the multivariate Monte Carlo estimation problem is at least as hard as producing high-overlap bit strings, and lower bound the hardness of this problem when one has regular phase oracle access to this bit string, in \cref{subsec:highoverlapreduction}. After that, we define specific instances of the multivariate Monte Carlo estimation problem and relate their hardness to the hardness of the aforementioned high-overlap bit-string problem. In \cref{subsec:LBR}, we tailor these instances to arrive at lower bounds for the query complexity to the reward oracle, and in \cref{subsec:LBP}, we focus on the query complexity to the probability transition oracle.

\subsection{Reduction to producing high-overlap bit strings}
\label{subsec:highoverlapreduction}

First, we introduce the concept of high-overlap bit strings. To that end, let $d \in \N$ and $O \in \R^{d \times d}$ an orthogonal matrix. We say that $\mathbf{b}, \mathbf{b}^* \in \{0,1\}^d$ have high overlap if $\norm{O(\mathbf{b} - \mathbf{b}^*)}_1 \leq d/4$. In particular, if $O$ is the identity matrix, this condition implies that $\mathbf{b}$ and $\mathbf{b}^*$ disagree in at most a quarter of the bits, so here the notion of high overlap makes intuitive sense. We will use the same terminology in the case where $O$ is an arbitrary orthogonal matrix, even though in this case $\mathbf{b}$ and $\mathbf{b}^*$ might differ in many more bits than just $d/4$.

\cref{lem:smallell1distancereduction} shows that if we have any algorithm $\A$ that solves the multivariate Monte Carlo estimation problem on a particular family of instances indexed by a bit string $\mathbf{b} \in \{0,1\}^d$, we are able to produce a high-overlap bit string $\mathbf{b}^*$ with just one run of $\A$.

\begin{lemma}[Reduction to producing a high-overlap bit string]
	\label{lem:smallell1distancereduction}
	Let $d \in \N$, $O,O' \in \R^{d \times d}$ orthogonal matrices, $S$ a state space, $s_0 \in S$ an initial state, and $\mathbf{y} \in \R^d$. Suppose that for every $\mathbf{b} \in \{0,1\}^d$, we have an instance of the multivariate Monte Carlo estimation problem (in any setting) whose value function, denoted by $V^{(\mathbf{b})}$, equals
	\[V^{(\mathbf{b})}(s_0) = \mathbf{y} + \frac{8\varepsilon}{d}O'\mathbf{b}.\]
	Then, using a single call to any algorithm $\A$ that solves the multivariate Monte Carlo estimation problem rotated by $O$ up to precision $\varepsilon$ in $\ell_1$-norm with high probability, we are able to construct an algorithm $\B$ that when run on the instance labeled by $\mathbf{b} \in \{0,1\}^d$, with high probability produces a bit string $\mathbf{b}^* \in \{0,1\}^d$ that satisfies
	\[\norm{OO'(\mathbf{b} - \mathbf{b}^*)}_1 \leq \frac{d}{4}.\]
\end{lemma}

\begin{proof}
	Let $\A(\mathbf{b})$ denote the random variable that describes the outcome of $\A$ when run on the instance labeled by $\mathbf{b}$. We now let $\B$ be the algorithm that runs $\A$, and then outputs the bit string $\mathbf{b}^*$ as follows:
	\[\mathbf{b}^* = \underset{\mathbf{x} \in \{0,1\}^d}{\argmin} \norm{O(\A(\mathbf{b}) - V^{(\mathbf{x})}(s_0))}_1.\]
	We know that $\norm{O(\A(\mathbf{b}) - V^{(\mathbf{b})}(s_0))}_1 \leq \varepsilon$, with probability at least $2/3$. Moreover, since we took the minimum over all possible assignments for $\mathbf{x}$ in the definition of $\mathbf{b}^*$, we know that $\norm{O(\A(\mathbf{b}) - V^{(\mathbf{b}^*)}(s_0))}_1$ is at most $\norm{O(\A(\mathbf{b}) - V^{(\mathbf{b})}(s_0))}_1$, and hence at most $\varepsilon$ as well. Thus, with probability at least $2/3$, we find that
	\begin{align*}
		\norm{OO'(\mathbf{b} - \mathbf{b}^*)}_1 &= \norm{OO'\left(\frac{d}{8\varepsilon}(O')^T\left(V^{(\mathbf{b})}(s_0) - \mathbf{y}\right) - \frac{d}{8\varepsilon}(O')^T\left(V^{(\mathbf{b}^*)}(s_0) - \mathbf{y}\right)\right)}_1 \\
		&= \frac{d}{8\varepsilon}\norm{O\left(V^{(\mathbf{b})}(s_0) - V^{(\mathbf{b}^*)}(s_0)\right)}_1 \\
		&\leq \frac{d}{8\varepsilon} \left[\norm{O\left(V^{(\mathbf{b})}(s_0) - \A(\mathbf{b})\right)}_1 + \norm{O\left(\A(\mathbf{b}) - V^{(\mathbf{b}^*)}(s_0)\right)}_1\right] \\
		&\leq \frac{d}{4\varepsilon} \norm{O\left(V^{(\mathbf{b})}(s_0) - \A(\mathbf{b})\right)}_1 \leq \frac{d}{4\varepsilon} \cdot \varepsilon = \frac{d}{4}.
	\end{align*}
	This completes the proof.
\end{proof}

Next, in \cref{lem:highoverlapbitstring}, we show that in general, if we have access to a bit string $\mathbf{b} \in \{0,1\}^d$ via a regular phase oracle, then it is difficult to produce a high-overlap bit string $\mathbf{b}^*$.

\begin{lemma}[Hardness of producing a high-overlap bit string with a regular phase oracle]
	\label{lem:highoverlapbitstring}
	Let $d \in \N$, and $O \in \R^{d \times d}$ an orthogonal matrix. Suppose that we are given access to a bit string $\mathbf{b} \in \{0,1\}^d$ by means of a (controlled) phase oracle that acts as
	\[O^{(\mathbf{b})} : \ket{j} \mapsto (-1)^{b_j}\ket{j}.\]
	Then, in order to produce a bit string $\mathbf{b}^* \in \{0,1\}^d$ that with high probability satisfies
	\[\norm{O(\mathbf{b} - \mathbf{b}^*)}_1 \leq \frac{d}{4},\]
	we need to make a number of calls to $O^{(\mathbf{b})}$ that scales as $\Omega(d)$, as $d \to \infty$.
\end{lemma}

\begin{proof}
	Let $\A$ be an algorithm that solves the problem stated in the lemma, i.e., that produces a bit string $\mathbf{b}^* \in \{0,1\}^d$ such that $\norm{O(\mathbf{b} - \mathbf{b}^*)}_1 \leq d/4$, with probability at least $2/3$. Then, we let $\B$ be the algorithm that calls $\A$, and subsequently chooses a new bit string $\mathbf{b}^{**}$ uniformly at random from the set of all bit strings that satisfy $\norm{O(\mathbf{b}^* - \mathbf{b}^{**})}_1 \leq d/4$. The probability that this new bit string exactly equals $\mathbf{b}$ can now be lower bounded in terms of the size of this set, namely
	\[\P\left[\mathbf{b}^{**} = \mathbf{b}\right] \geq \frac23 \cdot \frac{1}{|\{\mathbf{x} \in \{0,1\}^d : \norm{O(\mathbf{b}^* - \mathbf{x})}_1 \leq \d/4\}|} = \frac{2}{3 \cdot 2^d} \cdot \underset{\mathbf{x} \sim \Unif(\{0,1\}^d)}{\P}\left[\norm{O(\mathbf{b}^* - \mathbf{x})}_1 \leq \frac{d}{4}\right]^{-1}.\]
	By the information theoretic lower bound, \cite{FGGS99}, Equation 4, this implies that
	\begin{equation}
	\label{eq:informationbound}
	2^d \leq \frac{3 \cdot 2^d}{2} \cdot \sum_{k=0}^{Q_{\A}} \binom{d}{k} \cdot \underset{\mathbf{x} \sim \Unif(\{0,1\}^d)}{\P}\left[\norm{O(\mathbf{b}^* - \mathbf{x})}_1 \leq \frac{d}{4}\right],
	\end{equation}
	where $Q_{\A}$ is the number of queries that $\A$ makes to $O^{(\mathbf{b})}$. Next, observe that the $j$th entry of $\mathbf{b}^* - \mathbf{x}$ takes values uniformly in $\{0,-1\}$ if $b^*_j = 0$, and uniformly in $\{0,1\}$ if $b^*_j = 1$. Hence, $(2b^*_j-1)(b^*_j-x_j)$ takes values uniformly in $\{0,1\}$, and so $\mathbf{y} = \diag(2\mathbf{b}^*-\mathbf{1})(\mathbf{b}^* - \mathbf{x})$ is uniformly distributed over $\{0,1\}^d$. Furthermore, after defining $\overline{O} = O\diag(2\mathbf{b}^*-\mathbf{1})$, we find that
	\[\overline{O}\mathbf{y} = O\diag(2\mathbf{b}^*-\mathbf{1})\diag(2\mathbf{b}^*-\mathbf{1})(\mathbf{b}^*-\mathbf{x}) = O(\mathbf{b}^*-\mathbf{x}).\]
	Moreover, $\overline{O}$ is also an orthogonal matrix, since it can be derived from $O$ by multiplying some of its columns with $-1$. Hence, we can now use \cref{lem:concentrationbound} to obtain that
	\[\underset{\mathbf{x} \sim \Unif(\{0,1\}^d)}{\P}\left[\norm{O(\mathbf{b}^* - \mathbf{x})}_1 \leq \frac{d}{4}\right] = \underset{\mathbf{y} \sim \Unif(\{0,1\}^d)}{\P}\left[\norm{\overline{O}\mathbf{y}}_1 \leq \frac{d}{4}\right] \leq 2^{-\frac{\log(e)}{2}\left(\frac{1}{2\sqrt{2}} - \frac14\right)^2d}.\]
	We can upper bound the other factor in the right-hand side of \cref{eq:informationbound} with the binary-entropy function, a proof of which can for instance be found in \cite{FG06}, Lemma 16.19, which implies that
	\[2^d \leq \frac{3 \cdot 2^d}{2} \cdot 2^{dH\left(\frac{Q_{\A}}{d}\right)} \cdot 2^{-\frac{\log(e)}{2}\left(\frac{1}{2\sqrt{2}} - \frac14\right)^2d} \leq \frac32 \cdot 2^{d\left(1 + H\left(\frac{Q_{\A}}{d}\right) - \frac{\log(e)}{2}\left(\frac{1}{2\sqrt{2}} - \frac14\right)^2\right)},\]
	where $H(x) = -x\log(x) - (1-x)\log(1-x)$, for all $x \in [0,1]$. Taking the logarithm and dividing by $d$ on both sides yields
	\[1 \leq \frac{\log(3) - 1}{d} + 1 + H\left(\frac{Q_{\A}}{d}\right) - \frac{\log(e)}{2}\left(\frac{1}{2\sqrt{2}} - \frac14\right)^2,\]
	which in turn implies that
	\[H\left(\frac{Q_{\A}}{d}\right) \geq \frac{\log(e)}{2}\left(\frac{1}{2\sqrt{2}} - \frac14\right)^2 - \frac{\log(3)-1}{d} \geq \frac{1}{200},\]
	where we used that $d \geq 250$. Since $H$ is increasing in the interval $[0,1/2]$ and $H(1/3000) < 1/200$, we find that
	\[Q_{\A} \geq \frac{d}{3000},\]
	completing the proof.
\end{proof}

Now that we know that finding a high-overlap bit string can be reduced to the multivariate Monte Carlo estimation problem, and we also know how hard this problem is when one has access to a regular phase oracle, all that remains is relating this access model to the way we assume to have encoded the input in the multivariate Monte Carlo estimation problem. This is the objective of the next two sections.

\subsection{Lower bounds on the query complexity to the reward function}
\label{subsec:LBR}

In this section, we focus on lower bounding the number of queries to the reward oracles required to solve the multivariate Monte Carlo estimation problem. We first aim to prove a lower bound in the path-independent setting, for the case where we have access to the reward function via the phase oracle, defined in \cref{eq:ORS}. This oracle, however, later on will provide access to a bit string $\mathbf{b}$ via a \textit{fractional phase oracle}, so we prove the hardness of the problem of finding high-overlap bit strings when one has access to fractional phase oracles first.

\begin{lemma}[Hardness of producing high-overlap bit strings with a fractional phase oracle]
	\label{lem:fractionalphasequerylowerbound}
	Let $d \in \N$, $\varepsilon > 0$ and $O \in \R^{d \times d}$ be an orthogonal matrix. Suppose that we are given access to a bit string $\mathbf{b} \in \{0,1\}^d$ by means of a (controlled) fractional phase oracle that acts as
	\[O_{\varepsilon}^{(\mathbf{b})} : \ket{j} \mapsto e^{i\varepsilon b_j}\ket{j}.\]
	Then, in order to produce a bit string $\mathbf{b}^* \in \{0,1\}^d$ that satisfies $\norm{O(\mathbf{b} - \mathbf{b}^*)}_1 \leq d/4$ with high probability, we need to make a number of calls to $O^{(\mathbf{b})}_{\varepsilon}$ that scales at least as
	\[\Omega\left(\frac{d}{\varepsilon}\right), \qquad (d \to \infty, \;\; \varepsilon \downarrow 0).\]
\end{lemma}

\begin{proof}
	This proof is inspired by \cite{LMRSS11}, Appendix B. Analogous results for computing functions with fractional phase oracles can be derived more easily, e.g. via \cite{Belovs15}, Theorem 37, or \cite{YM11}, Theorem 3.1. However, in this setting we specifically need to consider the case in which we are evaluating a relation, rather than a function, which requires us to redo the analysis in this slightly different setting.
	
	Let $r \subseteq \{0,1\}^d \times \{0,1\}^d$ be the relation defined as
	\[r = \left\{(\mathbf{b},\mathbf{b}^*) \in \{0,1\}^d \times \{0,1\}^d : \norm{O(\mathbf{b} - \mathbf{b}^*)}_1 \leq \frac{d}{4}\right\}.\]
	Observe that finding a bit string $\mathbf{b}^* \in \{0,1\}^d$ that satisfies $\norm{O(\mathbf{b} - \mathbf{b}^*)}_1 \leq d/4$ is equivalent to finding a bit string $\mathbf{b}^*$ such that $(\mathbf{b},\mathbf{b}^*) \in r$.
	
	Note that instead of querying $O_{\varepsilon}^{(\mathbf{b})}$, it is equivalent to make queries to a slightly modified oracle, defined as
	\[\overline{O}_{\varepsilon}^{(\mathbf{b})} : \ket{j} \mapsto e^{i\varepsilon(b_j-\frac12)}\ket{j}.\]
	Any algorithm that makes a call to $O_{\varepsilon}^{(\mathbf{b})}$ can without changing the outcome of the algorithm also make a call to $\overline{O}_{\varepsilon}^{(\mathbf{b})}$ instead, and vice versa, because the two differ only by a global phase. Similarly, if an algorithm makes a controlled call to $O_{\varepsilon}^{(\mathbf{b})}$, then it can equivalently make a controlled call to $\overline{O}_{\varepsilon}^{(\mathbf{b})}$, and subsequently apply some power of the $Z$-gate to the control qubit to correct for the mismatch in the phase between the two oracles. Hence, the number of queries required to solve the problem does not change when we switch from the oracle $O_{\varepsilon}^{(\mathbf{b})}$ to $\overline{O}_{\varepsilon}^{(\mathbf{b})}$.
	
	Note that the same also holds for the oracle that we saw in \cref{lem:highoverlapbitstring}. We can modify it in the same way, to arrive at the oracle $\overline{O}^{(\mathbf{b})}$, defined as
	\[\overline{O}^{(\mathbf{b})} : \ket{j} \mapsto i^{-1+2b_j}\ket{j},\]
	and solving the problem of finding a small $\ell_1$-distance bit string $\mathbf{b}^*$ to $\mathbf{b}$ takes equally many calls to $\overline{O}^{(\mathbf{b})}$ as it does to $O^{(\mathbf{b})}$.
	
	We can characterize how many queries we need to evaluate a relation $r$, given any particular oracle access, using the adversary bound for relations, taken from \cite{Belovs15}, Equation 21. The proof that its optimal value, denoted by $\ADV^{\pm}(r,\Delta)$, indeed equals the query complexity up to constants can be found in Theorem~32 of the same paper. The adversary bound in question is the SDP
	\begin{align*}
	\max \quad & \lambda_{\max}\left(\Gamma - \frac13 N\right) \\
	\text{s.t.} \quad & \norm{\Gamma \circ \Delta} \leq 1, \\
	& \Gamma \preceq N \circ E_{\mathbf{b}}, \qquad \text{for all } \mathbf{b} \in \{0,1\}^d.
	\end{align*}
	Here $\Gamma$, $N$ and $E_{\mathbf{b}}$ are Hermitian matrices indexed by $\{0,1\}^d$, and $N$ and $E_{\mathbf{b}}$ are diagonal. $\Delta$ is a family of matrices $\{\Delta_{\mathbf{b},\mathbf{b}'} : \mathbf{b}, \mathbf{b}' \in \{0,1\}^d\}$, where $\Delta_{\mathbf{b},\mathbf{b}'}$ equals the difference between the oracle for $\mathbf{b}$ and $\mathbf{b}'$. Hence, if we have access to the input via the regular phase oracles, we have $\Delta_{\mathbf{b},\mathbf{b}'} = \overline{O}^{(\mathbf{b})} - \overline{O}^{(\mathbf{b}')}$. The product $\Gamma \circ \Delta$ produces a bigger matrix, with blocks labeled by $\mathbf{b}, \mathbf{b}'$, containing $\Gamma_{\mathbf{b},\mathbf{b}'}\Delta_{\mathbf{b},\mathbf{b}'}$. The matrix entries $E_{\mathbf{b}}[\mathbf{b}',\mathbf{b}']$ are defined to be $1$ if $(\mathbf{b},\mathbf{b}') \not\in r$ and $0$ otherwise, and $N \circ E_{\mathbf{b}}$ denotes the entry-wise product between these two matrices.
	
	Let us now compute the two matrices $\Delta$ for the two different input models, i.e., for all $\mathbf{b},\mathbf{b}' \in \{0,1\}^d$, we have
	\[\Delta_{\mathbf{b},\mathbf{b}'} = \overline{O}^{(\mathbf{b})} - \overline{O}^{(\mathbf{b}')} \qquad \text{and} \qquad \Delta'_{\mathbf{b},\mathbf{b}'} = \overline{O}_{\varepsilon}^{(\mathbf{b})} - \overline{O}_{\varepsilon}^{(\mathbf{b}')}.\]
	Since all our oracles are diagonal operators, we immediately find that these newly-defined quantities $\Delta_{\mathbf{b},\mathbf{b}'}$ and $\Delta'_{\mathbf{b},\mathbf{b}'}$ are diagonal too. Furthermore, for all $j \in [d]$, we have that
	\[\bra{j}\Delta_{\mathbf{b},\mathbf{b}'}\ket{j} = \bra{j}\overline{O}^{(\mathbf{b})}\ket{j} - \bra{j}\overline{O}^{(\mathbf{b}')}\ket{j} = i^{-1+2b_j} - i^{-1+2b_j'} = 2i^{-1+2b_j}\left(1-\delta_{b_j,b_j'}\right),\]
	and
	\[\bra{j}\Delta'_{\mathbf{b},\mathbf{b}'}\ket{j} = \bra{j}\overline{O}_{\varepsilon}^{(\mathbf{b})}\ket{j} - \bra{j}\overline{O}_{\varepsilon}^{(\mathbf{b}')}\ket{j} = e^{i\varepsilon\left(b_j - \frac12\right)} - e^{i\varepsilon\left(b_j' - \frac12\right)} = 2i^{-1+2b_j}\sin\left(\frac{\varepsilon}{2}\right)\left(1 - \delta_{b_j,b_j'}\right).\]
	Hence, we deduce that
	\[\Delta' = \sin\left(\frac{\varepsilon}{2}\right)\Delta,\]
	which implies that the feasible region of the semidefinite program is enlarged by a factor of $\sin(\varepsilon/2)$ when we switch from $\Delta$ to $\Delta'$. Consequently, let $\Gamma$ and $N$ compose an optimal solution for the semidefinite program with constraint matrix $\Delta$. Then, we can construct an optimal solution to the SDP with constraint matrix $\Delta'$ by plugging in $\Gamma'$ and $N'$, where
	\[\Gamma' = \frac{\Gamma}{\sin\left(\frac{\varepsilon}{2}\right)}, \qquad \text{and} \qquad N' = \frac{N}{\sin\left(\frac{\varepsilon}{2}\right)},\]
	which in turn implies that
	\[\ADV^{\pm}(r,\Delta') = \frac{\ADV^{\pm}(r,\Delta)}{\sin\left(\frac{\varepsilon}{2}\right)} = \Omega\left(\frac{\ADV^{\pm}(r,\Delta)}{\varepsilon}\right), \qquad (\ADV^{\pm}(r,\Delta) \to \infty, \;\; \varepsilon \downarrow 0).\]
	Combining this with the fact that $\ADV^{\pm}(r,\Delta)$ is equal up to constants to the optimal query complexity of computing $r$ with $\Delta$ depending on the input model, we find that the number of queries required  to the fractional phase oracle to find a low $\ell_1$-distance bit string is in big-$\O$-notation $1/\varepsilon$ times as big as the number of queries required to the regular phase oracle. Together with the result from \cref{lem:highoverlapbitstring}, this completes the proof.
\end{proof}

Now, we are ready to define the specific instances of the multivariate Monte Carlo estimation problem that we will be using in the upcoming lower bound proofs. We only consider the path-independent setting of the problem, because the algorithm for this case uses the algorithm for the cumulative-depth and exact-depth settings as subroutines. Hence, if we can achieve tight bounds in the path-independent setting, tight bounds for the other settings follow immediately.

Let $S = \{s_0\}$ be the state space, with the transition probability function $P(s_0,s_0) = 1$. Let $R_S(s_0) = \mathbf{r} \in \R^d$, and $\gamma \in [0,1]$. The resulting instance is displayed in \cref{fig:LBR}.

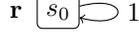
\begin{figure}[h!]
	\centering
	\begin{tikzpicture}[state/.style={draw, rounded corners=.2em},every loop/.style={min distance=2em,looseness=5}]
	\node[state] (s0) at (0,0) {$s_0$};
	\node[left=1em] at (0,0) {$\mathbf{r}$};
	\draw[->] (s0) to[loop right] node[right] {$1$} (s0);
	\end{tikzpicture}
	\caption{Graphical depiction of the instance of the multivariate Monte Carlo estimation problem we use in this section. There is only one state, $s_0$, which the Markovian walk loops in indefinitely. At every time step, a reward of $\mathbf{r} \in \R^d$ is obtained.}
	\label{fig:LBR}
\end{figure}

To ease notation, we define
\begin{equation}
	\label{eq:Tgamma}
	T_{\gamma} = \sum_{t=0}^T \gamma^t = \begin{cases}
		T, & \text{if } \gamma = 1, \\
		\frac{1-\gamma^{T+1}}{1-\gamma}, & \text{otherwise}, \\
	\end{cases}
\end{equation}
and we observe that the value function of these instances is $V(s_0) = T_{\gamma}\mathbf{r}$. Moreover,
\begin{equation}
	\label{eq:Tgammascaling}
	\frac{T_{\gamma}}{T^*} = \max\left\{\frac1T\sum_{t=0}^T \gamma^t, 1-\gamma^{T+1}\right\} \geq \max\left\{\frac{T+1}{T} \gamma^T, 1 - \gamma^{T+1}\right\} \geq \max\left\{\gamma^{T+1}, 1 - \gamma^{T+1}\right\} \geq \frac12,
\end{equation}
and hence $T_{\gamma} = \Omega(T^*)$, in the limit where $T^* \to \infty$.

Next, we use these instances to prove lower bounds on the query complexities to the reward oracles. We start with phase oracles, in \cref{lem:LBRphase}.

\begin{lemma}[Query complextity lower bounds with access to phase oracles]
	\label{lem:LBRphase}
	Let $d \in \N$, $O \in \R^{d \times d}$ an orthogonal matrix, $T \in \N \cup \{\infty\}$, $0 < \varepsilon < dR_{\max}$, $q \in [1,\infty]$, and $\gamma \in [0,1]$, such that if $T = \infty$, then $\gamma < 1$. Suppose that we have a quantum algorithm $\A$, making $Q_{\A}$ queries to a phase oracle $O_{R_S}$, that solves the path-independent multivariate Monte Carlo estimation problem rotated by $O$, with rewards bounded by $R_{\max}$ in $\ell_q$-norm, up to precision $\varepsilon$ in $\ell_1$-norm. Then the number of queries to $O_{R_S}$ it makes scales at least as
	\[\Omega\left(T^*\frac{R_{\max}}{\varepsilon}d^2\right), \qquad (T^*, R_{\max}, d \to \infty, \;\; \varepsilon \downarrow 0).\]
\end{lemma}

\begin{proof}
	Let
	\[\varepsilon' = \frac{4\varepsilon}{R_{\max}T_{\gamma}d}.\]
	Next, suppose that we have access to some bit string $\mathbf{b} \in \{0,1\}^d$ via the fractional phase oracle $O_{\varepsilon'}^{(\mathbf{b})}$. From \cref{lem:fractionalphasequerylowerbound}, we know that constructing a bit string $\mathbf{b}^* \in \{0,1\}^d$ that satisfies $\norm{O(\mathbf{b} - \mathbf{b}^*)}_1 \leq d/4$ takes a number of queries to $O_{\varepsilon'}^{(\mathbf{b})}$ that scales at least as
	\[\Omega\left(\frac{d}{\varepsilon'}\right) = \Omega\left(\frac{R_{\max}T_{\gamma}d^2}{\varepsilon}\right) = \Omega\left(\frac{R_{\max}T^*d^2}{\varepsilon}\right), \qquad (T^*, R_{\max}, d \to \infty, \;\; \varepsilon \downarrow 0).\]
	We now show that we can also find such a bit string $\mathbf{b}^*$ with $Q_{\A}$ queries to $O_{\varepsilon'}^{(\mathbf{b})}$. To that end, for $\mathbf{b} \in \{0,1\}^d$, let
	\[\mathbf{r}^{(\mathbf{b})} = \frac{8\varepsilon}{T_{\gamma}d}\mathbf{b}.\]
	The value function associated to an instance of the multivariate Monte Carlo estimation problem labeled by $\mathbf{b}$ and displayed in \cref{fig:LBR} now becomes
	\[V^{(\mathbf{b})}(s_0) = T_{\gamma}\mathbf{r}^{(\mathbf{b})} = \frac{8\varepsilon}{d}\mathbf{b},\]
	and hence by \cref{lem:smallell1distancereduction}, with $O'= I$ and $\mathbf{y} = \mathbf{0}$, $\A$ will be able to find a high-overlap bit string $\mathbf{b}^* \in \{0,1\}^d$ with high probability. Finally, note that the corresponding phase oracle, $O_{R_S}^{(\mathbf{b})}$, acts as
	\[O_{R_S}^{(\mathbf{b})} : \ket{s_0}\ket{j} \mapsto e^{i\frac{4\varepsilon}{T_{\gamma}dR_{\max}}b_j}\ket{s_0}\ket{j} = e^{i\varepsilon'b_j}\ket{s_0}\ket{j},\]
	and hence $O_{R_S}^{(\mathbf{b})} = O_{\varepsilon'}^{(\mathbf{b})}$. Thus $\A$ indeed queries $O_{\varepsilon'}^{(\mathbf{b})}$ a total number of $Q_{\A}$ times. This completes the proof.
\end{proof}

We use the same ideas to prove lower bounds when we have access to the probability oracle, in \cref{lem:LBRprobability}.

\begin{lemma}[Query complexity lower bounds with access to probability oracles]
	\label{lem:LBRprobability}
	Let $d \in \N$, $O \in \R^{d \times d}$ an orthogonal matrix, $T \in \N \cup \{\infty\}$, $0 < \varepsilon < dR_{\max}$, $q \in [1,\infty]$, and $\gamma \in [0,1]$, such that if $T = \infty$, then $\gamma < 1$. Suppose that we have a quantum algorithm $\A$, making $Q_{\A}$ queries to a probability oracle $U_{R_S}$, that solves the path-independent multivariate Monte Carlo estimation problem rotated by $O$, with rewards bounded by $R_{\max}$ in $\ell_q$-norm, up to precision $\varepsilon$ in $\ell_1$-norm. Then the number of queries to $U_{R_S}$ it makes scales at least as
	\[\Omega\left(T^*\frac{R_{\max}}{\varepsilon}d^{2-\frac{1}{2q}}\right), \qquad (T^*, R_{\max}, d \to \infty, \;\; \varepsilon \downarrow 0).\]
\end{lemma}

\begin{proof}
	Let 
	\[y = \arcsin\sqrt{\frac{1}{2d^{\frac1q}} - \frac{8\varepsilon}{T_{\gamma}dR_{\max}}}, \qquad \text{and} \qquad \varepsilon' = \arcsin\sqrt{\frac{1}{2d^{\frac1q}}} - \arcsin\sqrt{\frac{1}{2d^{\frac1q}} - \frac{8\varepsilon}{T_{\gamma}dR_{\max}}}.\]
	Suppose that we have access to some bit string $\mathbf{b} \in \{0,1\}^d$ via the fractional phase oracle $O_{\varepsilon'}^{(\mathbf{b})}$. From \cref{lem:fractionalphasequerylowerbound}, we know that constructing a bit string $\mathbf{b}^* \in \{0,1\}^d$ such that $\norm{O(\mathbf{b} - \mathbf{b}^*)}_1 \leq d/4$ takes a number of queries to $O_{\varepsilon'}^{(\mathbf{b})}$ that scales at least as
	\[\Omega\left(\frac{d}{\varepsilon'}\right), \qquad (d \to \infty, \;\; \varepsilon' \downarrow 0).\]
	We can analyze how $\varepsilon'$ scales in terms of $T^*$, $d$, $R_{\max}$ and $\varepsilon$ by expanding the $f_a(x) = \arcsin\sqrt{a+x}$ around $x = 0$, with $a \downarrow 0$, $x \to 0$. We find that
	\begin{align}
		f_a(x) &= f_a(0) + f_a'(0)x + \O(x^2) = \arcsin\sqrt{a} + \frac{1}{\sqrt{1-a}} \cdot \frac{x}{2\sqrt{a}} + \O(f_a''(0)x^2)\nonumber \\
		&= \arcsin\sqrt{a} + \frac{x}{2\sqrt{a(1-a)}} + \O\left(\frac{x^2}{a^{3/2}}\right).
		\label{eq:taylorexpansionarcsin}
	\end{align}
	When we plug in $a = 1/(2d^{1/q})$, then
	\begin{align*}
	\varepsilon' &= f_a(0) - f_a\left(-\frac{8\varepsilon}{T_{\gamma}dR_{\max}}\right) = \frac{8\varepsilon}{T_{\gamma}dR_{\max}} \cdot \frac{1}{2\sqrt{\frac{1}{2d^{\frac1q}}\left(1-\frac{1}{2d^{\frac1q}}\right)}} + \O\left(\frac{\varepsilon^2}{R_{\max}^2T_{\gamma}^2d^{2-\frac{3}{2q}}}\right) \\
	&= \O\left(\frac{\varepsilon\sqrt{d^{\frac1q}}}{T_{\gamma}dR_{\max}}\right) = \O\left(\frac{\varepsilon}{T_{\gamma}d^{1-\frac{1}{2q}}R_{\max}}\right), \qquad (T_{\gamma},R_{\max},d \to \infty, \;\; \varepsilon \downarrow 0),
	\end{align*}
	and hence, recovering a high-overlap bit string $\mathbf{b}^* \in \{0,1\}^d$ given access to the oracle $O_{\varepsilon'}^{(\mathbf{b})}$ requires a number of queries that scales at least as
	\[\Omega\left(\frac{d}{\varepsilon'}\right) = \Omega\left(\frac{T_{\gamma}R_{\max}}{\varepsilon}d^{2-\frac{1}{2q}}\right) = \Omega\left(\frac{T^*R_{\max}}{\varepsilon}d^{2-\frac{1}{2q}}\right), \qquad (T^*, d, R_{\max} \to \infty, \;\; \varepsilon \downarrow 0).\]
	We now show that we could also construct a high-overlap bit string $\mathbf{b}^*$ with a number of calls to $O_{\varepsilon'}^{(\mathbf{b})}$ that is only $2Q_{\A}$. This is sufficient to prove that $Q_{\A}$ has to scale at least as quickly as stated in the lemma. To that end, for every $\mathbf{b} \in \{0,1\}^d$, we let
	\[\mathbf{r}^{(\mathbf{b})} = \left(\frac{R_{\max}}{2d^{\frac1q}} - \frac{8\varepsilon}{T_{\gamma}d}\right)\mathbf{1} + \frac{8\varepsilon}{T_{\gamma}d}\mathbf{b}.\]
	If we use these vectors $\mathbf{r}^{(\mathbf{b})}$ in the instance to the multivariate Monte Carlo estimation problem, as displayed in \cref{fig:LBR}, the value function becomes
	\[V^{(\mathbf{b})}(s_0) = T_{\gamma}\mathbf{r}^{(\mathbf{b})} = \left(\frac{R_{\max}T_{\gamma}}{2d^{\frac1q}} - \frac{8\varepsilon}{d}\right)\mathbf{1} + \frac{8\varepsilon}{d}\mathbf{b}.\]
	Hence, by \cref{lem:smallell1distancereduction}, with $O' = I$, the algorithm $\A$ is able to construct a high-overlap bit string $\mathbf{b}^*$ with $Q_{\A}$ calls to the probability oracle $U_{R_S}^{(\mathbf{b})}$. Now, we define the modified oracle $\overline{O}_{\varepsilon'}^{(\mathbf{b})}$, that differs from $O_{\varepsilon'}^{(\mathbf{b})}$ merely by a global phase:
	\begin{equation}
		\label{eq:modifiedfractionalphaseoracle}
		\overline{O}_{\varepsilon'}^{(\mathbf{b})} : \ket{j} \mapsto e^{iy + \varepsilon'b_j}\ket{j},
	\end{equation}
	which we can call in a controlled manner simply by performing one call to $O_{\varepsilon'}^{(\mathbf{b})}$, and applying some power of the Pauli-$Z$-gate to the control qubit. Now, suppose that we start in the state
	\[\ket{s_0}\ket{j}\ket{+},\]
	we apply the operation $\overline{O}_{\varepsilon'}^{(\mathbf{b})}$ to the second register conditioned on the final register being in state $\ket{0}$, and we apply its inverse when the final register is in state $\ket{1}$. Then, we obtain the state
	\[\ket{s_0}\ket{j}\frac{1}{\sqrt{2}}\left(e^{i(y + \varepsilon'b_j)}\ket{0} + e^{-i(y + \varepsilon'b_j)}\ket{1}\right) = \ket{s_0}\ket{j}\left(\cos\left(y+\varepsilon'b_j\right)\ket{+} + i\sin\left(y+\varepsilon'b_j\right)\ket{-}\right).\]
	With some extra single qubit gates on the last qubit can be turned into the following operation:
	\[\ket{s_0}\ket{j}\ket{0} \mapsto \ket{s_0}\ket{j}\left(\sin(y+\varepsilon'b_j)\ket{1} + \cos(y+\varepsilon'b_j)\ket{0}\right),\]
	which is a probability oracle to the function
	\begin{align*}
	& R_{\max}\sin^2\left(y\mathbf{1}+\varepsilon'\mathbf{b}\right) \\
	=\;\; & R_{\max}\sin^2\left(\arcsin\sqrt{\frac{1}{2d^{\frac1q}} - \frac{8\varepsilon}{T_{\gamma}dR_{\max}}}\mathbf{1} + \left(\arcsin\sqrt{\frac{1}{2d^{\frac1q}}} - \arcsin\sqrt{\frac{1}{2d^{\frac1q}} - \frac{8\varepsilon}{T_{\gamma}dR_{\max}}} \right)\mathbf{b}\right) \\
	=\;\; & \left(\frac{R_{\max}}{2d^{\frac1q}} - \frac{8\varepsilon}{T_{\gamma}d}\right)\mathbf{1} + \frac{8\varepsilon}{T_{\gamma}d}\mathbf{b} = \mathbf{r}^{(\mathbf{b})},
	\end{align*}
	where in the last line we used that $\mathbf{b} \in \{0,1\}^d$ and hence each entry in $\mathbf{b}$ is either $0$ or $1$. Thus, with two calls to $O_{\varepsilon'}^{(\mathbf{b})}$, we have constructed the probability oracle $U_{R_S}^{(\mathbf{b})}$, meaning that $Q_{\A}$ has to satisfy the earlier derived lower bound too. This completes the proof.
\end{proof}

Finally, we prove a lower bound when we have access to a distribution oracle, in \cref{lem:LBRdistribution}.

\begin{lemma}[Query complextity lower bounds with access to distribution oracles]
	\label{lem:LBRdistribution}
	Let $d \in \N$, $O \in \R^{d \times d}$ an orthogonal matrix, $T \in \N \cup \{\infty\}$, $0 < \varepsilon < dR_{\max}$, $q \in [1,\infty]$, and $\gamma \in [0,1]$, such that if $T = \infty$, then $\gamma < 1$. Suppose that we have a quantum algorithm $\A$, making $Q_{\A}$ queries to a probability oracle $D_{R_S}$, that solves the path-independent multivariate Monte Carlo estimation problem rotated by $O$, with rewards bounded by $R_{\max}$ in $\ell_q$-norm, up to precision $\varepsilon$ in $\ell_1$-norm. Then the number of queries to $D_{R_S}$ it makes scales at least as
	\[\Omega\left(T^*\frac{R_{\max}}{\varepsilon}d^{2-\frac1q}\right), \qquad (T^*, R_{\max}, d \to \infty, \;\; \varepsilon \downarrow 0).\]
\end{lemma}

\begin{proof}
	We follow a similar argument as in the previous lemma, and refer to that proof on several occasions for brevity. Let
	\[y = \arcsin\sqrt{\frac12 - \frac{8\varepsilon}{T_{\gamma}d^{1-\frac1q}R_{\max}}}, \qquad \text{and} \qquad \varepsilon' = \arcsin\sqrt{\frac12} - \arcsin\sqrt{\frac12 - \frac{8\varepsilon}{T_{\gamma}d^{1-\frac1q}R_{\max}}}.\]
	Then, using the expansion derived in \cref{eq:taylorexpansionarcsin}, we obtain that
	\[\varepsilon' = \O\left(\frac{\varepsilon}{T_{\gamma}d^{1-\frac1q}R_{\max}}\right), \qquad (T_{\gamma},R_{\max},d \to \infty, \;\; \varepsilon \downarrow 0),\]
	which implies that in order to get a high-overlap bit string $\mathbf{b}^*$, given access to a fractional phase oracle $O_{\varepsilon'}^{(\mathbf{b})}$, \cref{lem:fractionalphasequerylowerbound} tells us that we need to call it a number of times that scales at least as
	\[\Omega\left(\frac{d}{\varepsilon'}\right) = \Omega\left(\frac{T_{\gamma}R_{\max}}{\varepsilon}d^{2-\frac1q}\right) = \Omega\left(\frac{T^*R_{\max}}{\varepsilon}d^{2-\frac1q}\right), \qquad (T^*,R_{\max},d \to \infty, \;\; \varepsilon \downarrow 0).\]
	We use the instances to the multivariate Monte Carlo estimation problem that is depicted in \cref{fig:LBR}, where for every bit string $\mathbf{b} \in \{0,1\}^d$, we use the reward vector
	\[\mathbf{r}^{(\mathbf{b})} = \left(\frac{R_{\max}}{2d^{\frac1q}} - \frac{8\varepsilon}{T_{\gamma}d}\right)\mathbf{1} + \frac{8\varepsilon}{T_{\gamma}d}\mathbf{b}.\]
	Using identical reasoning as in the proof of \cref{lem:LBRprobability}, it now suffices to prove that we can implement the distribution oracle $D_{R_S}^{(\mathbf{b})}$ with $2$ calls to $\overline{O}_{\varepsilon'}^{(\mathbf{b})}$, defined in \cref{eq:modifiedfractionalphaseoracle}. Using the same construction as in the proof of \cref{lem:LBRprobability}, we can construct the following operation with $2$ calls to $\overline{O}_{\varepsilon'}^{(\mathbf{b})}$, acting as
	\[\ket{s_0}\ket{0}\ket{0} \mapsto \ket{s_0}\left(\sum_{j=1}^d \frac{\sin\left(y+\varepsilon'b_j\right)}{\sqrt{d}}\ket{j}\ket{1} + \sum_{j=0}^d \frac{\cos\left(y+\varepsilon'b_j\right)}{\sqrt{d}}\ket{j}\ket{0}\right).\]
	This is a distribution oracle evaluating the function
	\begin{align*}
	& \frac{R_{\max}}{d^{\frac1q}}\sin^2\left(y\mathbf{1}+\varepsilon'\mathbf{b}\right) \\
	=\;\; & \frac{R_{\max}}{d^{\frac1q}}\sin^2\left(\arcsin\sqrt{\frac12 - \frac{8\varepsilon}{T_{\gamma}d^{1-\frac1q}R_{\max}}}\mathbf{1} + \left(\arcsin\sqrt{\frac12} - \arcsin\sqrt{\frac12 - \frac{8\varepsilon}{T_{\gamma}d^{1-\frac1q}R_{\max}}}\right)\mathbf{b}\right) \\
	=\;\; & \left(\frac{R_{\max}}{2d^{\frac1q}} - \frac{8\varepsilon}{T_{\gamma}d}\right)\mathbf{1} + \frac{8\varepsilon}{T_{\gamma}d}\mathbf{b} = \mathbf{r}^{(\mathbf{b})},
	\end{align*}
	and hence we have constructed $D_{R_S}^{(\mathbf{b})}$ with $2$ calls to $\overline{O}_{\varepsilon}^{(\mathbf{b})}$. This implies that the number of calls to $D_{R_S}^{(\mathbf{b})}$ made by $\A$ must be at least the lower bound we derived earlier, completing the proof.
\end{proof}

\subsection{Lower bounds on the query complexity to the probability transition matrix}
\label{subsec:LBP}

In this subsection, we focus on lower bounds on the number of calls we need to make to the state transition matrix oracle, $D_P$, in order to solve the multivariate Monte Carlo estimation problem. We provide a lower bound for the path-independent setting, from which we derive lower bounds for the other settings as well.

Similar to the previous section, we use the hardness of the high-overlap bit-string problem. Contrary to the previous section, though, we do not assume fractional phase oracle access, but instead we provide a construction in which the individual bits in the high-overlap problem are computed via a composition of a majority and a parity function. In \cref{lem:DPlowerboundhardness}, we analyze the hardness of the resulting composed problem.

\begin{lemma}[Composition of high-overlap bit string problem with majority and parity]
	\label{lem:DPlowerboundhardness}
	Let $d,k,T \in \N$, with $k$ and $T$ odd. Define the domain
	\begin{equation}
		\label{eq:defD}
		D_{d,k,T} = \left\{\mathbf{b} \in \{0,1\}^d \times \{0,1\}^k \times \{0,1\}^T : \forall j \in [d], \sum_{\ell=1}^k \Parity((b_{j,\ell,t})_{t=1}^T) \in \left\{\left\lfloor \frac{k}{2}\right\rfloor, \left\lceil \frac{k}{2} \right\rceil\right\}\right\},
	\end{equation}
	and for every element $\mathbf{b} \in D_{d,k,T}$, let $\mathbf{c}_{\mathbf{b}} \in \{0,1\}^d$ be such that for all $j \in [d]$,
	\begin{equation}
		\label{eq:defc}
		(c_{\mathbf{b}})_j = \begin{cases}
			1, & \text{if } \sum_{\ell=1}^k \Parity((b_{j,\ell,t})_{t=1}^T) = \lceil k/2 \rceil, \\
			0, & \text{otherwise}.
		\end{cases}
	\end{equation}
	Now suppose that the input element from $\mathbf{b} \in D_{d,k,T}$ is encoded in a phase oracle $O_{\mathbf{b}}$, acting on three registers containing $d$, $k$ and $T$ orthogonal states each, and defined as
	\[O_{\mathbf{b}} : \ket{j}\ket{\ell}\ket{t} \mapsto (-1)^{b_{j,\ell,t}}\ket{j}\ket{\ell}\ket{t}.\]
	Then, any quantum algorithm that finds a vector $\mathbf{c}^* \in \{0,1\}^d$ such that $\norm{O(\mathbf{c}_{\mathbf{b}} - \mathbf{c}^*)}_1 \leq d/4$ needs to make a number of queries to $O_{\mathbf{b}}$ that scales at least as $\Omega\left(dkT\right)$, in the limit $d,k,T \to \infty$.
\end{lemma}

\begin{proof}
	First, we define the function $f : D'_{k,T} \to \{0,1\}$, with
	\[D'_{k,T} = \left\{\mathbf{b} \in \{0,1\}^k \times \{0,1\}^T : \sum_{\ell=1}^k \Parity((b_{\ell,t})_{t=1}^T) \in \left\{\left\lfloor \frac{k}{2} \right\rfloor, \left\lceil \frac{k}{2} \right\rceil \right\}\right\},\]
	as
	\[f(\mathbf{b}) = \begin{cases}
		1, & \text{if } \sum_{\ell=1}^k \Parity((b_{\ell,t})_{t=1}^T) = \lceil k/2 \rceil, \\
		0, & \text{otherwise}.
	\end{cases}\]
	Since this function is the composition of majority on $k$ bits, with inputs having Hamming weight either $k/2-1$ or $k/2$, and parity on $T$ bits, we know by the composition theorem, Theorem 1.5 in \cite{Reichardt10}, that the query complexity of this function is $\Theta(k \cdot T)$, with $k,T \to \infty$.
	
	Next, we define the relation $r \subseteq \{0,1\}^d \times \{0,1\}^d$, as
	\[r = \left\{(\mathbf{b},\mathbf{b}^*) \in \{0,1\}^d \times \{0,1\}^d : \norm{O(\mathbf{b} - \mathbf{b}^*)}_1 \leq \frac{d}{4}\right\}.\]
	From \cref{lem:highoverlapbitstring}, we know that it takes at least $\Omega(d)$ queries to evaluate this relation, where $d \to \infty$.
	
	We now observe that the problem posed in the lemma statement is the composition of these two problems. If we have $2$ elements in $\mathbf{b}, \mathbf{b}^* \in D_{d,k,T}$, then we say that they are related by the relation $r \bullet f$ if $(f(b_{j,\cdot,\cdot}))_{j=1}^d$ and $(f(b_{j,\cdot,\cdot}^*))_{j=1}^d$ are related by $r$. It remains to show that evaluating $r \bullet f$ takes at least the product of the number of queries it takes to evaluate $r$ and $f$ individually.
	
	To that end, suppose that we have optimal adversary matrices $\Gamma_r, N_r \in \R^{\{0,1\}^d \times \{0,1\}^d}$ for the adversary bound for the relation $r$, given by Equation 21 in \cite{Belovs15},
	\begin{align*}
	\max \;\; & \lambda_{\max} \left(\Gamma - \frac13N\right) \\
	\text{s.t.} \;\; & \norm{\Gamma \circ \Delta_r} \leq 1 \\
	& \Gamma \preceq N \circ E_b, \qquad \forall b \in \{0,1\},
	\end{align*}
	and that we have the optimal adversary matrix $\Gamma_f \in \R^{D'_{k,T} \times D'_{k,T}}$ for the adversary bound for the function $f$, given by Equation 24 in \cite{Belovs15},
	\begin{align*}
	\max \;\; & \norm{\Gamma} \\
	\text{s.t.} \;\; & \norm{\Gamma \circ \Delta_f} \leq 1 \\
	& \Gamma[f^{-1}(b),f^{-1}(b)] = 0, \qquad \forall b \in \{0,1\}.
	\end{align*}
	We now follow the proof from Section 6.1 in \cite{HLS07}. Specifically, we construct an explicit solution to the dual adversary bound for the relation $r \bullet f$ from $\Gamma_r$, $N_r$ and $\Gamma_f$, in a way that is inspired by Definition 6 in \cite{HLS07}. That is, we let $\Gamma_{r \bullet f} \in \R^{D_{d,k,T} \times D_{d,k,T}}$ be such that for all $x,y \in D$,
	\[\Gamma_{r \bullet f}[x,y] = \frac{2}{\norm{\Gamma_f}^{d-1}} \cdot \Gamma_r[(f(x_{j,\cdot,\cdot}))_{j=1}^d,(f(x_{j,\cdot,\cdot}))_{j=1}^d] \cdot \prod_{j=1}^d \begin{cases}
	\norm{\Gamma_f}\delta_{(x_{j,\cdot,\cdot}),(y_{j,\cdot,\cdot})}, & \text{if } f(x_{j,\cdot,\cdot}) = f(y_{j,\cdot,\cdot}), \\
	\Gamma_f[x_{j,\cdot,\cdot},y_{j,\cdot,\cdot}], & \text{otherwise},
	\end{cases}\]
	and we let $N_{r \bullet f} \in \R^{D_{d,k,T} \times D_{d,k,T}}$ be a diagonal matrix such that for all $x \in D_{d,k,T}$,
	\[N_{r \bullet f}[x,x] = 2\norm{\Gamma_f} \cdot N_r[(f(x_{j,\cdot,\cdot}))_{j=1}^d,(f(x_{j,\cdot,\cdot}))_{j=1}^d].\]
	Note that these definitions differ from Definition 6 in \cite{HLS07} by a prefactor of $2/\norm{\Gamma_f}^{d-1}$. On top of that, note that the matrix $D_f$ in \cite{HLS07} differs with a factor of $2$ from the matrix $\Delta_f$ from \cite{Belovs15}. With these subtleties in mind, the sentence at the top of page 21 of \cite{HLS07} implies that
	\begin{align}
	\norm{\Gamma_{r \bullet f} \circ \Delta_{r \bullet f}} &=  \frac{4}{\norm{\Gamma_f}^{d-1}} \norm{\frac{\norm{\Gamma_f}^{d-1}}{2}\Gamma_{r \bullet f} \circ \frac12\Delta_{r \bullet f}} = \frac{4}{\norm{\Gamma_f}^{d-1}} \norm{\Gamma_r \circ \frac12\Delta_r} \cdot \norm{\Gamma_f \circ \frac12\Delta_f} \cdot \norm{\Gamma_f}^{d-1}\nonumber \\
	&= \norm{\Gamma_r \circ \Delta_r} \cdot \norm{\Gamma_f \circ \Delta_f} \leq 1.
	\label{eq:compositionconstraintbound}
	\end{align}
	We also observe that, for all $x,y \in D_{d,k,T}$ and $b \in \{0,1\}$,
	\begin{align}
		\left[\Gamma_{r \bullet f} - N_{r \bullet f} \circ E_{r \bullet f, b}\right][x,y] = \; &\frac{2}{\norm{\Gamma_f}^{d-1}} \cdot \left[\Gamma_r - N_r \circ E_{r, b}\right][(f(x_{j,\cdot,\cdot}))_{j=1}^d, (f(x_{j,\cdot,\cdot}))_{j=1}^d] \nonumber \\
		&\cdot \prod_{j=1}^d \begin{cases}
		\norm{\Gamma_f}\delta_{(x_{j,\cdot,\cdot}),(y_{j,\cdot,\cdot})}, & \text{if } f(x_{j,\cdot,\cdot}) = f(y_{j,\cdot,\cdot}), \\
		\Gamma_f[x_{j,\cdot,\cdot},y_{j,\cdot,\cdot}], & \text{otherwise}.
		\end{cases}
		\label{eq:matrixproduct}
	\end{align}
	Thus, neglecting the constant in front, the matrix $A = \Gamma_{r \bullet f} - N_{r \bullet f} \circ E_{r \bullet f}$ can be written as the element-wise product of two matrices. Hence, in order to prove that $A \preceq 0$, by Schur's product theorem, it suffices to show that the left and right factors in \cref{eq:matrixproduct} are $\preceq 0$ and $\succeq 0$, respectively.
	
	We start with the right factor, since we can easily observe that this can be written as
	\[(\Gamma_f + \norm{\Gamma_f}I)^{\otimes d}[x,y] = \prod_{j=1}^d (\Gamma_f + \norm{\Gamma_f}I)[x_{j,\cdot,\cdot},y_{j,\cdot,\cdot}] = \prod_{j=1}^d \begin{cases}
	\norm{\Gamma_f}\delta_{(x_{j,\cdot,\cdot}),(y_{j,\cdot,\cdot})}, & \text{if } f(x_{j,\cdot,\cdot}) = f(y_{j,\cdot,\cdot}), \\
	\Gamma_f[x_{j,\cdot,\cdot},y_{j,\cdot,\cdot}], & \text{otherwise}.
	\end{cases}\]
	Moreover, it trivially holds that $\Gamma_f + \norm{\Gamma_f}I$ is positive semidefinite, and hence so is its $d$-fold tensor product.
	
	Next, we focus at the left factor in \cref{eq:matrixproduct}. To that end, let $\mathbf{b} \in D'_{k,T} \subseteq \{0,1\}^k \times \{0,1\}^T$, and let $\overline{\mathbf{b}}$ be the bit-wise negation of $\mathbf{b}$. Then, since $T$ is odd, we find that
	\[\sum_{\ell=1}^k \Parity((\overline{b}_{\ell,t})_{t=1}^T) = \sum_{\ell=1}^k (1 - \Parity((b_{\ell,t})_{t=1}^T))  = k - \sum_{\ell=1}^T \Parity((b_{\ell,t})_{t=1}^T),\]
	and since $k$ is odd, we find that $f(\overline{\mathbf{b}}) = 1 - f(\mathbf{b})$. Thus, bit-wise negation defines a bijection between $0$- and $1$-instances of the function $f$, and hence there are equally many of them. This implies that there are equally many instances $\mathbf{b} \in D$, for any resulting bit string $(f(b_{j,\cdot,\cdot}))_{j=1}^d$ as well. Thus, the left product in \cref{eq:matrixproduct} can be rearranged by permuting the rows and columns in such a way that every entry of $\Gamma_r - N_r \circ E_{r,b}$ gets blown up to a block of all equal values. This multiplies all elements in the spectrum by the same factor, and hence from $\Gamma_r - N_r \circ E_{r,b} \preceq 0$, we find that this bigger matrix is $\preceq 0$ as well. Thus, we find that $\Gamma_{r \bullet f} - N_{r \bullet f} \circ E_{r \bullet f,b} \preceq 0$, ensuring that the tuple $(\Gamma_{r \bullet f}, N_{r \bullet f})$ is indeed a feasible solution to the general adversary bound for the relation $r \bullet f$.
	
	Finally, without loss of generality, we can assume that there exists an eigenvector $\mathbf{v}_f \in \R^{D'_{k,T}}$ of $\Gamma_f$ with eigenvalue $\norm{\Gamma_f}$, because if it does not exist, then $\Gamma_f$ must admit an eigenvector with eigenvalue $-\norm{\Gamma_f}$, in which case we can always multiply $\Gamma_f$ with $-1$, without leaving the feasible region. Similarly, let $\mathbf{v}_r \in \R^d$ be an eigenvector of $\Gamma_r - N_r/3$, with largest eigenvalue $\lambda$. Next, define $\mathbf{v}_{r \bullet f} \in \R^{D_{d,k,T}}$, with
	\[v_{r \bullet f}[x] = v_r[(f(x_{j,\cdot,\cdot}))_{j=1}^d] \cdot \prod_{j=1}^d v_f[x_{j,\cdot,\cdot},y_{j,\cdot,\cdot}] = v_r[(f(x_{j,\cdot,\cdot}))_{j=1}^d] \cdot (\mathbf{v}_f)^{\otimes d}[x,y].\]
	Then, following the same reasoning as in the proof of Lemma 16 in \cite{HLS07}, we find that $v_{r \bullet f}$ is an eigenvector of $\Gamma_{r \bullet f} - N_{r \bullet f}/3$, with eigenvalue $\lambda \cdot \norm{\Gamma_f}$. Thus, we find that
	\[\lambda_{\max}\left(\Gamma_{r \bullet f} - \frac13N_{r \bullet f}\right) \geq \lambda_{\max} \left(\Gamma_r - \frac13N_r\right) \cdot \norm{\Gamma_f},\]
	and so the optimal value of the adversary bound of the relation $r \bullet f$ is at least as large the product of the optimal values of the adversary bounds of the relation $r$ and the function $f$. This completes the proof.
\end{proof}

Finally, we construct instances of the multivariate Monte Carlo estimation problem that feature a value function of the form considered in \cref{lem:smallell1distancereduction}, and whose probability transition oracle can be constructed from the oracle considered in \cref{lem:DPlowerboundhardness}. The depth of the instances that we are using depends on the choice of $T$ and $\gamma$. This amounts to encoding the composition of the majority and parity functions in a carefully crafted instance. Specifically, we define the quantity
\[T' = \max\left\{32,2\left(2\left\lfloor\frac{T_{\gamma}}{4}\right\rfloor - 1\right)\right\},\]
where we use the definition of $T_{\gamma}$ from \cref{eq:Tgamma}. It follows directly from \cref{eq:Tgammascaling} that $T' = \Omega(T^*)$, when $T^* \to \infty$. The resulting instances are graphically depicted in \cref{fig:DPlowerboundinstances}.

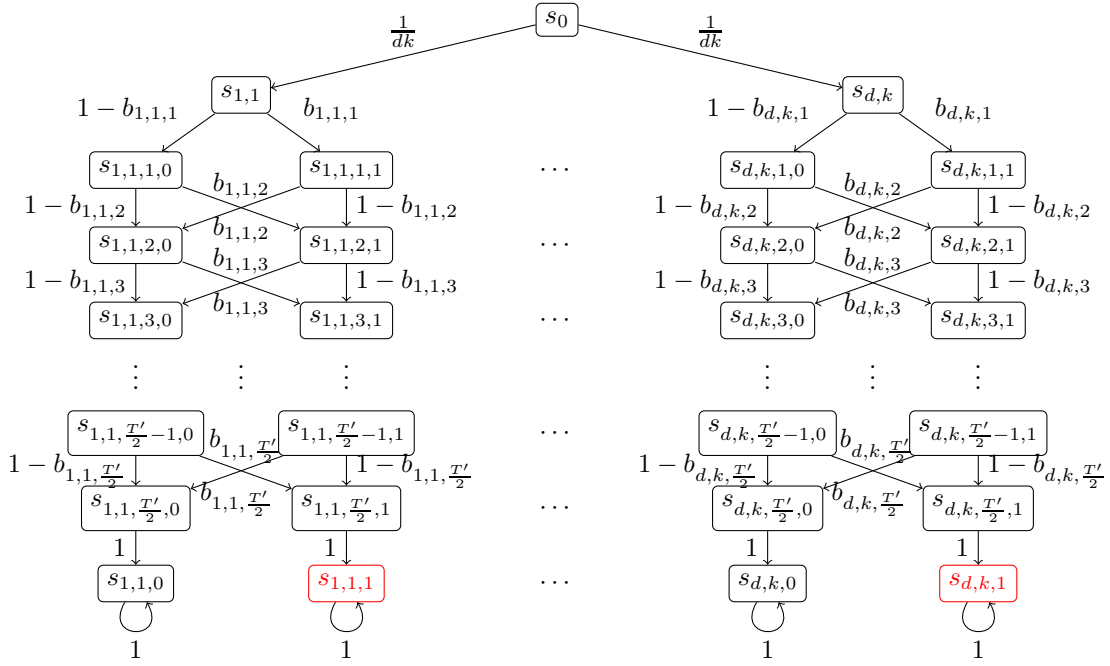
\begin{figure}[h!]
	\centering
	\begin{tikzpicture}[xscale=1.4,yscale=-1,state/.style={draw, rounded corners=.2em},every loop/.style={min distance=3em,looseness=10}]
	\begin{scope}[shift={(-3,1)}]
	\node[state] (s110) at (0,0) {$s_{1,1}$};
	\node[state] (s10) at (-1,1) {$s_{1,1,1,0}$};
	\node[state] (s11) at (1,1) {$s_{1,1,1,1}$};
	\node[state] (s20) at (-1,2) {$s_{1,1,2,0}$};
	\node[state] (s21) at (1,2) {$s_{1,1,2,1}$};
	\node[state] (s30) at (-1,3) {$s_{1,1,3,0}$};
	\node[state] (s31) at (1,3) {$s_{1,1,3,1}$};
	\node at (-1,3.65) {$\vdots$};
	\node at (0,3.65) {$\vdots$};
	\node at (1,3.65) {$\vdots$};
	\node[state] (sT10) at (-1,4.5) {$s_{1,1,\frac{T'}{2}-1,0}$};
	\node[state] (sT11) at (1,4.5) {$s_{1,1,\frac{T'}{2}-1,1}$};
	\node[state] (sT0) at (-1,5.5) {$s_{1,1,\frac{T'}{2},0}$};
	\node[state] (sT1) at (1,5.5) {$s_{1,1,\frac{T'}{2},1}$};
	\node[state] (terminal0) at (-1,6.5) {$s_{1,1,0}$};
	\node[state,red] (terminal1) at (1,6.5) {$s_{1,1,1}$};
	\draw[->] (s110) to node[above left] {$1-b_{1,1,1}$} (s10);
	\draw[->] (s110) to node[above right] {$b_{1,1,1}$} (s11);
	\draw[->] (s10) to node[left] {$1-b_{1,1,2}$} (s20);
	\draw[->] (s10) to node[above] {$b_{1,1,2}$} (s21);
	\draw[->] (s11) to node[below] {$b_{1,1,2}$} (s20);
	\draw[->] (s11) to node[right] {$1-b_{1,1,2}$} (s21);
	\draw[->] (s20) to node[left] {$1-b_{1,1,3}$} (s30);
	\draw[->] (s20) to node[above] {$b_{1,1,3}$} (s31);
	\draw[->] (s21) to node[below] {$b_{1,1,3}$} (s30);
	\draw[->] (s21) to node[right] {$1-b_{1,1,3}$} (s31);
	\draw[->] (sT10) to node[left] {$1-b_{1,1,\frac{T'}{2}}$} (sT0);
	\draw[->] (sT10) to node[above] {$b_{1,1,\frac{T'}{2}}$} (sT1);
	\draw[->] (sT11) to node[below] {$b_{1,1,\frac{T'}{2}}$} (sT0);
	\draw[->] (sT11) to node[right] {$1-b_{1,1,\frac{T'}{2}}$} (sT1);
	\draw[->] (sT0) to node[left] {$1$} (terminal0);
	\draw[->] (sT1) to node[left] {$1$} (terminal1);
	\draw[->] (terminal0) to[loop above] node[below] {$1$} (terminal0);
	\draw[->] (terminal1) to[loop above] node[below] {$1$} (terminal1);
	\end{scope}
	\begin{scope}[shift={(3,1)}]
	\node[state] (t0) at (0,0) {$s_{d,k}$};
	\node[state] (t10) at (-1,1) {$s_{d,k,1,0}$};
	\node[state] (t11) at (1,1) {$s_{d,k,1,1}$};
	\node[state] (t20) at (-1,2) {$s_{d,k,2,0}$};
	\node[state] (t21) at (1,2) {$s_{d,k,2,1}$};
	\node[state] (t30) at (-1,3) {$s_{d,k,3,0}$};
	\node[state] (t31) at (1,3) {$s_{d,k,3,1}$};
	\node at (-1,3.65) {$\vdots$};
	\node at (0,3.65) {$\vdots$};
	\node at (1,3.65) {$\vdots$};
	\node[state] (tT10) at (-1,4.5) {$s_{d,k,\frac{T'}{2}-1,0}$};
	\node[state] (tT11) at (1,4.5) {$s_{d,k,\frac{T'}{2}-1,1}$};
	\node[state] (tT0) at (-1,5.5) {$s_{d,k,\frac{T'}{2},0}$};
	\node[state] (tT1) at (1,5.5) {$s_{d,k,\frac{T'}{2},1}$};
	\node[state] (tterminal0) at (-1,6.5) {$s_{d,k,0}$};
	\node[state,red] (tterminal1) at (1,6.5) {$s_{d,k,1}$};
	\draw[->] (t0) to node[above left] {$1-b_{d,k,1}$} (t10);
	\draw[->] (t0) to node[above right] {$b_{d,k,1}$} (t11);
	\draw[->] (t10) to node[left] {$1-b_{d,k,2}$} (t20);
	\draw[->] (t10) to node[above] {$b_{d,k,2}$} (t21);
	\draw[->] (t11) to node[below] {$b_{d,k,2}$} (t20);
	\draw[->] (t11) to node[right] {$1-b_{d,k,2}$} (t21);
	\draw[->] (t20) to node[left] {$1-b_{d,k,3}$} (t30);
	\draw[->] (t20) to node[above] {$b_{d,k,3}$} (t31);
	\draw[->] (t21) to node[below] {$b_{d,k,3}$} (t30);
	\draw[->] (t21) to node[right] {$1-b_{d,k,3}$} (t31);
	\draw[->] (tT10) to node[left] {$1-b_{d,k,\frac{T'}{2}}$} (tT0);
	\draw[->] (tT10) to node[above] {$b_{d,k,\frac{T'}{2}}$} (tT1);
	\draw[->] (tT11) to node[below] {$b_{d,k,\frac{T'}{2}}$} (tT0);
	\draw[->] (tT11) to node[right] {$1-b_{d,k,\frac{T'}{2}}$} (tT1);
	\draw[->] (tT0) to node[left] {$1$} (tterminal0);
	\draw[->] (tT1) to node[left] {$1$} (tterminal1);
	\draw[->] (tterminal0) to[loop above] node[below] {$1$} (tterminal0);
	\draw[->] (tterminal1) to[loop above] node[below] {$1$} (tterminal1);
	\end{scope}
	\node[state] (s0) at (0,0) {$s_0$};
	\node at (0,2) {$\cdots$};
	\node at (0,3) {$\cdots$};
	\node at (0,4) {$\cdots$};
	\node at (0,5.5) {$\cdots$};
	\node at (0,6.5) {$\cdots$};
	\node at (0,7.5) {$\cdots$};
	\draw[->] (s0) to node[above] {$\frac{1}{dk}$} (s110);
	\draw[->] (s0) to node[above] {$\frac{1}{dk}$} (t0);
	\end{tikzpicture}
	\caption{The instances we use in the lower bound proof on the query complexity to the probability-transition oracle, in \cref{lem:DPlowerboundproof}. The input consists of bits $b_{j,\ell,t} \in \{0,1\}$, with $j \in [d]$, $\ell \in [k]$, and $t \in [T'/2]$, where $k \in \N$. Only the red nodes give non-zero rewards: the node $s_{j,k,1}$ gives reward $R_{\max}\mathbf{e}_j$ if $q \in [1,2]$, and $R_{\max}d^{\frac12-\frac1q}H\mathbf{e}_j$, if $q \in (2,\infty]$, where $H$ is a $d$-dimensional Hadamard matrix.}
	\label{fig:DPlowerboundinstances}
\end{figure}

\begin{lemma}[Lower bound on the query complexity to the probability transition oracle]
	\label{lem:DPlowerboundproof}
	Let $d \in \N$, $\gamma \in [0,1]$, $q \in [1,\infty]$, $0 < \varepsilon < R_{\max}(d/2)^{\max\{0,\frac12-\frac1q\}}$ and $T \in \N \cup \{\infty\}$, such that if $T = \infty$, then $\gamma < 1$. Let $O \in \R^{d \times d}$ be an orthogonal matrix. Suppose that $\A$ is a quantum algorithm solving the multivariate Monte Carlo estimation problem rotated by $O$ up to precision $\varepsilon$ w.r.t.\ the $\ell_1$-norm, with $Q_P$ calls to the probability transition oracle $D_P$. Then,
	\[Q_P = \Omega\left((T^*)^2\frac{R_{\max}}{\varepsilon}d^{\max\{1,\frac32-\frac1q\}}\right), \qquad (T^*,R_{\max},d \to \infty, \;\; \varepsilon \downarrow 0).\]
\end{lemma}

\begin{proof}
	We first focus on the case where $q \in [1,2]$. In this case, we need to prove that the exponent of $d$ in the asymptotic lower bound on the query complexity is $1$. To that end, let
	\[k = 2\left\lfloor \frac{T'R_{\max}}{32\varepsilon} \right\rfloor - 1, \qquad \text{and} \qquad \varepsilon' = \frac{T'R_{\max}}{16k},\]
	and immediately observe that $k \geq 1$ and $\varepsilon' \geq \varepsilon$.
	
	Next, suppose we have a bit string $\mathbf{b} \in D_{d,k,T'}$, where $D_{d,k,T'}$ is defined in \cref{eq:defD}. We define the following instance of the multivariate Monte Carlo estimation problem. Let the state space be $S = \{s_0\} \cup \{s_{\ell,0,0} : \ell \in [k]\} \cup \{s_{\ell,t,b} : \ell \in [k], t \in [T'/2], b \in \{0,1\}\}$. Let the probability transition function be defined as follows, for all $\ell \in [k]$, $t \in [T'/2]$, and $b,b' \in \{0,1\}$,
	\[P(s_0,s_{\ell,0,0}) = \frac{1}{dk}, \qquad P(s_{\ell,\frac{T'}{2},b},s_{\ell,\frac{T'}{2},b}) = 1,\]
	\[P(s_{\ell,t-1,b},s_{\ell,t,b'}) = \begin{cases}
		1, & \text{if } (b = b' \land b_{\ell,t} = 0) \lor (b \neq b' \land b_{\ell,t} = 1), \\
		0, & \text{otherwise},
	\end{cases}\]
	and $0$ elsewhere. Let the state-reward function be
	\[R_S\left(s_{\ell,\frac{T'}{2},1}\right) = R_{\max}\mathbf{e}_j,\]
	and $0$ elsewhere. It immediately follows that for all states $\norm{R_S(s)}_q \leq \norm{R_S(s)}_1 \leq R_{\max}$, and hence $R_S$ is bounded by $R_{\max}$ in $\ell_q$-norm.
	
	For every element $\mathbf{b} \in D_{d,k,T'}$, let $\mathbf{c}_{\mathbf{b}} \in \{0,1\}^d$ be defined as in \cref{eq:defc}. Now, we can write the value function of the instance labeled by $\mathbf{b}$ as
	\[V^{(\mathbf{b})}(s_0) = \frac{1}{dk} \cdot \sum_{j=1}^d \sum_{\ell=1}^k \Parity((b_{j,\ell,t})_{t=1}^{T'}) \cdot \frac{T'}{2} \cdot R_{\max}\mathbf{e}_j = \frac{T'R_{\max}}{2dk} \cdot \left\lfloor \frac{k}{2} \right\rfloor\mathbf{1} + \frac{T'R_{\max}}{2dk}\mathbf{c}_{\mathbf{b}} = \mathbf{y} + \frac{8\varepsilon'}{d}\mathbf{c}_{\mathbf{b}},\]
	where we defined $\mathbf{y} = T'R_{\max}/(2dk) \lfloor k/2\rfloor \mathbf{1}$. Since $\A$ is $\varepsilon$-precise and $\varepsilon \leq \varepsilon'$, $\A$ is also $\varepsilon'$-precise, and hence from \cref{lem:smallell1distancereduction}, with $O' = I$, we find that with a single run of algorithm $\A$ we can find a vector $\mathbf{c}^* \in \{0,1\}^d$ such that with probability at least $2/3$, $\norm{O(\mathbf{c}_{\mathbf{b}} - \mathbf{c}^*)}_1 \leq d/4$. Furthermore, the probability transition oracle of this particular instance of the multivariate Monte Carlo estmiation problem can be constructed using a single call to the bit oracle that provides access to the individual bits in $\mathbf{b}$. This implies, using \cref{lem:DPlowerboundhardness}, that we need to make a number of calls to $D_P$ that scales as
	\[\Omega\left(dkT'\right) = \Omega\left(d \cdot \frac{R_{\max}T'}{\varepsilon} \cdot T'\right) = \Omega\left((T')^2\frac{R_{\max}}{\varepsilon}d\right) = \Omega\left((T^*)^2\frac{R_{\max}}{\varepsilon}d\right), \qquad (T^*, R_{\max}, d \to \infty, \;\; \varepsilon \downarrow 0).\]
	This completes the proof in the case where $q \in [1,2]$.
	
	On the other hand, suppose that $q \in (2,\infty]$. Then, let $H \in \R^{d' \times d'}$ be a normalized Hadamard matrix with $d' \geq d/2$. Such a Hadamard matrix always exists, because they can be trivially constructed for every dimension that is a power of $2$. We let
	\[k = 2\left\lfloor \frac{T'R_{\max}(d')^{\frac12-\frac1q}}{32\varepsilon}\right\rfloor - 1, \qquad \text{and} \qquad \varepsilon' = \frac{T'R_{\max}(d')^{\frac12-\frac1q}}{16k},\]
	and again observe that $k \geq 1$ and $\varepsilon' \geq \varepsilon$. We use the same state space and probability transition matrix as before, but with $d'$ instead of $d$, but now we use the rewards
	\[R_S(s_{\ell,\frac{T}{2},1}) = (d')^{\frac12-\frac1q}R_{\max}H\mathbf{e}_j.\]
	One readily verifies that all rewards are bounded by $R_{\max}$ in $\ell_q$-norm, since
	\[\norm{R_S(s_{\ell,\frac{T}{2},1})}_q = (d')^{\frac12-\frac1q}R_{\max} \norm{H\mathbf{e}_j}_q \leq (d')^{\frac12-\frac1q}R_{\max} \cdot (d')^{\frac1q}\norm{H\mathbf{e}_j}_{\infty} = (d')^{\frac12}R_{\max} \cdot \frac{1}{\sqrt{d}} = R_{\max}.\]
	Again, for all $\mathbf{b} \in D_{d',k,T'}$, we define $\mathbf{c}_{\mathbf{b}} \in \{0,1\}^{d'}$ as in \cref{eq:defc}, which allows us to write the value function as
	\begin{align*}
		V^{(\mathbf{b})}(s_0) &= \frac{1}{d'k} \sum_{j=1}^{d'} \sum_{\ell=1}^k \Parity((b_{j,\ell,t})_{t=1}^{T'}) \cdot \frac{T'}{2} \cdot (d')^{\frac12-\frac1q}R_{\max} H\mathbf{e}_j \\
		&= \frac{T'R_{\max}(d')^{\frac12-\frac1q}}{2d'k} \left\lfloor \frac{k}{2} \right\rfloor H\mathbf{1} + \frac{T'R_{\max}(d')^{\frac12-\frac1q}}{2d'k}H\mathbf{c}_{\mathbf{b}} = \mathbf{y} + \frac{8\varepsilon'}{d}H\mathbf{c}_{\mathbf{b}},
	\end{align*}
	where we defined $\mathbf{y} = T'R_{\max}(d')^{\frac12-\frac1q}/(2d'k)\lfloor k/2 \rfloor H\mathbf{1}$. Thus, through \cref{lem:smallell1distancereduction}, with $O' = H$, we find that with a single run of $\A$, we can produce a vector $\mathbf{c}^* \in \{0,1\}^d$ such that with probability at least $2/3$, $\norm{OH(\mathbf{c}_{\mathbf{b}} - \mathbf{c}^*)}_1 \leq d'/4$. Since $OH$ is also an orthogonal matrix, we can use similar reasoning as before to arrive at a lower bound of
	\[\Omega(d'kT') = \Omega\left(d' \cdot \frac{R_{\max}T'(d')^{\frac12-\frac1q}}{\varepsilon} \cdot T'\right) = \Omega\left((T')^2\frac{R_{\max}}{\varepsilon}d^{\frac32-\frac1q}\right) = \Omega\left((T^*)^2\frac{R_{\max}}{\varepsilon}d^{\frac32-\frac1q}\right),\]
	in the limit where $T^*, R_{\max}, d \to \infty$ and $\varepsilon \downarrow 0$. This completes the proof in the case where $q \in (2,\infty]$ as well.
\end{proof}

\subsection{Results}

We now aggregate all results from the lower bound section into the following theorem.

\begin{theorem}[Query complexity lower bounds]
    \label{thm:mvmc-lower-bounds}
	Let $d \in \N$, $O \in \R^{d \times d}$ an orthogonal matrix, $T \in \N \cup \{\infty\}$, $0 < \varepsilon < R_{\max}$, $\gamma \in [0,1]$, $p,q \in [1,\infty]$, such that if $T = \infty$, then $\gamma < 1$. Let $S$ a state space, $P : S \times S \to [0,1]$ a probability transition matrix. The table below specifies the query complexity of the algorithm that solves the multivariate Monte Carlo estimation problem rotated by $O$, with reward vectors bounded by $R_{\max}$ in $\ell_q$-norm, up to error $\varepsilon$ in $\ell_p$-norm, with high probability:
	\begin{center}
		\begin{tabular}{r|c|c}
			Case & Calls to reward oracle & Calls to probability transition oracle \\\hline
			Exact-depth & $\displaystyle\Omega\left(\frac{R_{\max}}{\varepsilon}d^{\frac1p} \cdot \xi(d,q)\right)$ & $\displaystyle\Omega\left(\frac{R_{\max}}{\varepsilon}d^{\frac1p} \cdot d^{\max\{0,\frac12-\frac1q\}}\right)$ \\
			Cumulative-depth & $\displaystyle\Omega\left(T^*\frac{R_{\max}}{\varepsilon}d^{\frac1p} \cdot \xi(d,q)\right)$ & $\displaystyle\Omega\left(\frac{R_{\max}}{\varepsilon}d^{\frac1p} \cdot d^{\max\{0,\frac12-\frac1q\}}\right)$ \\
			Path-independent & $\displaystyle\Omega\left(T^*\frac{R_{\max}}{\varepsilon}d^{\frac1p} \cdot \xi(d,q)\right)$ & $\displaystyle\Omega\left(\frac{R_{\max}}{\varepsilon}d^{\frac1p} \cdot d^{\max\{0,\frac12-\frac1q\}}\right)$
		\end{tabular}
	\end{center}
	where the lower bounds for the probability transition oracle complexity only hold in the regime where $\varepsilon < R_{\max}(d/2)^{\max\{0,1/2-1/q\}+1/p-1}$, and $\xi$ depends on the access model to the reward function, listed in the following table:
	\begin{center}
		\begin{tabular}{r|c}
			Access model & $\xi(d,q)$ \\\hline
			Phase oracle & $d$ \\
			Probability oracle & $d^{1-\frac{1}{2q}}$ \\
			Distribution oracle & $d^{1-\frac1q}$ \\
			Lattice oracle on $G_O$ & $d^{1-\frac1q}$
		\end{tabular}
	\end{center}
	The $\Omega$-notation holds in the limit where $T^*,R_{\max},d \to \infty$ and $\varepsilon \downarrow 0$.
\end{theorem}

\begin{proof}
	All statements in this theorem in the case where $p = 1$ follow directly from \cref{lem:LBRphase,lem:LBRprobability,lem:LBRdistribution,lem:DPlowerboundproof}. For the other values of $p \in (1,\infty]$, suppose that we have a quantum algorithm $\A$ that solves the multivariate Monte Carlo estimation problem up to precision $\varepsilon$ w.r.t.\ the $\ell_p$-norm. Then, by H\"older's inequality, we find that it solves the problem up to precision $\varepsilon' = d^{1-\frac1p}\varepsilon$ w.r.t.\ the $\ell_1$-norm. We can now plug this $\varepsilon'$ into our lower bounds for the $\ell_1$-case, and obtain the claimed complexities.
\end{proof}

There is one important subtlety remaining in \cref{thm:mvmc-lower-bounds} that we highlight here, which is that the lower bounds for the probability transition oracle only hold in the regime where the precision $\varepsilon$ is upper bounded by $R_{\max}(d/2)^{-1+1/p+\max\{0,1/2-1/q\}}$. Hence, we make no lower bound statement about the query complexity if $\varepsilon$ is bigger than this value. We expect some other effects to come into play in this regime, because in the related setting considered by \cite{dagum00} where one estimates a probability distribution up to $\ell_2$-norm given sampling access to it, there surprisingly exists a \textit{classical} algorithm that runs in $\widetilde{\O}(1/\varepsilon^2)$ samples, and hence no lower bound depending polynomially on the dimension can exist. Nailing down the query complexity in this regime would be an interesting question for further research.

\section{Applications}
\label{sec:applications}

In this section, we describe some applications of our results. We first show an application in the context of policy-based reinforcement learning, where our formulation of the MVMC problem in terms of Markov reward processes is particularly convenient. We then show how the general problem of estimating the expectation values of mutually-commuting observables on a given quantum state can be viewed as an instance of the MVMC problem. We also list applications in quantum machine learning where this problem arises. 

\subsection{Policy-based reinforcement learning}
\label{subsec:rl-application}

In the context of reinforcement learning \cite{sutton98}, an agent-environment interaction is described by a Monte Carlo process where, sequentially, the agent acts probabilistically on the environment, the latter updates probabilistically its state depending on the actions of the agent, and then issues a reward. A common description of this interaction is in terms of Markov decision processes (MDP) \cite{sutton98}, where the probability distributions involved in this Monte Carlo process are assumed to be Markovian. Notably, the agent's actions are sampled from a stationary policy $\pi(a|s)$, i.e., a probability distribution over actions $a$ given a state $s$, and the environment dynamics are described by a transition probability function $P(s'|s,a)$. Interestingly, when we assume that this policy is fixed, the MDP can be viewed as a Markov reward process (MRP, i.e., an instance of our MVMC processes) by absorbing the action of the policy $\pi(a|s)$ in the environment transitions $P(s'|s) = P(s'|s,a)\pi(a|s)$. However, this remains an MRP with a 1-dimensional reward at this point.

The goal of the reinforcement learning problem then is to find a policy that maximizes the resulting value function $V_\pi(s_0)$ of its MRP. To do this, policy-based algorithms define a certain family of parametrized policies $\pi_{\bm{\theta}}(a|s) \in \Pi_{\bm{\theta}}$ (e.g., deep neural networks) and explore this policy family using gradient ascent on the value function $V_{\pi_{\bm{\theta}}}(s_0)$. The so-called policy gradient theorem \cite{sutton00} gives a formulation of the gradient of the value function as:
\begin{equation}
    \nabla_{\bm{\theta}} V_{\pi_{\bm{\theta}}}(s_0) = \underset{\tau \sim P(T;s_0)}{\E}\left[ \sum_{t=0}^T \nabla_{\bm{\theta}} \log(\pi_{\bm{\theta}}(a_t|s_t)) \sum_{t'=t}^T \gamma^{t'} r_{t'}\right]
\end{equation}
for an MDP/MRP of depth $T \in \mathbb{N} \cup {\infty}$. Notice that this is in turn the cumulative-depth value function of an MRP with a multidimensional reward vector given by $R^{(t')}(s_0, \dots, s_{t'}) = \gamma^{t'} r_{t'}\sum_{t=0}^{t'} \nabla_{\bm{\theta}} \log(\pi_{\bm{\theta}}(a_t|s_t))$ for each step $t'$ of the interaction (where the actions $a_t$ have been absorbed in their associated states $s_t$).\\
Given $\ell_p$-norm constraints on the gradients $\nabla_{\bm{\theta}} \log(\pi_{\bm{\theta}}(a_t|s_t))$ of the log-policy family and a bound $R_\textrm{max}$ on the rewards $r_t$ of the MDP, we can therefore make claims on the query complexity of estimating the gradient $\nabla_{\bm{\theta}} V_{\pi_{\bm{\theta}}}(s_0)$ and hence the expected speed-ups in policy-gradient reinforcement learning. Notably, only when $\nabla_{\bm{\theta}} \log(\pi_{\bm{\theta}}(a_t|s_t))$ is bounded in $\ell_1$-norm can we guarantee that a quadratic speed-up in gradient evaluation (in $\ell_\infty$-norm) without an associated slowdown in the dimension of $\bm{\theta}$, up to logarithmic factors.

\subsection{Estimating expectation values of commuting observables}

The second problem we apply our results to is that of computing expectation values of mutually commuting observables for a preparable quantum state $\ket{\psi}$. Quantum algorithms have been studied for the univariate version of this problem (i.e., for one observable) \cite{knill07,wocjan09}, but as we show in this subsection, the multivariate case appears as well naturally in machine learning applications.

\subsubsection{Problem definition}

Let $U$ be a unitary transformation and let $O_1, \ldots, O_d$ be $d$ mutually-commuting observables (Hermitian operators), acting both on a complex Hilbert space spanned by $n$ qubits. We define their corresponding Monte Carlo process as measuring one sample per observable\footnote{Sampling from an observable here means that we measure one of its eigenstates according to the Born rule on state $\ket{\psi}$ and the measurement outcome is defined as the corresponding eigenvalue of this eigenstate.} on the quantum state $\ket{\psi} = U\ket{0^{\otimes n}}$. We want to compute estimates of the $d$ expectation values $\expval{O_i} = \bra{\psi}O_i\ket{\psi}$, again up to some $\varepsilon$ error in $\ell_p$-norm.\\
Since all the observables commute, they all share a common eigenbasis $\left\{\ket{\phi_j}\right\}_{1\leq j\leq 2^n}$. Hence, measuring $\ket{\psi}$ in this eigenbasis allows to generate one sample for each of the observables $O_1, \ldots, O_d$.\footnote{When the observables do not commute, one cannot ``parallelize" measurements in such a manner, and would then be required to use more complicated techniques like shadow tomography \cite{aaronson19,huang21}.} In this common eigenbasis, we can also compare the eigenvalues $\lambda_{i,j} \in \mathbb{R}$ associated to the basis states $\ket{\phi_j}$ by each observable $O_i$: we assume the vectors $\bm{\lambda}_{.,j} = (\lambda_{1,j}, \ldots, \lambda_{d,j})$ to be bounded in some $\ell_q$-norm, for all $j$.

When $U$ is a unitary that simulates a classical probabilistic computation (i.e., creates a superposition of computational basis states without relative phases) and all the observables $O_i$ are diagonal in the computational basis, this problem can be trivially formulated as depth-one MVMC estimation. We show however that is also the case of general unitaries $U$ and general observables $O_i$ that all mutually commute (but not necessarily in the computational basis). For this, associate the basis states $\ket{s_j}$ to the shared basis states $\ket{\phi_j}$ of these observables, such that the action of any unitary $U$ can be re-written as:
\begin{equation}\label{eq:unitary-oracle}
    D_P = I \otimes U : \ket{0}\ket{0^{\otimes n}} \mapsto \ket{0}\sum_{s_j \in S} e^{i\varphi_j} \sqrt{P(0,s_j)} \ket{s_j}
\end{equation}
and access to the observables $O_i$ can be made, e.g., through a phase oracle of the form:
\begin{equation}
O_R : \ket{s_{j}}\ket{i} \mapsto e^{i\frac{\lambda_{i,j}}{2R_\textrm{max}}}\ket{s_{j}}\ket{i}
\end{equation}
and similarly for the other types of oracle access we consider.\\
Note that the relative phases $e^{i\varphi_j}$ appearing in \cref{eq:unitary-oracle} do not contribute the expectation values $\expval{O_i} = \sum_{s_j \in S} P(0,s_j) \lambda_{i,j}$ and moreover do not come into play in the application of our quantum algorithms (notably in \cref{lem:valuefunctionoracleexactdepth}), as they can be absorbed in the states $\ket{s_j}$ in this depth-one case. Hence, all our results (upper and lower bounds) are applicable to this problem.

\subsubsection{Examples of applications}

\paragraph{Training variational quantum circuits}
A straightforward application fitting this problem definition appears in some variational quantum algorithms for machine learning \cite{benedetti19}. In a multidimensional regression setting \cite{mitarai18} or a reinforcement learning setting \cite{jerbi21b,skolik21}, a variational quantum circuit defined by a parametrized and data-dependent unitary $U(\bm{x},\bm{\theta})$ and a set of observables $(O_1, \ldots, O_d)$ can be used as a hypothesis family $f_{\bm{\theta}}(\bm{x}) = (\expval{O_1}_{\bm{x},\bm{\theta}}, \ldots, \expval{O_d}_{\bm{x},\bm{\theta}})$, for $\expval{O_i}_{\bm{x},\bm{\theta}} = \bra{0^{\otimes n}} U^\dagger(\bm{x},\bm{\theta}) O_i U(\bm{x},\bm{\theta})\ket{0^{\otimes n}}$, to model target functions $g$ with $d$-dimensional outputs. When the observables $O_1, \ldots, O_d$ all commute (e.g., weighted tensor products of Pauli-Z operators or projectors on some basis states), the problem of estimating $f_{\bm{\theta}}(\bm{x})$ fits the problem definition above.

\paragraph{Training Boltzmann machines}
Another application considers the problem of estimating updates of a Boltzmann machine in a machine learning setting (e.g., a classification or generative modeling problem) \cite{wiebe16,wiebe19,kieferova17,jerbi21}. Take for instance a Boltzmann machine defined by a Hamiltonian:
\begin{equation}
H = \sum_{i<j} J_{i,j} \sigma^z_i\sigma^z_j + \sum_{i} b_i\sigma^z_i
\end{equation}
where $J_{i,j}$ and $b_i$ are real weights and biases and $\sigma^z_i$ is a Pauli-$Z$ operator acting on a qubit $i$ out of $n$ total qubits. The updates on the weights and biases of this Boltzmann machine take the form:
\begin{equation}
\Delta J_{i,j} = -\mathcal{L}(\bm{J},\bm{b})\expval{\sigma^z_i\sigma^z_j}, \quad \Delta b_{i} = -\mathcal{L}(\bm{J},\bm{b})\expval{\sigma^z_i}
\end{equation}
where $\mathcal{L}(\bm{J},\bm{b})$ is a loss dependent on the Boltzmann machine performance at the machine learning task and the expectation values $\expval{\sigma^z_i}, \expval{\sigma^z_i\sigma_z^j}$ are with respect to the Gibbs state:
\begin{equation}\label{eq:Gibbs-state}
\ket{\psi} = \frac{1}{\textrm{Tr}_{\bm{x}}[e^{-H}]} \sum_{\bm{x}} \sqrt{e^{-H(\bm{x})}} \ket{\bm{x}}
\end{equation}
for computational basis states $\ket{\bm{x}}$.\\
Assume having access to a unitary $U$ that prepares the Gibbs state of Eq.~(\ref{eq:Gibbs-state}), e.g., using one of the subroutines in \cite{wiebe16,wiebe19,kieferova17,jerbi21}, then estimating the updates of the Boltzmann machine is an instance of the problem above for observables $\left\{ -\mathcal{L}(\bm{J},\bm{b})\sigma^z_i\sigma^z_j,\ -\mathcal{L}(\bm{J},\bm{b})\sigma^z_i\right\}_{i,j}$, i.e., weighted $\sigma^z_i$ and $\sigma^z_i\sigma^z_j$ operators, which are all diagonal in the computational basis.

\subsubsection{Implications of our results}
Since the eigenvalue decomposition of the observables in the applications presented above is known in general, we can can therefore implement all the oracle access required by our quantum algorithms, making them applicable to these problems. Moreover, our lower bound results indicate that, in the case of training Boltzmann machines, all quantum encodings of the observables we consider lead to a trade-off between a quadratic speed-up in the precision of the updates and an exponential slowdown in the number of parameters $\abs{\bm{\theta}}$. As for training variational quantum circuits, we can only guarantee that this trade-off won't appear in the case where the eigenvalues of the observables $O_i$ satisfy $\norm{\bm{\lambda}_{.,j}}_1 \leq R_\textrm{max},\ \forall j$, i.e., when they are bounded in $\ell_1$-norm for any given shared eigenstate $\ket{\phi_j}$.

\section{Discussion \& Outlook}
\label{sec:MMCdiscussion}

To the best of our knowledge, the results obtained in this text provide the first complete characterization of the quantum query complexity of a multivariate problem, when one has oracular access to the individual variables via any of the quantum oracles outlined in \cref{subsec:oracles}. This suggests that the quantum algorithmic techniques and lower bounds outlined in this document are exhaustive in this setting. We expect these techniques to be relevant in studying the query complexities in other problems that use a similar access model as well. For instance, the observations made in this document might help in closing the optimality gap that remains in the gradient estimation problem when the partial derivatives of the objective function are bounded by Gevrey conditions, as considered in \cite{GAW18} and \cite{Cor19}, but it is not immediately clear whether the construction of the hard instances considered in this paper carry over directly to a setting where the asymptotics of higher-order partial derivatives of a Gevrey function play a non-trivial role as well.

The oracle conversions considered in \cref{subsec:rewardoracleconversions} give rise to more interesting questions that fall outside the scope of this research. In particular, there are some directed edges missing in the graph displayed in \cref{fig:oracleconversions}. In our results, we did not need these conversions, but it would be an interesting direction for further research to nail down the optimal complexities of the other oracle conversions too, since we expect them to be useful in other use cases.

Coming back on a remark made in the introduction, we only considered here input oracles for the Monte Carlo random variables that are natural for quantum algorithms, i.e., encode the information in phases or amplitudes. This leaves as an open question whether speed-ups with respect to the precision of the estimates with better scaling in the dimension of the random variables are possible in the more general access model of binary oracles (which can be converted to all the oracles we consider without any overhead). However, to the best of our knowledge, we do not know of any quantum algorithms that do not rely on an encoding in phase or amplitude as part of their processing, as covered by the input oracles we consider. Moreover, proving lower bounds in the binary setting would be challenging due to the ability to simulate classical algorithms from MVMC (e.g., the algorithm in \cref{appdx:classical-complexity-MVMC}), which trade off speed-ups in precision for a better dependence in the dimension of the random variables. Assuming a speed-up without any associated slowdown is impossible in the general case, these lower bounds would then need to feature this trade-off in query complexity, which we are not aware is possible to show using existing tools for lower bounds.

When comparing quantum and classical query complexities for the MVMC problem, it only makes sense to compare the query complexity to the transition probability oracles with our considerations. We know of an exact scaling with respect to the dimension $d$ for the case $p,q=\infty$ (see \cref{appdx:classical-complexity-MVMC}). For the special case studied by \cite{vanApeldoorn21} of estimating probability vectors (i.e., $q=1$), we also know of classical algorithms with query complexity $\widetilde{\Theta} \left( \varepsilon^{-2} \right)$ for $p\in\{2,\infty\}$ \cite{dagum00} and $\widetilde{\mathcal{O}} \left( d \varepsilon^{-2} \right)$ for $p=1$ \cite{kamath15}. We leave open a more detailed analysis of the classical query complexity of this problem for general $p,q$, as to fully characterize the slowdown in the dimension $d$ associated to our quantum algorithms, as well as nailing down the query complexity to the probability transition oracle in the high-error regime.

\section*{Acknowledgments}

First of all, both authors would like to thank Vedran Dunjko and M\={a}ris Ozols for many insightful discussions and motivating us to think about the problems considered in this paper. AJ would like to thank Ronald de Wolf for multiple insightful and helpful discussions, and Joran van Apeldoorn for interesting conversations and providing an early version of his manuscript. Finally, AJ would like to extend his gratitude to the anonymous legends that answered the math overflow question posted here: \cite{Mathoverflow}. SJ acknowledges support from the Austrian Science Fund (FWF) through the projects DK-ALM:W1259-N27 and SFB BeyondC F7102. SJ also acknowledges the Austrian Academy of Sciences as a recipient of the DOC Fellowship.

\bibliographystyle{alpha}
\bibliography{references}

\newcommand{\etalchar}[1]{$^{#1}$}
\begin{thebibliography}{ADFDJ03}

\bibitem[Aar19]{aaronson19}
Scott Aaronson.
\newblock Shadow tomography of quantum states.
\newblock {\em SIAM Journal on Computing}, 49(5):STOC18--368, 2019.

\bibitem[ADFDJ03]{andrieu03}
Christophe Andrieu, Nando De~Freitas, Arnaud Doucet, and Michael~I Jordan.
\newblock An introduction to mcmc for machine learning.
\newblock {\em Machine learning}, 50(1):5--43, 2003.

\bibitem[Amb04]{Ambainis03}
A.~Ambainis.
\newblock Quantum walk algorithm for element distinctness.
\newblock In {\em 45th Annual IEEE Symposium on Foundations of Computer
  Science}, pages 22--31, 2004.

\bibitem[BCH{\etalchar{+}}12]{binder12}
Kurt Binder, David~M Ceperley, J-P Hansen, MH~Kalos, DP~Landau, D~Levesque,
  H~Mueller-Krumbhaar, D~Stauffer, and J-J Weis.
\newblock {\em Monte Carlo methods in statistical physics}, volume~7.
\newblock Springer Science \& Business Media, 2012.

\bibitem[Bel15]{Belovs15}
Aleksandrs Belovs.
\newblock Variations on quantum adversary, 2015.

\bibitem[BHMT02]{brassard02}
Gilles Brassard, Peter Hoyer, Michele Mosca, and Alain Tapp.
\newblock Quantum amplitude amplification and estimation.
\newblock {\em Contemporary Mathematics}, 305:53--74, 2002.

\bibitem[BLSF19]{benedetti19}
Marcello Benedetti, Erika Lloyd, Stefan Sack, and Mattia Fiorentini.
\newblock Parameterized quantum circuits as machine learning models.
\newblock {\em Quantum Science and Technology}, 4(4):043001, 2019.

\bibitem[CEG95]{canetti95}
Ran Canetti, Guy Even, and Oded Goldreich.
\newblock Lower bounds for sampling algorithms for estimating the average.
\newblock {\em Information Processing Letters}, 53(1):17--25, 1995.

\bibitem[Cor18]{Cor18}
Arjan Cornelissen.
\newblock Quantum gradient estimation and its application to quantum
  reinforcement learning.
\newblock Master's thesis, Delft University of Technology, Sep 2018.

\bibitem[Cor19]{Cor19}
Arjan Cornelissen.
\newblock Quantum gradient estimation of gevrey functions.
\newblock {\em arXiv preprint arXiv:1909.13528}, 2019.

\bibitem[DKLR00]{dagum00}
Paul Dagum, Richard Karp, Michael Luby, and Sheldon Ross.
\newblock An optimal algorithm for monte carlo estimation.
\newblock {\em SIAM Journal on computing}, 29(5):1484--1496, 2000.

\bibitem[FG06]{FG06}
J{\"o}rg Flum and Martin Grohe.
\newblock {\em Parameterized complexity theory}.
\newblock Springer Science \& Business Media, 2006.

\bibitem[FGGS99]{FGGS99}
Edward Farhi, Jeffrey Goldstone, Sam Gutmann, and Michael Sipser.
\newblock Bound on the number of functions that can be distinguished
  withkquantum queries.
\newblock {\em Physical Review A}, 60(6):4331–4333, Dec 1999.

\bibitem[GAW19]{GAW18}
András Gilyén, Srinivasan Arunachalam, and Nathan Wiebe.
\newblock Optimizing quantum optimization algorithms via faster quantum
  gradient computation.
\newblock {\em Proceedings of the Thirtieth Annual ACM-SIAM Symposium on
  Discrete Algorithms}, page 1425–1444, Jan 2019.

\bibitem[Gla13]{glasserman13}
Paul Glasserman.
\newblock {\em Monte Carlo methods in financial engineering}, volume~53.
\newblock Springer Science \& Business Media, 2013.

\bibitem[GSLW19]{GSLW18}
András Gilyén, Yuan Su, Guang~Hao Low, and Nathan Wiebe.
\newblock Quantum singular value transformation and beyond: exponential
  improvements for quantum matrix arithmetics.
\newblock {\em Proceedings of the 51st Annual ACM SIGACT Symposium on Theory of
  Computing}, Jun 2019.

\bibitem[HKP21]{huang21}
Hsin-Yuan Huang, Richard Kueng, and John Preskill.
\newblock Efficient estimation of pauli observables by derandomization.
\newblock {\em arXiv preprint arXiv:2103.07510}, 2021.

\bibitem[HLS07]{HLS07}
Peter Hoyer, Troy Lee, and Robert Spalek.
\newblock Negative weights make adversaries stronger.
\newblock {\em Proceedings of the thirty-ninth annual ACM symposium on Theory
  of computing - STOC ’07}, 2007.

\bibitem[JGM{\etalchar{+}}21]{jerbi21b}
Sofiene Jerbi, Casper Gyurik, Simon Marshall, Hans~J Briegel, and Vedran
  Dunjko.
\newblock Variational quantum policies for reinforcement learning.
\newblock {\em arXiv preprint arXiv:2103.05577}, 2021.

\bibitem[Jor05]{jordan05}
Stephen~P Jordan.
\newblock Fast quantum algorithm for numerical gradient estimation.
\newblock {\em Physical review letters}, 95(5):050501, 2005.

\bibitem[JTN{\etalchar{+}}21]{jerbi21}
Sofiene Jerbi, Lea~M Trenkwalder, Hendrik~Poulsen Nautrup, Hans~J Briegel, and
  Vedran Dunjko.
\newblock Quantum enhancements for deep reinforcement learning in large spaces.
\newblock {\em PRX Quantum}, 2(1):010328, 2021.

\bibitem[KOPS15]{kamath15}
Sudeep Kamath, Alon Orlitsky, Dheeraj Pichapati, and Ananda~Theertha Suresh.
\newblock On learning distributions from their samples.
\newblock In {\em Conference on Learning Theory}, pages 1066--1100. PMLR, 2015.

\bibitem[KOS07]{knill07}
Emanuel Knill, Gerardo Ortiz, and Rolando~D Somma.
\newblock Optimal quantum measurements of expectation values of observables.
\newblock {\em Physical Review A}, 75(1):012328, 2007.

\bibitem[KW17]{kieferova17}
M{\'a}ria Kieferov{\'a} and Nathan Wiebe.
\newblock Tomography and generative training with quantum boltzmann machines.
\newblock {\em Physical Review A}, 96(6):062327, 2017.

\bibitem[LMR{\etalchar{+}}11]{LMRSS11}
Troy Lee, Rajat Mittal, Ben~W. Reichardt, Robert \v{S}palek, and Mario Szegedy.
\newblock Quantum query complexity of state conversion.
\newblock {\em 2011 IEEE 52nd Annual Symposium on Foundations of Computer
  Science}, Oct 2011.

\bibitem[Mat21]{Mathoverflow}
Probability of $\ell^1$-norms of vertices of the rotated hamming cube.
\newblock \url{https://mathoverflow.net/q/390129/115370}, 2021.
\newblock Accessed: April 20th, 2021.

\bibitem[MNKF18]{mitarai18}
Kosuke Mitarai, Makoto Negoro, Masahiro Kitagawa, and Keisuke Fujii.
\newblock Quantum circuit learning.
\newblock {\em Physical Review A}, 98(3):032309, 2018.

\bibitem[Mon15]{montanaro15}
Ashley Montanaro.
\newblock Quantum speedup of monte carlo methods.
\newblock {\em Proceedings of the Royal Society A: Mathematical, Physical and
  Engineering Sciences}, 471(2181):20150301, 2015.

\bibitem[NC00]{NC00}
Michael~A. Nielsen and Isaac~L. Chuang.
\newblock {\em Quantum Computation and Quantum Information}.
\newblock 2000.

\bibitem[NW99]{nayak99}
Ashwin Nayak and Felix Wu.
\newblock The quantum query complexity of approximating the median and related
  statistics.
\newblock In {\em Proceedings of the thirty-first annual ACM symposium on
  Theory of computing}, pages 384--393, 1999.

\bibitem[Rei11]{Reichardt10}
Ben~W. Reichardt.
\newblock Reflections for quantum query algorithms.
\newblock In {\em Proceedings of the Twenty-Second Annual ACM-SIAM Symposium on
  Discrete Algorithms}, SODA '11, page 560–569, USA, 2011. Society for
  Industrial and Applied Mathematics.

\bibitem[SB98]{sutton98}
Richard~S Sutton and Andrew~G Barto.
\newblock {\em Introduction to reinforcement learning}, volume 135.
\newblock MIT press Cambridge, 1998.

\bibitem[SJD21]{skolik21}
Andrea Skolik, Sofiene Jerbi, and Vedran Dunjko.
\newblock Quantum agents in the gym: a variational quantum algorithm for deep
  q-learning.
\newblock {\em arXiv preprint arXiv:2103.15084}, 2021.

\bibitem[SMSM00]{sutton00}
Richard~S Sutton, David~A McAllester, Satinder~P Singh, and Yishay Mansour.
\newblock Policy gradient methods for reinforcement learning with function
  approximation.
\newblock In {\em Advances in neural information processing systems}, pages
  1057--1063, 2000.

\bibitem[Sze04]{Szegedy04}
Mario Szegedy.
\newblock Quantum speed-up of markov chain based algorithms.
\newblock pages 32-- 41, 11 2004.

\bibitem[vA20]{vanApeldoorn20}
Joran van Apeldoorn.
\newblock {\em A quantum view on convex optimization}.
\newblock PhD thesis, February 2020.

\bibitem[vA21]{vanApeldoorn21}
Joran van Apeldoorn.
\newblock Quantum probability oracles \& multidimensional amplitude estimation.
\newblock {\em 16th Conference on the Theory of Quantum Computation,
  Communication and Cryptography}, 2021.

\bibitem[vH16]{vanHandel16}
Ramon van Handel.
\newblock Probability in high dimension.
\newblock \url{https://web.math.princeton.edu/~rvan/APC550.pdf}, December 2016.

\bibitem[WCNA09]{wocjan09}
Pawel Wocjan, Chen-Fu Chiang, Daniel Nagaj, and Anura Abeyesinghe.
\newblock Quantum algorithm for approximating partition functions.
\newblock {\em Physical Review A}, 80(2):022340, 2009.

\bibitem[WKS16]{wiebe16}
Nathan Wiebe, Ashish Kapoor, and Krysta~M Svore.
\newblock Quantum deep learning.
\newblock {\em Quantum Information \& Computation}, 16(7-8):541--587, 2016.

\bibitem[WW19]{wiebe19}
Nathan Wiebe and Leonard Wossnig.
\newblock Generative training of quantum boltzmann machines with hidden units.
\newblock {\em arXiv preprint arXiv:1905.09902}, 2019.

\bibitem[YM11]{YM11}
David Yonge-Mallo.
\newblock Adversary lower bounds in the hamiltonian oracle model, 2011.

\end{thebibliography}

\appendix

\section{Complexity of a classical MVMC algorithm\label{appdx:classical-complexity-MVMC}}

For a Monte Carlo process $\mathcal{A}$ generating a $d$-dimensional random variable $\bm{v} = (v_1, \ldots, v_d)$ bounded in $\ell_\infty$-norm $\norm{\bm{v}}_\infty \leq B$, consider the following algorithm:\\
\begin{enumerate}
	\item Execute $N = \left\lceil \frac{2B^2}{\varepsilon^2} \log\left(\frac{2d}{\delta}\right)\right\rceil$ runs of the Monte Carlo process $\mathcal{A}$ and store their outcomes $(v^{(i)}_1, \ldots, v^{(i)}_d)_{1\leq i \leq N}$.
	\item Compute the $d$ averages $\widehat{v_j} = \frac{1}{N}\sum_{i=1}^{N} v^{(i)}_j$ and use these as estimates.
\end{enumerate}
Now consider the probability of failure of this algorithm, i.e., that at least one of the estimates is more than $\varepsilon$ away from its expected value:
\begin{align*}
\P\left( \bigcup_{j \in [d]} \abs{\widehat{v_j} - \mathbb{E}_\mathcal{A}[v_j]} \geq \varepsilon \right) &\leq \sum_{j=1}^d \P\left(\abs{\widehat{v_j} - \mathbb{E}_\mathcal{A}[v_j]} \geq \varepsilon \right) & \#\textrm{ union bound}\\
&\leq d \times \max_{j \in [d]} \P\left(\abs{\widehat{v_j} - \mathbb{E}_\mathcal{A}[v_j]} \geq \varepsilon \right) &\\
&\leq 2d\exp\left( -\frac{2N^2\varepsilon^2}{4NB^2} \right)  & \#\textrm{ Hoeffding's bound and bound on $v_j$}\\
&\leq \delta. & \#\textrm{ definition of }N
\end{align*}
Hence, for arbitrary $\varepsilon$ and $\delta$, the $d$ expectations can be estimated to $\varepsilon$ error in $\ell_\infty$-norm with success probability $1-\delta$ using $N=\mathcal{O} \left( \frac{B^2}{\varepsilon^2} \log(\frac{d}{\delta})\right)$ runs of $\mathcal{A}$.

\section{Concentration bound of $\ell_1$-norm of vertices of the rotated Hamming cube}

In the lower bound of the query complexity to the reward oracle, we need a rather technical probability theory lemma, which we prove in the lemma below.

\begin{lemma}[Concentration bound of $\ell_1$-norm of vertices of the rotated Hamming cube]
	\label{lem:concentrationbound}
	Let $d \in \N$, and $O \in \R^{d \times d}$ be an orthogonal matrix, i.e., $O^TO = OO^T = I$. Let $\mathbf{x}$ be a random variable, taking values uniformly in the set $\{0,1\}^d$. Then,
	\[\P\left[\norm{O\mathbf{x}}_1 \leq \frac{d}{4}\right] \leq 2^{-\frac{\log(e)}{2}\left(\frac{1}{2\sqrt{2}} - \frac14\right)^2d}.\]
\end{lemma}

\begin{proof}
	First of all, we let $\varepsilon_1, \dots, \varepsilon_d$ be i.i.d.\ Rademacher random variables, i.e., $\P(\varepsilon_j = \pm 1) = 1/2$. Observe that
	\begin{align*}
	\E\left[\norm{O\mathbf{x}}_1\right] &= \sum_{j=1}^d \E\left[\left|\left(O\mathbf{x}\right)_j\right|\right] = \sum_{j=1}^d \E\left[\left|\sum_{k=1}^d O_{jk}x_k\right|\right] = \sum_{j=1}^d \E\left[\left|\sum_{k=1}^d O_{jk}\left(\frac12 + \frac12\varepsilon_k\right)\right|\right] \\
	&= \frac12\sum_{j=1}^d \E\left[\left|\sum_{k=1}^d O_{jk} + \sum_{k=1}^d O_{jk}\varepsilon_k\right|\right].
	\end{align*}
	We focus on each of the terms in the right-hand side individually. To that end, let $j \in [d]$. Since $\varepsilon \in \{-1,1\}^d$ has a distribution that is point-symmetric in the origin, we can rewrite the $j$th term as
	\begin{align*}
	\E\left[\left|\sum_{k=1}^d O_{jk} + \sum_{k=1}^d O_{jk}\varepsilon_k\right|\right] &= \frac12\E\left[\left|\sum_{k=1}^d O_{jk} + \sum_{k=1}^d O_{jk}\varepsilon_k\right|\right] + \frac12\E\left[\left|\sum_{k=1}^d O_{jk} - \sum_{k=1}^d O_{jk}\varepsilon_k\right|\right] \\
	&= \frac12\E\left[\left|\sum_{k=1}^d O_{jk} + \sum_{k=1}^d O_{jk}\varepsilon_k\right| + \left|\sum_{k=1}^d O_{jk} - \sum_{k=1}^d O_{jk}\varepsilon_k\right|\right] \\
	&\geq \frac12\E\left[\left|\sum_{k=1}^d O_{jk} + \sum_{k=1}^d O_{jk}\varepsilon_k - \sum_{k=1}^d O_{jk} + \sum_{k=1}^d O_{jk}\varepsilon_k\right|\right] = \E\left[\left|\sum_{k=1}^d O_{jk}\varepsilon_k\right|\right] \\
	&\geq \frac{1}{\sqrt{2}}\sqrt{\sum_{k=1}^d \left|O_{jk}\right|^2} = \frac{1}{\sqrt{2}},
	\end{align*}
	where we used Khintchine's inequality in the last step, and that for any $a,b \in \R$,
	\[|a+b| + |a-b| = |a+b| + |b-a| \geq |a+b + b-a| = 2|b|.\]
	Thus, we find that
	\[\E\left[\norm{O\mathbf{x}}_1\right] = \frac12\sum_{j=1}^d \E\left[\left|\sum_{k=1}^d O_{jk} + \sum_{k=1}^d O_{jk}\varepsilon_k\right|\right] \geq \frac{d}{2\sqrt{2}}.\]
	Now, it remains to prove a concentration bound for $\norm{O\mathbf{x}}_1$. To that end, we define a function $f : \R^d \to \R$ as
	\[f(\mathbf{x}) = \norm{O\mathbf{x}}_1,\]
	and we also define the vector $\sgn(O\mathbf{x}) \in \R^d$ as follows, for all $j \in [n]$,
	\[\sgn(O\mathbf{x})_j = \begin{cases}
	1, & \text{if } (O\mathbf{x})_j \geq 0, \\
	-1, & \text{otherwise.}
	\end{cases}\]
	Now, for all $\mathbf{z} \in \R^d$,
	\[\norm{O\mathbf{x} + \mathbf{z}}_1 \geq \norm{O\mathbf{x}}_1 + \mathbf{z}^T\sgn(O\mathbf{x}).\]
	Thus, for all $\mathbf{x},\mathbf{y} \in \{0,1\}^d$, with $\mathbf{z} = O(\mathbf{y} - \mathbf{x})$,
	\begin{align*}
	f(\mathbf{x}) - f(\mathbf{y}) &= \norm{O\mathbf{x}}_1 - \norm{O\mathbf{y}}_1 = \norm{O\mathbf{x}}_1 - \norm{O\mathbf{x} + O(\mathbf{y} - \mathbf{x})}_1 \\
	&\leq \norm{O\mathbf{x}}_1 - \norm{O\mathbf{x}}_1 - (\mathbf{y}-\mathbf{x})^TO^T\sgn(O\mathbf{x}) = (\mathbf{x} - \mathbf{y})^TO^T\sgn(O\mathbf{x}).
	\end{align*}
	Hence, if we define, for all $j \in [d]$,
	\[c_j(\mathbf{x}) = |(O^T\sgn(O\mathbf{x}))_j|,\]
	then
	\[f(\mathbf{x}) - f(\mathbf{y}) \leq \sum_{j=1}^d |x_j-y_j| \cdot c_j(\mathbf{x}) = \sum_{j=1}^d c_j(\mathbf{x})1_{x_j \neq y_j}.\]
	Thus, we can use Talagrand's inequality, in the form of Theorem 4.20 in \cite{vanHandel16}. We find that $f(\mathbf{x}) = \norm{O\mathbf{x}}_1$ is subgaussian, with a constant
	\[\max_{\mathbf{x} \in \{0,1\}^d}\sum_{j=1}^d c_j(\mathbf{x})^2 = \max_{\mathbf{x} \in \{0,1\}^d} \sum_{j=1}^d |(O^T\sgn(O\mathbf{x}))_j|^2 = \max_{\mathbf{x} \in \{0,1\}^d} \norm{O^T\sgn(O\mathbf{x})}_2^2 = \max_{\mathbf{x} \in \{0,1\}^d} \norm{\sgn(O\mathbf{x})}_2^2 = d.\]
	Thus, for all $t \in \R$,
	\[\E\left[e^{t\left(f(\mathbf{x}) - \E\left[f(\mathbf{x})\right]\right)}\right] = \E\left[e^{t\left(\norm{O\mathbf{x}}_1 - \E\left[\norm{O\mathbf{x}}_1\right]\right)}\right] \leq e^{\frac{t^2d}{2}},\]
	from which we deduce, with $t > 0$,
	\begin{align*}
	\P\left[\norm{O\mathbf{x}}_1 \leq \frac{d}{4}\right] &\leq \P\left[\norm{O\mathbf{x}}_1 - \E\left[\norm{O\mathbf{x}}_1\right] \leq \frac{d}{4} - \frac{d}{2\sqrt{2}}\right] = \P\left[e^{-t\left(\norm{O\mathbf{x}}_1 - \E\left[\norm{O\mathbf{x}}_1\right]\right)} \geq e^{-td\left(\frac14 - \frac{1}{2\sqrt{2}}\right)}\right] \\
	&\leq \frac{\E\left[e^{-t\left(\norm{O\mathbf{x}}_1 - \E\left[\norm{O\mathbf{x}}_1\right]\right)}\right]}{e^{-td\left(\frac14 - \frac{1}{2\sqrt{2}}\right)}} \leq e^{\frac{t^2d}{2} + td\left(\frac14 - \frac{1}{2\sqrt{2}}\right)} = e^{-\frac{d}{2}\left(\frac{1}{2\sqrt{2}} - \frac14\right)^2},
	\end{align*}
	where we used Markov's inequality, and in the last equality we plugged in $t = 1/(2\sqrt{2}) - 1/4 > 0$. Thus, we find
	\[\P\left[\norm{O\mathbf{x}}_1 \leq \frac{d}{4}\right] \leq 2^{-\frac{\log(e)}{2}\left(\frac{1}{2\sqrt{2}} - \frac14\right)^2d}.\]
	This completes the proof.
\end{proof}

\end{document}